\documentclass{amsart}
\usepackage{amssymb,amsmath,amsfonts,stmaryrd}

\usepackage[utf8]{inputenc}
\usepackage[final,nopatch=footnote]{microtype}

\usepackage{leftidx}
\usepackage{comment}
\usepackage{booktabs}
\usepackage{yfonts}
\makeatletter

\usepackage[usenames, dvipsnames]{xcolor}

\usepackage[backref=page, colorlinks=true]{hyperref}
\hypersetup{
 linkcolor={red!50!black},
 citecolor={green!50!black},
 urlcolor={blue!80!black}
}

\usepackage{annotate-equations}

\usepackage[margin=1.3in]{geometry}
\usepackage[noabbrev, capitalize]{cleveref} 
\numberwithin{equation}{section}
\usepackage{amsthm}
\newtheorem*{thm*}{Theorem}
\newtheorem{conj}[equation]{Conjecture}
\newtheorem*{conj*}{Conjecture}
\newtheorem*{bosconj}{\cref{upstream_conjecture}}
\newtheorem*{anomvalues}{\cref{prop:actual_anomaly}}
\newtheorem{thm}[equation]{Theorem}
\newtheorem{lem}[equation]{Lemma}
\newtheorem{prop}[equation]{Proposition}
\newtheorem{cor}[equation]{Corollary}
\newtheorem{ansatz}[equation]{Ansatz}
\theoremstyle{definition}
\newtheorem{defn}[equation]{Definition}
\newtheorem{exm}[equation]{Example}
\newtheorem*{hdef}{Heuristic Definition}
\theoremstyle{remark}
\newtheorem{rem}[equation]{Remark}

\crefname{thm}{Theorem}{Theorems}
\crefname{prop}{Proposition}{Propositions}
\crefname{lem}{Lemma}{Lemmas}

\usepackage{mathtools}
\DeclarePairedDelimiter{\set}{\{}{\}}
\DeclarePairedDelimiter{\paren}{(}{)}
\DeclarePairedDelimiter{\angl}{\langle}{\rangle}
\usepackage{spectralsequences}
\usepackage{tikz-cd}

\usepackage{booktabs}

\usetikzlibrary{shapes.geometric,calc}
\usetikzlibrary{shapes,arrows,chains}
\usetikzlibrary{decorations.markings}
\usetikzlibrary{decorations.pathmorphing}
\usetikzlibrary{shapes.multipart}
\tikzset{snake it/.style={decorate, decoration=snake}}
\tikzstyle{GraphNode}=[circle, draw=black, fill=black, inner sep=2pt, minimum size=5pt]
\tikzstyle{GraphEdge}=[black]
\usepackage{subcaption}
\usepackage{setspace}
\usetikzlibrary{calc}
\usetikzlibrary{tikzmark}

\usepackage[mathscr]{eucal}
\makeatletter
\def\instring#1#2{TT\fi\begingroup
  \edef\x{\endgroup\noexpand\in@{#1}{#2}}\x\ifin@}
\def\isuppercase#1{%
  \instring{#1}{ABCDEFGHIJKLMNOPQRSTUVWXYZ}%
}%
\newcommand{\C@lIfUpper}[1]{
 \if\isuppercase{#1}\mathscr{#1}%
 \else #1%
 \fi
}
\newcommand{\cat}[1]{\mathit{\@tfor\next:=#1\do{\C@lIfUpper{\next}}}}
\makeatother

\usepackage{arydshln}

\usepackage{xparse}
\newcommand{\Z}{\mathbb Z}

\newcommand{\C}{\mathbb C}
\newcommand{\U}{\mathrm U}
\newcommand{\Spin}{\mathrm{Spin}}
\newcommand{\EPin}{\mathrm{EPin}}

\newcommand{\KO}{\mathit{KO}}

\newcommand{\MTSpin}{\mathit{MTSpin}}

\newcommand{\ko}{\mathit{ko}}
\newcommand{\TODO}{\textcolor{red}{TODO}}
\newcommand{\Sq}{\mathrm{Sq}}
\newcommand{\abs}[1]{\lvert #1 \rvert}
\newcommand{\pt}{\mathrm{pt}}

\newcommand{\SU}{\mathrm{SU}}

\newcommand{\SO}{\mathrm{SO}}

\renewcommand{\O}{\mathrm O}
\newcommand{\id}{\mathrm{id}}
\newcommand{\RP}{\mathbb{RP}}
\newcommand{\bl}{\text{--}}
\newcommand{\R}{\mathbb R}

\newcommand{\Pin}{\mathrm{Pin}}

\newcommand{\CP}{\mathbb{CP}}

\newcommand{\fC}{\cat C}
\newcommand{\sVect}{\cat{sVect}}
\newcommand{\Vect}{\cat{Vect}}
\newcommand{\sAlg}{\cat{sAlg}}
\newcommand{\Bord}{\cat{Bord}}
\newcommand{\Fun}{\cat{Fun}}

\newcommand{\End}{\mathrm{End}}

\usepackage{xspace}
\newcommand{\pinp}{pin\textsuperscript{$+$}\xspace}
\newcommand{\pinm}{pin\textsuperscript{$-$}\xspace}
\newcommand{\spinc}{spin\textsuperscript{$c$}\xspace}
\newcommand{\pinc}{pin\textsuperscript{$c$}\xspace}

\newcommand{\Q}{\mathbb Q}
\newcommand{\Snb}{S_{\mathit{nb}}}

\newcommand{\term}{\emph}

\newcommand{\Hom}{\mathrm{Hom}}

\newcommand\MAILTO[1]{\href{mailto:#1}{\nolinkurl{#1}}}

\DeclareDocumentCommand{\shortexact}{s O{} O{} mmmm}{
\IfBooleanTF{#1}{ 
\begin{tikzcd}[ampersand replacement=\&]
	{1} \& {#4} \& {#5} \& {#6} \& {1#7}
	\arrow[from=1-1, to=1-2]
	\arrow["#2", from=1-2, to=1-3]
	\arrow["#3", from=1-3, to=1-4]
	\arrow[from=1-4, to=1-5]
\end{tikzcd}
}{ 
\begin{tikzcd}[ampersand replacement=\&]
	{0} \& {#4} \& {#5} \& {#6} \& {0#7}
	\arrow[from=1-1, to=1-2]
	\arrow["#2", from=1-2, to=1-3]
	\arrow["#3", from=1-3, to=1-4]
	\arrow[from=1-4, to=1-5]
\end{tikzcd}
}}

\let\oldtocsection=\tocsection
\let\oldtocsubsection=\tocsubsection
\let\oldtocsubsubsection=\tocsubsubsection
\renewcommand{\tocsection}[2]{\hspace{0em}\bfseries\oldtocsection{#1}{#2}}
\renewcommand{\tocsubsection}[2]{\hspace{1em}\oldtocsubsection{#1}{#2}}
\renewcommand{\tocsubsubsection}[2]{\hspace{2em}\oldtocsubsubsection{#1}{#2}}

\usepackage[normalem]{ulem}

\newcommand{\mc}{\mathcal}
\newcommand{\Diff}{\mathrm{Diff}}
\def\U{\mathrm{U}(1)}

\definecolor{sanddune}{rgb}{0.59, 0.44, 0.09}
\definecolor{darkblue}{RGB}{0,0,102}
\definecolor{darkred}{rgb}{0.5,0.,0.}
\definecolor{BlueViolet}{RGB}{138,43,226}
\definecolor{SkyBlue}{RGB}{30,144,255}
\definecolor{DarkGreen}{RGB}{0,100,0}

\title{Bosonization and Anomaly Indicators of (2+1)-D Fermionic Topological Orders}

\author{Arun Debray}
\address{Department of Mathematics, University of Kentucky,
719 Patterson Office Tower, Lexington, KY 40506-0027}
\email{\href{mailto:a.debray@uky.edu}{a.debray@uky.edu}}

\author{Weicheng Ye}
\address{Department of Physics and Astronomy, and Stewart Blusson Quantum Matter Institute,
University of British Columbia, Vancouver, BC, Canada V6T 1Z1}
\address{Department of Physics, The Chinese University of Hong Kong, Shatin, New Territories, Hong Kong}
\email{\href{mailto:victoryeofphysics@gmail.com}{victoryeofphysics@gmail.com}}

\author{Matthew Yu}
\address{Mathematical Institute, Andrew Wiles Building, Woodstock Road, Oxford, OX2 6GG}
\email{\href{mailto:yumatthew70@gmail.com}{yumatthew70@gmail.com}}

\begin{document}

\begin{abstract}
We provide a mathematical proposal for the anomaly indicators of symmetries of (2+1)-d fermionic topological orders, and work out the consequences of our proposal in several nontrivial examples. Our proposal is an invariant of a super modular tensor category with a fermionic group action, which gives a (3+1)-d topological field theory (TFT) that we conjecture to be invertible; the anomaly indicators are partition functions of this TFT on $4$-manifolds generating the corresponding twisted spin bordism group. Our construction relies on a bosonization construction due to Gaiotto-Kapustin and Tata-Kobayashi-Bulmash-Barkeshli, together with a ``bosonization conjecture'' which we explain in detail. In the second half of the paper, we discuss several examples of our invariants relevant to condensed-matter physics. The most important example we consider is $\Z/4^T\times \Z/2^f$ time-reversal symmetry with symmetry algebra $\mathcal T^2 = (-1)^FC$, which many fermionic topological orders enjoy, including the $\U_5$ spin Chern-Simons theory. Using newly developed tools involving the Smith long exact sequence, we calculate the cobordism group that classifies its anomaly, present the generating manifold, and calculate the partition function on the generating manifold which serves as our anomaly indicator. Our approach allows us to
reproduce anomaly indicators known in the literature with simpler proofs, including $\Z/4^{Tf}$ time-reversal symmetry with symmetry algebra $\mathcal T^2 = (-1)^F$, and other symmetry groups in the 10-fold way involving Lie group symmetries.
\end{abstract}

\date{\today}

\thanks{It is a pleasure to thank Jaume Gomis for proposing this problem, Liujun Zou for collaboration in a related project and Yu-An Chen, Diego Delmastro, Dan Freed, and Juven Wang for helpful conversations. We thank Justin Kulp for providing us with tikz examples. The work of the author Matthew Yu is supported by the by the  EPSRC Open Fellowship EP/X01276X/1.  The authors contribute equally and are listed in alphabetical order. Part of this research was conducted at the Perimeter Institute.}

\maketitle

\tableofcontents

\section{Introduction}
A \textit{topological order} is a unique state of matter that emerges in certain gapped quantum systems~\cite{wen2004quantum}. Unlike traditional phases of matter, such as solids, liquids, and gases, topological order is not defined by local order parameters or spontaneous symmetry breaking, but rather by the presence of long-range quantum entanglement. One of the most notable examples of these exotic phases is the
fractional quantum Hall effect~\cite{stormer1999fractional}. This is a \textit{fermionic} topological order in (2+1)-dimensions and exists for two-dimensional electron systems subjected to strong magnetic fields. The electrons exhibit a collective behavior giving rise to particles known as \textit{anyons}, quasiparticle excitations with nontrivial statistics that may be neither bosonic nor
fermionic. For our set up in (2+1)-d, 
 a fermionic topological order is mathematically axiomatized to be a \textit{super modular tensor category} (MTC)~\cite{bruillard2017fermionic}, the definition of which we will further delve into in \S\ref{sec:prelimnary}.

There is a rich interplay between topological orders and symmetry. Most notably, anyons may transform in a projective representation under the symmetry action, and we sometimes say that anyons carry ``fractional'' quantum numbers. This is known as the phenomenon of \textit{symmetry fractionalization}. Symmetry action on a topological order gives a categorical group action on the corresponding tensor category \cite{barkeshli2014,Galindo2017,Bulmash2021}. Moreover, the symmetry action on the topological order can have a 't Hooft anomaly~\cite{Hooft1980}. In the condensed matter setting, this suggests that the symmetry action cannot
be realized as an on-site symmetry. The topological order therefore has to be realized on
the boundary of an \textit{invertible field theory} or a \textit{symmetry-protected topological phase} (SPT) in one dimension higher. Classifying the invertible phases would then give a classification of the anomalies. We will perform this classification by computing the relevant cobordism groups \cite{FH21InvertibleFT,Wang2014} for symmetries associated to fermionic theories.

One can use 't Hooft anomalies as a means of constraining the IR phase that a UV theory flows to, but doing so requires computing the specific value of the anomaly and not just the group that classifies it. A very common scenario in the high energy literature is when the UV is described by a quantum field theory where fermions are weakly coupled to some gauge fields \cite{Gomis:2017ixy,Choi:2018tuh,Komargodski:2017keh,DGY},
and in the IR the theory flows to a strongly interacting field theory, in particular a topological order. Another scenario that appears in the condensed matter literature involves the UV lattice system having some Lieb-Schultz-Mattis-type anomaly \cite{Cheng2016,Else2020,Ye2021LSM}, and anomaly matching can help us identify which IR theories can emerge from the UV lattice system \cite{Ye2021LSM,Zou2021,Ye2023}. However, computing the anomaly in the IR is a much more involved process than in the UV, see e.g. \cite{Tachikawa:2016cha} where the authors compute the $\Z/16$ valued anomaly for time-reversal in (2+1)-d topological theories by using a crosscap background. It is therefore desirable to understand how to compute anomalies for fermionic topological orders in a systematic way.

There has been significant progress in calculating the anomaly of (2+1)-d bosonic topological orders with symmetry, including \cite{barkeshli2014,Wang:2016qkb,Lapa2019,Wang2015b,walker2019,Bulmash2020,Kobayashi2021,Ye:2022bkx}. In particular, for a time reversal symmetry $\Z/2^T$ that may permute anyons, \cite{Wang:2016qkb} proposed a set of anomaly indicators to detect the anomaly of any bosonic topological order, and \cite{walker2019} derived the formula by calculating the partition function of a certain topological field theory (TFT) on the generating manifold of the corresponding unoriented bordism group $\Omega_4^\O$. In order to extend the calculation to a general finite group symmetry, which may contain anti-unitary elements and/or permutes anyons, \cite{Bulmash2020} gave a state-sum construction of the anomaly theory. Following these ideas, \cite{Ye:2022bkx} revisited the construction in the language of extended TFTs and wrote down a general recipe to obtain the explicit anomaly indicators for any bosonic topological order with general symmetry groups. This recipe is also generalized to include Lie group symmetries.

Many of these strategies have been applied to tackle fermionic topological orders, including \cite{Wang:2016qkb,Lapa2019,Ning2021,Tata2021,Bulmash:2021ryq,Aasen2021,Kobayashi2021}. In particular, \cite{Wang:2016qkb} also proposed an anomaly indicator to detect the anomaly of fermionic topological order with the $\Z/4^{Tf}$ symmetry. Later, \cite{Tata2021} generalized the construction and calculation in \cite{Bulmash2020} to the fermionic setting 
and derived the conjectured anomaly indicator. Lie group symmetries were also studied in the fermionic setting in \cite{Lapa2019,Ning2021,Kobayashi2021}. 

The purpose of this paper is to classify and compute the anomaly for symmetries of (2+1)-d fermionic topological orders. We provide a general mathematical approach and also study several example symmetries. In the body of our article, we focus on abelian time-reversal symmetry, with symmetry algebra $\mc{T}^2 = (-1)^F$ for the $\Z/4^{Tf}$ symmetry and $\mc{T}^2 = (-1)^FC$ for the $\Z/4^T\times \Z/2^f$ symmetry, with $\mc{T}$ the time-reversal generator, $C$ charge conjugation and $(-1)^F$ fermion parity.  
The tools that we employ can be applied to {general symmetry groups which may be
discrete or continuous, abelian or non-abelian, contain anti-unitary elements and/or permute anyons.\footnote{See \cite{Geiko2022} for the identification of anti-unitary time-reversal symmetries in Chern-Simon theory.} We follow the general procedure developed in \cite{walker2019, Bulmash2020, Ye:2022bkx} for bosonic topological orders and generalize it to the fermionic context in a spirit similar to \cite{Tata2021, Kobayashi2021}. We make a conjecture that the computation of the partition function for a spin TFT can be done via a bosonization procedure, which relates the spin theory to a bosonic one.
This conjecture is vital in establishing if a particular spin TFT associated to an anomaly is invertible, and if its partition function is a cobordism invariant. Moreover, our approach utilizes a handle decomposition \cite{gompf1994} of a manifold following \cite{Ye:2022bkx}. 
Compared to \cite{Tata2021,Kobayashi2021} which uses a cellulation of a manifold, our calculation is much simpler and will produce closed-form expressions for partition functions and anomaly indicators.

The general procedure can be summarized by a sequence of steps. We first compute the group that classifies the anomaly of a fermionic symmetry and identify the generating 4-manifold for the dual bordism group. We then evaluate the partition function on the generating manifold of a certain spin TFT, which should be thought of as the anomaly theory whose boundary hosts the given (symmetry-enriched) topological order. The partition function is written in terms of the anyon content of the topological order as well as the specific symmetry action on the set of anyons.
The result is the anomaly indicator for the fermionic symmetry. We showcase these steps by deriving the anomaly indicators known in the literature for 
fermionic symmetries in the 10-fold way classification \cite{Wang2014,FH21InvertibleFT}. 
As a brand new example with physical importance, we then focus on the $\Z/4^T\times \Z/2^f$ symmetry, which, as discussed in \cite{Delmastro:2019vnj}, is a symmetry of abelian spin Chern-Simons theories such as $\U_k$ with $k=5,13,\dots$ describing the $\nu=1/k$ fractional quantum Hall effect. We use this example to showcase all the technical tools and the power of our general method. As a bonus, we use this example to demonstrate how to use the \textit{Smith homomorphism} as a powerful method to perform the bordism calculations, in particular for resolving hard extension problems and constructing the generating manifold.

\subsection{Summary of Main Results}
We now give an overview of the organization of the paper, and present a summary of the main results.

\begin{enumerate}

\item In \S\ref{sec:prelimnary}, we give a brief summary of the background material and notation that will be used throughout the paper, including our definitions of fermionic symmetry, anomaly, and (2+1)-d fermionic topological order.

\item
 In \S\ref{sec:spinTFTsConj}, we state the main conjecture serving as the backbone that makes our techniques implementable. The obvious steps to take is to construct a spin version of
 Crane-Yetter theory corresponding to the fermionic symmetry. However, this is a difficult open problem so we take a different approach by bosonization. We frame the problem of identifying the spin TFT and writing down its partition function into a bosonic problem, which serves as an easier method for making concrete computations of partition functions.
\begin{hdef}
Given a fermionic $G$-symmetry acting on a super MTC $\cat C$, we define a 4d TFT $\alpha$ by the following recipe:
\begin{enumerate}
    \item Let $Z_b$ be the bosonic shadow of the 3d spin TFT defined by $\cat C$. The $G$-action on $\cat C$ induces a $G$-symmetry of $Z_b$.
    \item Let $\alpha_b$ be the 4d anomaly field theory of the $G$-symmetry of $Z_b$, as constructed in~\cite{walker2019,Bulmash2021,Ye:2022bkx}.
    \item Let $\alpha$ be the fermionization of $\alpha_b$.
\end{enumerate}
\term{Anomaly indicators} refer to the value of $\alpha$ on closed $4$-manifolds, especially generators of bordism groups of interest.
\end{hdef}
Most of \S\ref{sec:spinTFTsConj} is devoted to building up the ingredients of the precise version of this definition; said precise definition is given in \S\ref{subsec:recipe}, with formulas for the partition function of $\alpha$ in~\eqref{eq:summing_bosonic_shadow} and~\eqref{eq:main_computation}. Ideas similar to the above heuristic definition appear in work of Tata-Kobayashi-Bulmash-Barkeshli~\cite{Tata2021}.

A priori, the 4d TFT $\alpha$ defined above has no relation to the anomaly of the $G$-symmetry on $\cat C$, but computations in examples suggest that the two are the same. A significant goal of this paper is to provide a conjectural explanation of this phenomenon.

To explain our conjecture, we need a few pieces of notation. Gaiotto-Kapustin~\cite{Gaiotto2016} implement the bosonic shadow construction and its inverse as a kernel transform with an anomalous TFT called $z_c$: one tensors with $z_c$, then sums over spin structures or (higher) $\Z/2$ gauge fields. Tata-Kobayashi-Bulmash-Barkeshli~\cite{Tata2021} show how to incorporate the $G$-action into a more general definition of $z_c$. Because $z_c$ is defined on manifolds with data of a (possibly twisted) spin structure \emph{and} a $\Z/2$ higher gauge field, we re-express it in \S\ref{bos_conj_sec} as a defect between two TFTs, specifically two Dijkgraaf-Witten type TFTs $F_{\Spin}$ and $F_{B}$ defined by summing over spin structures, resp.\ $\Z/2$ higher gauge fields.
\begin{bosconj}[Bosonization Conjecture]
Given a fermionic $G$-symmetry acting on a super MTC $\fC$ with anomaly $\widetilde\alpha$, let $F_{\widetilde\alpha}$ denote the Dijkgraaf-Witten type theory obtained by summing $\widetilde\alpha$ over (twisted) spin structures, and let $F_{\widetilde\beta}$ be the same construction applied to the bosonic shadow of $\alpha$. Then, as theories of manifolds with (twisted) spin structures, $z_c$ extends from an $(F_\Spin, F_B)$-defect to an $(F_{\widetilde\alpha}, F_{\widetilde\beta})$-defect.
\end{bosconj}
The restriction to (twisted) spin manifolds is to work around the appearance of a different anomalies; see \S\ref{subsubsec:higherd}. We explain in \S\ref{bos_conj_sec} how the bosonization conjecture implies the following key result.
\begin{thm*}
Assuming \cref{upstream_conjecture} and \cref{conj:inv}, the theory $\alpha$ we define in \S\ref{subsec:recipe} is the anomaly of the $G$-action on $\fC$.
\end{thm*}
In particular, $\alpha$ is invertible and the anomaly indicators we calculate are Reinhardt bordism invariants; we show in \cref{usu_bordism} that (again assuming \cref{upstream_conjecture} and \cref{conj:inv}) they are bordism invariants in the usual sense. In other words:
\begin{conj*}
    The bosonic shadows $Z_b$ of the anomaly (3+1)-d fermionic theory assemble into a partition function $\mathcal{Z}^f$, which is a cobordism invariant in the (twisted) spin cobordism group that classifies anomalies of the (2+1)-d fermionic theory.
\end{conj*}




\item We then turn to specific examples of fermionic symmetries. We start with reproducing the known anomaly indicators for the $\Z/4^{Tf}$ symmetry with symmetry algebra $\mc{T}^2 = (-1)^F$ in \S\ref{sec:Z/2^Timerev}. Next we undertake the task of understanding the $\Z/4^T\times \Z/2^f$ symmetry with symmetry algebra $\mc{T}^2 = (-1)^FC$.

In \cref{subsec:smith_cal}, we calculate the relevant bordism group that classifies the anomalies of $\Z/4^T\times \Z/2^f$ symmetry in $(2+1)$-d and its generating manifold. More broadly, we generalize the case of $\Z/4^T$ to $\Z/k^T$ with $4\mid k$ and denote the tangential structure for this symmetry by $\EPin[k]$. When $k=4$ it reduces to the EPin structure which already appears in the literature. The Atiyah-Hirzebruch spectral sequence and the Adams spectral sequence can determine that $\Omega_4^{\EPin[k]}$ is of order 4 \cite{BG97,WWZ20}, but there is a highly nontrivial extension problem such that the specific isomorphism type of this group was an open question. Utilizing the recently developed techniques involving the Smith homomorphism \cite{DDKLPTT,DDKLPTT2}, we resolve the extension problem. The bordism group $\Omega_4^{\EPin[k]}$  and their generating manifolds are summarized as follows:
 \begin{thm*}\hfill
\label{epink_thm}
\begin{enumerate}
    \item If $k \equiv 4\bmod 8$, $\Omega_4^{\EPin[k]}\cong\Z/4$. Let $\mathcal M$ denote the manifold we construct in \cref{thm:generator}, which is the total space of a Klein bottle bundle over $S^2$, then $\mathcal M$ generates $\Omega_4^{\EPin[k]}$.
    \item If $k\equiv 0\bmod 8$, $\Omega_4^{\EPin[k]}\cong\Z/2\oplus\Z/2$, with a basis given by the bordism classes of $\mathcal M$ and the K3 surface.
\end{enumerate}
\end{thm*}
This is a combination of \cref{epink_thm_middle,thm:generator}.

For $k = 4$, this bordism group has been studied in the literature, but we are the first to compute it. Botvinnik-Gilkey~\cite{BG97} and Barrera-Yanez~\cite[Theorem 3.1]{BY99} study this and other epin$[k]$ bordism groups using algebraic and analytic methods, respectively, but do not determine the isomorphism type of $\Omega_4^{\EPin[k]}$; Wan-Wang-Zheng~\cite[\S B.1]{WWZ20} erroneously reported that $\Omega_4^{\EPin[4]}$ is isomorphic to $\Z/2\oplus\Z/2$; and Córdova-Hsin-Zhang~\cite{Cordova:2023bja} show that a closely related group is isomorphic to $\Z/4$ but do not study epin$[k]$ bordism.

 The general formula for the anomaly indicator is then obtained by calculating the partition function on the generator $\mathcal{M}$ of $\Omega_4^\EPin$, which is given in \cref{prop:anomaly_indicator}. We use it to compute the anomalies for abelian Chern-Simons theories, as well as $\U_2 \times \U_{-1}$ and $\SO(3)_3$, and we find:
\begin{anomvalues}
    The {fermionic topological order} $\U_k$ (for {$k=5,13,\dots,$ as} given in \cite{Delmastro:2019vnj}), $\U_2 \times \U_{-1}$ and $\SO(3)_3$ realize a time-reversal symmetry,  with algebra $\mathcal{T}^2=(-1)^FC$. {The anomaly for $\U_k$, $\U_2 \times \U_{-1}$ and $\SO(3)_3$} evaluates to $\nu=0, 2, 3\in \Z/4$, respectively.
\end{anomvalues}

\item In \S\ref{app:U1_time}, we extend our calculation to Lie group symmetries, including seven out of the ten symmetries in the 10-fold way classification of fermionic symmetries.\footnote{Combining the result in \S\ref{app:U1_time} and the result in \S\ref{sec:prelimnary}, we have discussed nine out of the ten symmetries in the 10-fold way. The missing symmetry type is class D with $G_f = \Z/2^f$, whose associated tangential structure is simply spin with no twist. We discuss this in \S\ref{sec:discussion}, point~\ref{item:K3}.}
We write down the generating manifolds of the corresponding bordism groups and calculate the partition functions on the generating manifolds following our general recipe in \S\ref{subsec:recipe}. The calculation correctly reproduces the anomaly indicators for these symmetries known in the literature \cite{Lapa2019,Ning2021}. The results and anomaly indicators are summarized in \cref{prop:anomaly_indicator_U(1)case0,prop:anomaly_indicator_SO(3)case0,prop:classA_anomaly,prop:classC_anomaly}. 

\item In Appendix~\ref{app:data}, we collect data of the fermionic topological orders and symmetry actions we explicitly discuss, including $\U_5$, $\U_2\times \U_{-1}$ and $\SO(3)$.
\item In Appendix~\ref{app:cascade}, we provide an additional test of our conjectures and results by comparing our calculation in \cref{prop:actual_anomaly} with a different method of calculating anomalies: the ``anomaly cascade conjecture'' of Bulmash-Barkeshli~\cite{Bulmash:2021ryq} (here appearing as \cref{BBconj}) identifying the manifestation of anomaly with the filtration data of $\mho^4_\xi$ coming from the Atiyah-Hirzebruch spectral sequence. In \cref{U15_BB,semion_fermion_BB,SO33_BB}, we calculate the anomalies of the $\Z/4^T\times \Z/2^f$ symmetry for $\U_5$, $\U_2\times\U_{-1}$, and $\SO(3)_3$ assuming this conjecture; our results are consistent with \cref{prop:actual_anomaly}, providing further support for our calculations.

Bulmash-Barkeshli relate the data of an anomaly in the layers of the Atiyah-Hirzebruch filtration to tensor-category-theoretic data, meaning that to use their conjecture we must perform manipulations on the tensor categories we study. To do so, we make use of a technique called \term{zesting} due to Bruillard-Galindo-Hagge-Ng-Plavnik-Rowell-Wang~\cite{bruillard2017fermionic} (see also~\cite{delaney2021braided}) to produce different modular extensions of super MTCs.
%
\end{enumerate}

\section{\texorpdfstring{
Preliminaries: Fermionic Symmetry, Anomaly and $(2+1)$-D Fermionic Topological Orders}{}}\label{sec:prelimnary}

In this section, we collect all the background knowledge that will be used throughout the paper. This includes the basic setup of a fermionic symmetry and its anomaly in \S\ref{subsection:fermionicsym}, as well as necessary information about (2+1)-d fermionic topological order with a symmetry action on it in \S\ref{subsection:FermionicTO}. Readers familiar with these topics can skip the exposition in this section and use it as a reference.

\subsection{Fermionic Symmetry and Anomaly}\label{subsection:fermionicsym}

In this subsection, we give a more careful definition of what we mean by a fermionic symmetry, as well as the relevant bordism group involved in the anomaly classification for the symmetry.

In order to specify a symmetry in a fermionic system, we need to specify the symmetry group $G_f$, identify the $\Z/2^f$ subgroup corresponding to the distinguished fermion parity symmetry, and all the anti-unitary elements in the symmetry group. It is helpful to summarize all the information of a fermionic symmetry into the following triple of data.

\begin{defn}[{Benson~\cite[\S 7]{Ben88}}]
\label{def:fermionic_symmetry}
A \textit{fermionic symmetry} is the data of a group $G_b$, a homomorphism $s\colon G_b\to\set{\pm 1}$, and a central extension
\begin{equation}
    1\rightarrow \Z/2\rightarrow G_f\rightarrow G_b\rightarrow 1\,.
\end{equation}
We refer to $G_b$ as the \term{bosonic symmetry group} of the fermionic symmetry.
\end{defn}
As is conventional, we will name a fermionic symmetry in terms of its full fermionic symmetry group $G_f$.

The homomorphism $s$ defines a class in $H^1(BG;\Z/2)$ that, with a slight abuse, we also call $s$, and the extension by $\Z/2$ defines a class $\omega\in H^2(BG;\Z/2)$. The data of $(G_b,s,\omega)$ characterizes a fermionic symmetry up to isomorphism of $G_b$, $s$, and the extension. Hence, we will also identify a fermionic symmetry by the triple $(G_b, s, \omega)$.

In particular, given a fermionic symmetry, $s$ tells us whether elements of $G_b$ act unitarily or antiunitarily:
\begin{align}
  \label{eq:qdef}
s(\bf g) = \begin{cases}
0 & \text{if $\bf g$ acts unitarily} \\
1 & \text{if $\bf g$ acts anti-unitarily.}
\end{cases}
\end{align}
$s$ also induces a map $Bs\colon BG_b\rightarrow B\Z/2$, and the pullback of the tautological bundle on $B\Z/2$ across $s$ gives a 1-dimensional line bundle on $BG_b$. This line bundle is denoted by $\sigma$ in this paper and will play a crucial role. In particular, $w_1(\sigma) = s$. 


%

The symmetry groups that we consider in the main text are listed as follows, given in terms of the triple data $(G_b, s, \omega)$:

\begin{enumerate}\label{enumerate:definition_of_groups}
    \item $\Z/4^{Tf}$. Here $G_b=\Z/2$, $s$ is the nontrivial element in $H^1(B\Z/2; \Z/2)\cong\Z/2$, and $\omega$ is the nontrivial element in $H^2(B\Z/2; \Z/2)\cong\Z/2$. 
    \item $\Z/4^T\times \Z/2^f$. Here $G_b=\Z/4$, $s$ is the nontrivial element in $H^1(B\Z/4; \Z/2)\cong\Z/2$, and $\omega$ is the \textit{trivial} element in $H^2(B\Z/4; \Z/2)\cong\Z/2$.
    \item Ten-fold way symmetries. The data for $(G_b, s, \omega)$ are collected in \cref{tab:definition_of_groups}.
\end{enumerate}
In this paper, we use $^T$ to denote that some elements in the symmetry group are anti-unitary, and hence the $s$ for these groups are nontrivial. We also use $^f$ to denote that some element in the symmetry group corresponds to fermion parity $(-1)^F$. 

We will sometimes need maps between different fermionic symmetries which we define as follows.
\begin{defn}\label{def:map_symmetry}
    A map $\phi$ between two different fermionic symmetries $(G_1, s_1, \omega_1)$ and $(G_2, s_2, \omega_2)$ is a map $\phi\colon G_1\rightarrow G_2$ such that $\phi^*(s_2) = s_1$ and $\phi^*(\omega_2) = \omega_1$.
\end{defn}
This definition guarantees that the induced maps between, e.g., the relevant bordism groups and the relevant (twisted) cohomology are all well-defined. For simplicity, we will usually only state the map $\phi \colon G_1\rightarrow G_2$ and let the reader check that the two requirements are satisfied.
\begin{rem}
There are a few other ways to package the data of a fermionic symmetry.
\begin{itemize}
    \item A normal $1$-type~\cite[\S 2]{Kre99}.
    \item Fermionic groups in the sense of Stehouwer~\cite[Definition 1]{Ste21}.
    \item A twist of spin bordism provided by the data of a map $BG_b\to B\O/B\Spin$ (see~\cite{HS20} and~\cite[\S 1.2.3]{Debray:2023tdd}).
\end{itemize}
All of these are equivalent to our \cref{def:fermionic_symmetry}, in that they provide $G_b$, $s$, and $\omega$ in the sense above, and are equivalent to such data.
\end{rem}
In the situations we consider in this paper, the anomaly of a fermionic symmetry acting on an $n$-dimensional field theory $Z$ is an $(n+1)$-dimensional \emph{invertible} field theory $\alpha$:\footnote{Anomalies are a very general subject in physics; we are not claiming that everything called ``anomaly'' can be described in this way.} there is some other $(n+1)$-dimensional theory $\alpha'$ and an isomorphism $\alpha\otimes\alpha'\overset\simeq\to \textbf{1}$, where $\textbf{1}$ is the trivial theory whose value on objects is the monoidal unit, and whose value on morphisms is the identity. The notion of an invertible field theory is due to Freed-Moore~\cite[Definition 5.7]{FM06}. The connection to anomalies is due to Freed-Teleman~\cite{FT14, Fre14}; see also~\cite{Wit00, FHT10} and see Freed~\cite{Fre23} for a nice overview.

To talk about examples of invertible field theories, we must specify what kinds of manifolds we place them on.
\begin{defn}\label{tangerine}
A \term{tangential structure} is a map $\xi\colon B\to B\O$, which we, without loss of generality, take to be a fibration. Given $\xi$, a \term{$\xi$-structure} on a vector bundle $V\to X$ is a lift of the classifying map $f_V\colon X\to B\O$ of $V$ to $B$, i.e.\ it is a map $\widetilde f_V\colon X\to B$ such that $f_V \simeq \xi\circ\widetilde f_V$. A $\xi$-structure on a manifold $M$ means a $\xi$-structure on $TM$.
\end{defn}
Two $\xi$-structures on a vector bundle $V\to X$ are equivalent if they belong to the same connected component of the space of $\xi$-structures on $V$. Whenever we count $\xi$-structures, we are referring to equivalence classes of $\xi$-structures.

Given a family of groups $H(n)$ with maps $H(n)\to H(n+1)$ and homomorphisms $\rho(n)\colon H(n)\to\O(n)$ commuting with the maps $H(n)\to H(n+1)$ and $\O(n)\to\O(n+1)$, one can take the classifying space of the colimit to obtain a $\xi$-structure $B\rho\colon BH\to B\O$. This is a good source of examples of $\xi$-structures: for example, when $H(n) = \SO(n)$, this notion of a $\xi$-structure is equivalent to an orientation, and when $H(n) = \Spin(n)$, it is equivalent to a spin structure. Equivalence of $\xi$-structures coincides with equality of orientations, resp.\ spin structures.

Following Lashof~\cite{Las63}, one may define bordism groups $\Omega_k^\xi$ of closed $k$-manifolds with $\xi$-structures, recovering the usual notions of oriented bordism, spin bordism, etc. Likewise, one can define bordism categories of $\xi$-manifolds, and therefore $\xi$-structured topological field theories.
\begin{defn}\label{def:twisted}
Let $V\to X$ be a vector bundle and $\xi\colon B\to B\O$ be a tangential structure. An \term{$(X, V)$-twisted
$\xi$-structure} on a vector bundle $E\to M$ is data of a map $f\colon M\to X$ and a $\xi$-structure on $E\oplus f^*(V)$.
\end{defn}
\begin{lem}[Shearing]
\label{shearing_lem}
Let $V_t\to B\O$ denote the tautological virtual vector bundle. If $\eta\colon B\times X\to B\O$ is the map classified by the vector bundle $\xi^*(V_t)\boxplus V$, then $\eta$-structures are equivalent to $(X, V)$-twisted $\xi$-structures.
\end{lem}
We do not know the origin of this result; see~\cite[Lemma 10.18]{DDHM23} for a proof.\footnote{The cited proof is given only for the case where $\xi$ is $\Spin$, but the same proof works for arbitrary $\xi$.} \Cref{shearing_lem} implies that we can consider bordism groups and topological field theories of $(X, V)$-twisted $\xi$-manifolds.
\begin{ansatz}
\label{prop:tangential_general}
Let $(G_b, s, \omega)$ be a fermionic symmetry and assume that there is a vector bundle $V\to BG_b$ such that $w_1(V) = s$ and $w_2(V) = \omega$.
The data of $(G_b, s, \omega)$ acting as a symmetry of a (spacetime dimension $n$) spin TFT $Z\colon\Bord_n^\Spin\to\cat C$ induces an extension of $Z$ to a (possibly anomalous) TFT $\hat{Z}\colon\Bord_n^\xi\to\cat C$ of manifolds with a $(BG_b, V)$-twisted spin structure $\xi$. $\hat{Z}$ lives on the boundary of an $(n+1)$-dimensional invertible field theory $\alpha$ of $(BG_b, V)$-twisted spin manifolds. We refer to $\alpha$ as the \textit{anomaly} of $Z$ with its $(G_b, s, \omega)$-symmetry.
\end{ansatz}

For a detailed discussion of this ansatz, see Freed-Hopkins~\cite[\S 2, \S 3]{FH21InvertibleFT} for general symmetries and Stehouwer~\cite{Ste21, Ste23} for fermionic symmetries specifically, as well as \cite{Kapustin:2014dxa,DGG21,Tho20,Wang2018} for justification from a physically motivated point of view. The assumption in \cref{prop:tangential_general} that a bundle $V\to BG_b$ with the required characteristic classes exists is not always true: see~\cite{GKT89, DY22, Debray:2023tdd} for some counterexamples and Teichner~\cite[Theorem 2.3.8]{teichner1992topological} for a necessary and sufficient criterion. In all examples relevant for this paper, though, such a bundle exists, so we will not worry about this nuance.

However, \cref{prop:tangential_general} chooses $V$, and generally there is more than one possible choice. Fortunately, if $V$ and $W$ have the same first two Stiefel-Whitney classes, the notions of $(X, V)$-twisted spin structure and $(X, W)$-twisted spin structures are equivalent, a fact which is implicit in Stolz~\cite{Sto98}, or more explicitly follows from~\cite[Corollary 3.3.8]{HJ20} (see also~\cite[Theorem 1.39]{Deb21}). In practice, one can forget about $V$ and just remember $w_1(V)$ and $w_2(V)$. This leads to a more expansive definition of a twisted spin structure than we gave in \cref{def:twisted}.
\begin{defn}[{Wang~\cite[Definition 8.2]{Wan08}}]
\label{coh_twspin}
Let $X$ be a space and choose $s\in H^1(X;\Z/2)$ and $\omega\in H^2(X;\Z/2)$. A \term{$(X, s, \omega)$-twisted spin structure} on a vector bundle $V\to M$ is data of a map $f\colon M\to X$ and trivializations of $w_1(V) + f^*(s)$ and $w_2(V) + f^*(s^2 + \omega)$.
\end{defn}
The Whitney sum formula implies $(X, V)$-twisted spin structures are in natural bijection with $(X, w_1(V), w_2(V))$-twisted spin structures. For us, $X$ will always be $BG_b$.

$(X, s,\omega)$-twisted spin structures are tangential structures in the sense of \cref{tangerine}.
\begin{defn}\label{tangspin} Let
\begin{subequations}
\begin{equation}\label{Xxisw}
    \xi_{X,s,\omega}\colon X\angl{s,s^2+\omega}\longrightarrow B\O\times X
\end{equation}
denote the fiber of the map
\begin{equation}
    (w_1 + s, w_2 + s w_1 +\omega)\colon B\O\times X\longrightarrow B\Z/2 \times B^2\Z/2.
\end{equation}
\end{subequations}
We will also use $\xi_{X,s,\omega}$ to refer to the tangential structure obtained by composing the map~\eqref{Xxisw} with the map $B\O\times X\to B\O$ given by projection onto the first factor.
\end{defn}
\begin{lem}\label{twotang}
$(X, s, \omega)$-twisted spin structures are in natural bijection with $\xi_{X,s,\omega}$-structures.
\end{lem}

Similarly, one can study twisted orientations.
\begin{defn}[{Olbermann~\cite[\S 1.4]{Olb07}}]
\label{twor_defn}
Let $X$ be a space and $s\in H^1(X;\Z/2)$. An \term{$(X, s)$-twisted orientation} on a vector bundle $V\to M$ is data of a map $f\colon M\to X$ and a trivialization of $w_1(V) - f^*(s)$.
\end{defn}
Just as in \cref{tangspin,twotang}, we let $\xi_{X,s}$ denote the tangential structure given by composing the fiber of $w_1 + s\colon B\O\times X\to B\Z/2$ with the projection $B\O\times X\to B\O$; then $\xi_{X,s}$-structures are in natural bijection with $(X, s)$-twisted orientations.

Therefore we are in the business of classifying invertible field theories of $(BG_b, V)$-twisted spin manifolds. These were classified by Freed-Hopkins-Teleman~\cite{FHT10} in terms of Reinhardt bordism~\cite{Rei63}, also called SKK bordism~\cite{KKNO73} or Madsen-Tillmann bordism~\cite{MT01, MW07}. However, because we are interested in anomalies of unitary field theories, we can make the simplifying assumption that the anomaly invertible field theory comes with data of \emph{reflection positivity}, the Wick-rotated analogue of unitarity. Freed-Hopkins~\cite{FH21InvertibleFT} classify reflection positive invertible topological field theories in terms of (ordinary) bordism.\footnote{Though all invertible field theories we consider in this paper are topological, see also Freed-Hopkins' conjecture~\cite[Conjecture 8.37]{FH21InvertibleFT}, recently proven by Grady~\cite{Gra23}, on reflection-positive invertible field theories that are not necessarily topological.} Before we state Freed-Hopkins' precise result, we need a little more notation.
\begin{defn}
\label{Pontr_defn}
The \term{Pontrjagin dual of the sphere spectrum} is the spectrum $I_{\U}$ representing the generalized cohomology theory $I_{\U}^*$ whose value on a space or spectrum $X$ is
\begin{equation}
\label{IU_coh}
    I_{\U}^n(X) \coloneqq \Hom(\pi_n^s(X), \U).
\end{equation}
For any spectrum $E$, we let $I_{\U}E\coloneqq \mathrm{Map}(E, I_{\U})$; the generalized cohomology theory defined by $I_{\U}E$ is
\begin{equation}
\label{IU_coh_2}
    (I_{\U}E)^n(X) \cong \Hom(E_n(X), \U).
\end{equation}
\end{defn}
It is not trivial that~\eqref{IU_coh} and~\eqref{IU_coh_2} satisfy the Eilenberg-Steenrod axioms: the proof makes use of the fact that $\U$ is an injective abelian group. In fact, one can define $I_A$ analogously for any injective abelian group $A$, and it is common to use $\Q/\Z$, $\R/\Z$, or $\C^\times$ in place of $\U$. Brown-Comenetz~\cite{BC76} were the first to consider a spectrum of this sort, with $A = \Q/\Z$, and so $I_{\Q/\Z}$ is sometimes called the \term{Brown-Comenetz dual of the sphere spectrum}; the use of $A = \U$ is more common in physics applications, beginning with Bunke-Schick~\cite[\S 4.2.3]{BS13}.
\begin{defn}\label{upside_down}
Recall that the bordism groups $\Omega_*^\xi(X)$ are the generalized homology theory associated to a spectrum $\mathit{MT\xi}$. We use the notation $\mho_\xi^*$ to denote the generalized cohomology theory represented by $I_{\U}\mathit{MT\xi}$.\footnote{The notation $\mho_\xi^*$ for the dual of $\Omega^\xi_*$ is meant to parallel ordinary homology and cohomology: the dual of ordinary homology, which uses a right-side-up $H$, is ordinary cohomology, which is written with an upside-down $H$.}
\end{defn}
\eqref{IU_coh_2} thus implies $\mho_\xi^n$ is a group of bordism invariants of $n$-dimensional $\xi$-manifolds.
\begin{thm}[{Freed-Hopkins~\cite[Theorem 8.29]{FH21InvertibleFT}}]
\label{FH_result}
There is a natural isomorphism from the group of isomorphism classes of reflection positive, invertible, $n$-dimensional TFTs $\alpha$ on manifolds with $\xi$-structure such that $\alpha(S^n) = 1$ to $\mho_\xi^n$.
\end{thm}
Here we need to specify a $\xi$-structure on $S^n$; we use the one induced by the trivial $\xi$-structure on $\underline\R^{n+1}\to S^n$, together with the isomorphism $TS^n\oplus\underline\R\cong\underline\R^{n+1}$. This condition on $\alpha(S^n)$ will not play an important role in our paper.

\Cref{FH_result} and \cref{prop:tangential_general} tell us that to study anomalies of fermionic symmetries, we should compute $\mho_\xi^4$ when $\xi$ is a $(BG_b, V)$-twisted spin structure. It is equivalent to compute the Pontrjagin-dual bordism groups. An element in the bordism group associated to the tangential structure in \cref{prop:tangential_general} is given by the following triple of data: 
\begin{enumerate}
\item a manifold ${M}$;
\item a map $f\colon{M}\rightarrow BG_b$, representing the classifying map of a $G_b$ gauge bundle on $M$;
\item a spin-structure $\zeta$ on $f^*(V)\oplus T{M}$. 
\end{enumerate}
In particular, to choose a spin-structure on $f^*(V)\oplus T{M}$, we must have
\begin{equation}\label{eq:requirement}
w_1(T{M}) = f^*(s), \quad w_2(T{M}) = w_1(T{M})^2 + f^*(\omega)
\end{equation}
In many situations we will not explicitly state the last two data when talking about an element in the bordism group, as they can be easily recovered from context. As an explicit example, in \S\ref{sec:Z/2^Timerev} we focus on the $\Z/4^{Tf}$ symmetry. 
From the point of view of \cref{prop:tangential_general},
the associated tangential structure is a $(B\Z/2, 3\sigma)$-twisted spin-structure, 
where $\sigma$ is the tautological line bundle on $B\Z/2\cong B\O(1)$. This is simply the notion of a \pinp structure, as shown by Stolz~\cite[\S 8]{Sto88}, and using Eq.~\eqref{eq:requirement}, we obtain the usual characteristic class condition $w_2(T{M})=0$ for a \pinp manifold $M$.

To compute twisted spin bordism or cobordism groups, we can use \cref{shearing_lem} to deduce that the $(X, V)$-twisted spin bordism groups are naturally isomorphic to $\Omega_*^\Spin(X^{V-r_V})$, where $r_V$ is the rank of $V$ and $X^E$ denote the \term{Thom spectrum} of a virtual vector bundle $E\to X$.\footnote{The term $-r_V$ shifts the Thom spectrum of $V$ so that the Thom class is in degree zero.} Then we can start applying the machinery of spectral sequences to do the calculation. In \cref{subsec:smith_cal}, we will also do the calculation with the help of the Smith homomorphism, as developed in \cite{DDKLPTT,DDKLPTT2}.

\begin{rem}
In this paper, we are interested in anomalies of $(2+1)$-d theories. The corresponding bordism groups, which are in degree $4$, can contain free summands, so their Pontrjagin duals can contain $\U$ summands. From the anomaly classification perspective, these summands are irrelevant, as the useful invariant is the deformation class of the anomaly. However, these $\U$ summands have a useful physical meaning: they describe Hall conductance, which will be discussed in Appendix~\ref{subapp:anomaly_U(1)}.
\end{rem}

\subsection{Fermionic Topological Order with Symmetry Action}\label{subsection:FermionicTO}

In this subsection, we give a quick review of the mathematical setup of (2+1)-d fermionic topological order with symmetry action, and list all the relevant data involved in the calculation of anomaly indicators. For a more comprehensive review of these concepts and notations, see  \cite{barkeshli2014,etingof2016tensor,Etingof2009,turaev1994quantum,Selinger2011,bakalov2001lectures}, and \cite{Galindo2017,Bulmash2021,Aasen2021} for a targeted review of fermionic topological order and super modular tensor category.

\begin{defn}\label{def:fTO}
    A (2+1)-d fermionic topological order is a unitary super modular tensor category (\emph{super-MTC}).\footnote{Unless explicitly stated otherwise, we always assume that we are working with a \emph{unitary} super-MTC.} This is a braided fusion category with Müger center given by the category of super vector spaces $\sVect$.\footnote{Our definition of a fermionic topological order uses a specific presentation of a super-MTC, with simple objects given by the collection of anyons. We take this as the definition because it is most convenient for performing the main computations in this paper.
Other methods of defining topological orders defines them up to gapped boundary, such as in \cite{kitaev2012models,kWZ2015boundary,KWZ2017boundary,LW2019classification,KLWZZ:2020,JF:2020usu,Johnson-Freyd:2021tbq}, and we will not be using this latter definition.}
\end{defn}

\begin{rem}
The Müger center of a braided fusion category $\mc{C}$ is the fusion subcategory of $\mc{C}$ which contains objects that have trivial double braiding with all other objects. It is well-known that a $(2+1)$-d bosonic topological order is a unitary modular tensor category (unitary-MTC) whose Müger center is simply $\Vect$. Therefore, we say that the unitary-MTC is \emph{non-degenerate} and the super-MTC is \emph{slightly degenerate}. 
\end{rem}

We now describe all the relevant data of a super-MTC $\mc{C}$. 
First, there is a finite set of simple objects $a$. They are referred to as (simple) anyons in the context of topological orders. Moreover, there is a special anyon $\psi$ in the M\"uger center which represents the \textit{local fermion}. The set of morphisms $\Hom(a, b)$ between two objects $a$ and $b$ forms a $\mathbb{C}$-linear vector space. 
In the context of bosonic topological order, $\Hom(a, b)$ can be viewed as the Hilbert space of states associated to a 2-sphere that hosts anyons $a$ and $\bar b$. For fermionic topological order, we can have a local fermion $\psi$ in the background.
Hence, it is sometimes useful to consider the $\Z/2$-graded Hilbert space $\Hom(a, b)\oplus\Hom(a\times\psi, b)$, with the grading denoting the background fermion number. This can be viewed as the Hilbert space of states associated to a 2-sphere that hosts anyons $a$ and $\bar b$ in fermionic topological order. 

$\mc{C}$ also has the structure of fusion and braiding.  Fusion means that there is a bifunctor $\times$ such that acting it on anyons $a$ and $b$ gives
\begin{equation}
a \times b \cong \bigoplus_c N^c_{ab} c\,,
\end{equation}
where $N^c_{ab}$ is interpreted as the dimension of the channels of how two anyons $a$ and $b$ fuse into a third anyon $c$. 

There are two related vector spaces, $V_{ab}^c$ and $V_c^{ab}$, referred to as the fusion and splitting vector spaces, respectively. The two vector spaces are dual to each other, and depicted graphically as:
\begin{equation}
\left( d_{c} / d_{a}d_{b} \right) ^{1/4}
 \raisebox{-0.5\height}{\includegraphics[page=1]{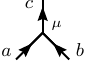}}
=\left\langle a,b;c\right|_\mu \in
V_{ab}^{c} ,
\label{eq:bra}
\end{equation}
\begin{equation}
\left( d_{c} / d_{a}d_{b}\right) ^{1/4}
 \raisebox{-0.5\height}{\includegraphics[page=2]{Draft2-pics}}
=\left| a,b;c\right\rangle_\mu \in
V_{c}^{ab},
\label{eq:ket}
\end{equation}
where $\mu=1,\ldots ,N_{ab}^{c}$, $d_a$ is the \emph{quantum dimension} of $a$, and the factors $\left(\frac{d_c}{d_a d_b}\right)^{1/4}$ are a normalization convention for the diagrams. 

More generally, for any integer $n$ and $m$ there are vector spaces 
$V^{a_1,a_2,\dots,a_n}_{b_1,b_2,\dots,b_m}$, which are referred to as the fusion space of
$m$ anyons into $n$ anyons. These vector spaces have a natural basis in terms of tensor products of the elementary splitting spaces $V^{ab}_c$ and fusion spaces $V^c_{ab}$. For instance, we have 
\begin{equation}
V^{abc}_d\cong \sum_e V^{ab}_e\otimes V^{ec}_d \cong \sum_f V^{af}_d\otimes V^{bc}_f
\end{equation}
The two vector spaces are related to each other by a basis transformation referred to as $F$-symbols, which is diagrammatically shown as follows
\begin{equation}
 \raisebox{-0.5\height}{\includegraphics[page=3]{Draft2-pics}}
= \sum_{f, \mu,\nu}
\left[F_d^{abc}\right]_{(e,\alpha,\beta),(f,\mu,\nu)}
  \raisebox{-0.5\height}{\includegraphics[page=4]{Draft2-pics}}
\end{equation}

There is a trivial anyon denoted by 1 such that $1\times a = a \times 1 = a$. We denote $\overline{a}$ as the anyon conjugate to $a$, for which
$N_{a \overline{a}}^1 = 1$, i.e.
\begin{align}\label{eq:defof1}
a \times \overline{a} = 1 \oplus\cdots
\end{align}
Note that $\bar a$ is unique for a given $a$. 

The $R$-symbols define the braiding properties of the anyons, and can be defined via the following
diagram:
\begin{equation}
 \raisebox{-0.5\height}{\includegraphics[page=5]{Draft2-pics}}
=\sum\limits_{\nu }\left[ R_{c}^{ab}\right] _{\mu \nu}
 \raisebox{-0.5\height}{\includegraphics[page=6]{Draft2-pics}}
  .
\end{equation}

Nevertheless, $F$ and $R$ symbols defined in such way will be dependent on the choice of the basis of the fusion spaces $V^{ab}_c$. Two sets of basis of $V^{ab}_c$ are related to each other by a basis transformation $\Gamma^{ab}_c : V^{ab}_c \rightarrow V^{ab}_c$ which is a unitary matrix acting on $V^{ab}_c$, see also \cite[Section II.C]{Delcamp:2024rjp}. Under the basis transformation, the $F$ and $R$ symbols change according to:
\begin{align}\label{eq:FRvertex_basis_trans}
F^{abc}_{def} &\rightarrow \tilde{F}^{abc}_d = \Gamma^{ab}_e \Gamma^{ec}_d F^{abc}_{def} [\Gamma^{bc}_f]^{-1} [\Gamma^{af}_d]^{-1}
\nonumber \\
R^{ab}_c & \rightarrow \tilde{R}^{ab}_c = \Gamma^{ba}_c R^{ab}_c [\Gamma^{ab}_c]^{-1} .
\end{align}
where we have suppressed splitting space indices and dropped brackets on the $F$-symbol for clarity of notation. In this paper, we refer to this basis transformation as a \emph{vertex basis transformation}. 
  
On the other hand, physical quantities, like the topological twist $\theta_a$ and the modular $S$-matrix $S_{ab}$, should always be basis-independent combinations of the data. The \emph{topological twist} $\theta_a$ is defined via the diagram:
\begin{equation}\label{eq:theta}
\theta _{a}=\theta _{\overline{a}}
=\sum\limits_{c,\mu } \frac{d_{c}}{d_{a}}\left[ R_{c}^{aa}\right] _{\mu \mu }
= \frac{1}{d_{a}}
 \raisebox{-0.5\height}{\includegraphics[page=7]{Draft2-pics}}
\end{equation}
Finally, the \emph{modular $S$-matrix} $S_{ab}$, is defined as
\begin{equation}
S_{ab} =D^{-1}\sum
\limits_{c}N_{\overline{a} b}^{c}\frac{\theta _{c}}{\theta _{a}\theta _{b}}d_{c}
=\frac{1}{D}
 \raisebox{-0.5\height}{\includegraphics[page=8]{Draft2-pics}}
\label{eqn:mtcs}
\end{equation}
where $D= \sqrt{\sum_a d_a^2}$ is the \emph{total dimension}. 

In a super-MTC $\mc{C}$, we have a special anyon $\psi$ which physically corresponds to the local fermion. $1$ and $\psi$ form the  M\"uger center $\sVect$, meaning that we must have $\theta_\psi=-1$, $\psi\times \psi=1$ and $\psi$ braids trivially with all anyons in $\mc{C}$. This also implies that the set of anyon labels of a super-MTC decomposes as $\mc{C}=\mc{C}_0\times\{1, \psi\}$. However, it does not mean that $\mc{C}$ is simply the Deligne tensor product of some unitary-MTC $\mc{C}_0$ with the M\"uger center $\{1, \psi\}$. Still, according to \cite[Theorem 3.5]{bruillard2017fermionic}, because the  M\"uger center is $\sVect$, the $S$-matrix of $\mc{C}$ does have a decomposition:
\begin{equation}\label{S-matrix}
S = \tilde S\otimes \frac{1}{\sqrt{2}}\left(\begin{array}{cc}
     1 & 1  \\
     1 & 1  \\
\end{array} \right),
\end{equation}
and $\tilde S$ is a unitary matrix. Hence we have ``modular'' in the terminology ``super modular tensor category''.

Now we want to equip $\mc{C}$ with a group action.
\begin{defn}[{Galindo-Venegas-Ramírez~\cite[Definition 3.11]{Galindo2017}}]
A \textit{fermionic action} of a fermionic symmetry $(G_b, s, \omega)$ on a super-MTC $\mc{C}$ is a categorical $G_b$ action on $\mc{C}$, such that $\psi$ is preserved under $G_b$ and the action of $G_b$ on $\psi$ canonically corresponds to $\omega\in H^2(BG_b; \Z/2)$.\footnote{In Bulmash-Barkeshli~\cite{Bulmash2021}, they also need to choose an explicit cochain representative for $\omega$, which lives in $Z^2(BG_b; \Z/2)$. We will be satisfied with $\omega$ being a cohomology element in $H^2(BG_b; \Z/2)$.}
\end{defn}
%

From a categorically point of view, let $\mc{C}$ be a fusion category, and we define $\underline{\mathrm{Aut}_\otimes(\mc{C})}$ to be the monoidal category where objects
are tensor autoequivalences of $\mc{C}$, arrows are tensor natural isomorphisms and the tensor product is the composition of tensor functors. A categorical $G$ action on $\mc{C}$ is a monoidal functor $\rho \colon \underline{G} \rightarrow \underline{\mathrm{Aut}_\otimes(\mc{C})}$, where $\underline{G}$ denotes the discrete monoidal category with objects the elements of $G$ and tensor product given by the multiplication of $G$. Thus, we have the following data.
\begin{itemize}
\item Braided tensor functors $\rho_{\bf{g}} \colon \mc{C} \rightarrow \mc{C}$, $\forall {\bf{g}} \in G$.
\item Natural isomorphisms $\eta({\bf g}, {\bf h}) \colon \rho_{\bf g}\colon \rho_{\bf h} \Longrightarrow \rho_{\bf gh}$, $\forall {\bf{g}},{\bf{h}} \in G$.
\end{itemize}

To perform explicit calculation, we write down the data associated to this definition and also the consistency conditions that the data should satisfy. Firstly, given an element ${\bf g}\in G_b$, we assign a functor $\rho_{\bf g}$ to it, which gives a group homomorphism (not a monoidal functor) 
\begin{equation}
\rho \colon G\rightarrow \mathrm{Aut}_\otimes(\mc{C}).
\end{equation}
In particular, we choose $\rho_{\bf 1}$ to always be the identity functor. $\bf g$ can permute the anyons and we use $\,^{\bf g} a$ to denote the (simple) anyon we get after the $\bf g$ action on the (simple) anyon labeled by $a$. According to the value of $s({\bf g})$ as in Eq.~\eqref{eq:qdef}, $\bf g$ is either a unitary or an anti-unitary element, and is mapped to either a unitary or anti-unitary monoidal functor, respectively. Moreover, as requested, the local fermion $\psi$ is preserved under every $\rho_{\bf g}$. When $G$ is continuous, we further choose $\rho_{\bf g}$ such that $\rho_{\bf g}$'s for different $\bf g$'s in the same connected component are the same functor.

The action of $\rho_{\bf{g}}$ on the fusion space $V_{ab}^c$ can be written in terms of the following matrix form
\begin{align}
\rho_{\bf g} |a,b;c\rangle_\mu = \sum_{\nu} U_{\bf g}(\,^{\bf g}a ,
  \,^{\bf g}b ; \,^{\bf g}c )_{\mu\nu} K^{s({\bf g})} |\,^{\bf g} a, \,^{\bf g} b; \,^{\bf g}c\rangle_\nu,
  \label{eqn:rhoStates}
\end{align}
where $U_{\bf g}(\,^{\bf g}a , \,^{\bf g}b ; \,^{\bf g}c ) $ is an $N_{ab}^c \times N_{ab}^c$ matrix, and $K$ denotes complex conjugation which appears when $s({\bf g})=1$ to account for the fact that the action $\rho_{\bf g}$ is now $\mathbb{C}$-anti-linear. We will call this set of data \textit{$U$-symbols}. 

To preserve the monoidal structure and braiding, the $U$-symbols should satisfy the following two consistency conditions, which involve $F$ and $R$ symbols, respectively:
\begin{align}\label{eqn:UFURConsistency}
U_{\bf g}(\,^{\bf g}a, \,^{\bf g}b; \,^{\bf g}e) U_{\bf g}(\,^{\bf g}e, \,^{\bf g}c; \,^{\bf g}d) F^{\,^{\bf g}a \,^{\bf g}b \,^{\bf g}c }_{\,^{\bf g}d \,^{\bf g}e \,^{\bf g}f} 
U^{-1}_{\bf g}(\,^{\bf g}b, \,^{\bf g}c; \,^{\bf g}f) U^{-1}_{\bf g}(\,^{\bf g}a, \,^{\bf g}f; \,^{\bf g}d) \notag = K^{s({\bf g})} F^{abc}_{def} K^{s({\bf g})} \\
U_{\bf g}(\,^{\bf g}b, \,^{\bf g}a; \,^{\bf g}c)  R^{\,^{\bf g}a \,^{\bf g}b}_{\,^{\bf g}c} U_{\bf g}(\,^{\bf g}a, \,^{\bf g}b; \,^{\bf g}c)^{-1}  = K^{s({\bf g})} R^{ab}_c K^{s({\bf g})}.
\end{align}
where we have suppressed the additional indices that
appear when $N^c_{ab} > 1$. Moreover, under a vertex basis transformation, $\Gamma^{ab}_c : V^{ab}_c \rightarrow V^{ab}_c$,  $U_{\bf g}(a, b; c )_{\mu\nu} $ also transforms according to
\begin{equation}\label{eq:Uvertex_basis_trans}
\tilde{U}_{{\bf g}}(a,b,c) = 
\left[\Gamma^{\,^{\bf \overline{g}}{a} \,^{\bf \overline{g}}{b}}_{\,^{\bf\overline{g}}{c}}\right]^{\varsigma({\bf g})} U_{{\bf g}}(a,b,c) \left[(\Gamma^{ab}_c)^{-1}\right].
\end{equation}
Here, for the ease of notation, we introduce the shorthand notation $\overline{{\bf g}}={\bf g}^{-1}$, and we also introduce $\varsigma({\bf g})$ which is simply another version of $s(\bf g)$,
\begin{align}
\varsigma(\bf g) = \left\{
\begin{array} {ll}
1 & \text{if $\bf g$ is unitary} \\
* & \text{if $\bf g$ is anti-unitary} \\
\end{array} \right.
\label{eqn:sigmaDef}
\end{align}
where $*$ denotes complex conjugation. 

Secondly, to account for the multiplication rule of $G_b$, there should be a natural isomorphism $\eta({\bf g}, {\bf h})$ connecting $\rho_{\bf g}\circ\rho_{\bf h}$ with $\rho_{\bf gh}$:
\begin{equation}
\eta(\bf g, \bf h):\quad\rho_{\bf g}\circ\rho_{\bf h}\Longrightarrow \rho_{\bf gh}\,.
\end{equation}
By the definition of natural isomorphism, for every anyon $a$, $\eta(\bf g, \bf h)$ assigns an isomorphism $\eta_{\,^{\bf gh}a}(\bf g, \bf h)\in \Hom(^{\bf g}\left(^{\bf h}a\right), \,^{\bf gh}a)$ to ${\,^{\bf gh}a}$. Therefore, for every simple anyon $a$, we must have $\,^{\bf g}\left(\,^{\bf h}a\right) = \,^{\bf gh}a$ and $\eta_a({\bf g}, {\bf h})$ is simply a $\U$ phase factor. We will call this set of data \textit{$\eta$-symbols}. 

Acting on the fusion space $V^c_{ab}$, we have the following consistency condition between $U$- and $\eta$-symbols
\begin{align}
\label{eq:UetaConsistency}
\frac{\eta_a({\bf g}, {\bf h}) \eta_b({\bf g}, {\bf h})}{\eta_c({\bf g}, {\bf h}) } = U_{\bf g}(a,b;c)^{-1} K^{s({\bf g})} U_{\bf h}( \,^{\overline{\bf g}}a, \,^{\overline{\bf g}}b; \,^{\overline{\bf g}}c  )^{-1} K^{s({\bf g})} U_{\bf gh}(a,b;c )\,.
\end{align}
To satisfy the group associativity on the nose, we impose the following constraint on $\eta$-symbols 
\begin{align}
	\eta_a({\bf g},{\bf h})\eta_a({\bf gh}, {\bf k}) =\eta_a({\bf g}, {\bf hk}) \eta_{\,^{\overline{\bf g}}a}({\bf h}, {\bf k})^{\varsigma({\bf g})}.
	\label{eq:etaConsistency}
\end{align}

These data, i.e., $U$-symbols and $\eta$-symbols, are associated to the data of braided tensor functors and natural isomorphisms in the definition of the categorical $G_b$ actions. Two different sets of functors $\rho_{\bf{g}}$ and $\tilde{\rho}_{\bf{g}}$ can be identified if they are connected by some natural isomorphism $\gamma(\bf g)$
    \begin{equation}\label{eq:natual_iso_gamma}
\gamma(\bf g):~~~\rho_{\bf g}\Longrightarrow \tilde{\rho}_{\bf g},
\end{equation}
which we refer to as the \textit{symmetry action gauge transformation}.
 This changes $U_{\bf g}(a,b;c)$ and $\eta_a(\bf g, \bf h)$ in the following way:
\begin{align}\label{eq:Gaction_gauge}
U_{\bf g}(a,b;c) &\rightarrow \frac{\gamma_a({\bf g}) \gamma_b({\bf g})}{ \gamma_c({\bf g}) } U_{\bf g}(a,b;c)
\nonumber \\
  \eta_a({\bf g}, {\bf h}) & \rightarrow \frac{\gamma_a({\bf gh})}{\gamma_a({\bf g})(\gamma_{\,^{\overline{\bf g}} a}(\bf h))^{\varsigma({\bf g})}  } \eta_a({\bf g}, {\bf h})\,.
\end{align}
Different gauge inequivalent choices of $\{\eta\}$ and $\{U\}$ characterize distinct symmetry fractionalization classes. In this paper we will always fix the gauge
\begin{align}\label{eq:gauge_choice}
  \eta_1({\bf g},{\bf h})=\eta_a({\bf 1},{\bf g}) = \eta_a({\bf g},{\bf 1})&=1 \nonumber \\
  U_{\bf g}(1,b;c)=U_{\bf g}(a,1;c)&=1.
  \end{align}
  
Finally, in the presence of the local fermion $\psi$, we further require that the cocycle $\eta_\psi({\bf g},{\bf h})$, which in fact represents an element $[\eta_\psi]$ in $H^2(BG_b; \Z/2)$, coincides with $\omega\in H^2(BG_b; \Z/2)$ in the definition of the fermionic symmetry, i.e.
\begin{equation}\label{eq:fermion_SF}
[\eta_\psi] = \omega.
\end{equation}

Given a set of functors $\{\rho_{\bf g}\}$ connected by some natural isomorphism $\eta_0({\bf g}, {\bf h})$ as in Eq.~\eqref{eq:UetaConsistency}, we may have different choices of $\eta({\bf g}, {\bf h})$.
\begin{defn}\label{def:symfrac}
For the same set of functors $\{\rho_{\bf g}\}$, different solutions $\eta_a(\bf g, \bf h)$ of Eq.~\eqref{eq:UetaConsistency}, \eqref{eq:etaConsistency}, and~\eqref{eq:fermion_SF} are referred to as different \emph{symmetry fractionalization classes}.
\end{defn}
In fact, Eq.~\eqref{eq:etaConsistency} and Eq.~\eqref{eq:fermion_SF} may never be satisfied by any choice of $\eta({\bf g}, {\bf h})$. Such requirements define two obstructions that take values in certain cohomology, which are referred to as the \emph{obstruction to symmetry fractionalization}. When the two obstructions vanish, different symmetry fractionalization classes form a torsor over $H^2_{[\rho]}(BG_b; \mc{A}/\{1, \psi\})$. Here $\mathcal{A}$ is the set of Abelian anyons in $\mc{C}$, which form an Abelian module of $G_b$ under the action of $\rho$. Because $\{1, \psi\}$ is a subgroup of $\mc{A}$ isomorphic to $\Z/2$ and also a trivial $\Z[G_b]$-module with respect to the action of $\rho$, there is a well-defined $\Z[G_b]$-module structure on the quotient $\mc{A}/\{1, \psi\}$, so that we can take the cohomology $H_{[\rho]}^2(BG_b; \mc{A}/\set{1, \psi})$. This construction is discussed in detail in~\cite{Galindo2017,Bulmash2021,Aasen2021}.

In summary, a (2+1)-d fermionic topological order is described by a super-MTC $\mc{C}$, and a fermionic symmetry with data $(G_b, s, \omega)$ acting on $\mc{C}$ requires the data $\{\rho_{\bf g}; U_{\bf g}(a,b;c), \eta_a({\bf g}, {\bf h})\rbrace$ associated to $G_b$, satisfying various consistency conditions as in Eqs.~\eqref{eqn:UFURConsistency}, \eqref{eq:UetaConsistency}, \eqref{eq:etaConsistency} and \eqref{eq:fermion_SF}.

\section{Bosonization Conjectures and Application to Fermionic Topological Orders}\label{sec:spinTFTsConj}

In this section, we build up the background to state the bosonization conjecture in \cref{upstream_conjecture}. After stating this we will discuss how to use the conjecture to arrive at \cref{conj:inv} which states that the anomaly of the fermionic topological order should be an invertible field theory. The former two points are the main focus of \S\ref{subsection:unpack}. With these formal statements in place, we will have established the groundwork to calculate the partition function for the anomaly of fermionic topological order. In \S\ref{subsec:recipe}, we present explicit recipes to calculate the partition function from the contents of the super modular tensor category and the manifold generators of the bordism group.

\subsection{Unpacking the Conjectures}\label{subsection:unpack}

Bosonization is a general correspondence between bosonic systems (whose fundamental degrees of freedom are bosonic) and fermionic systems (whose fundamental degrees of freedom
are fermionic). After some background of higher category theory in \S\ref{subsubsec:highercat} and boundary theories from relative field theory point of view in \S\ref{subsubsec:boundary}, we formalise our notion of bosonization/fermionization in \S\ref{subsubsec:2dExample} and \S\ref{subsubsec:higherd}. Much of the details is aimed to incorporate the data of the global symmetry, the tangential structure, and the anomaly, in a convenient formalism such that we can identify anomaly indicators as the partition functions of invertible fermionic topological field theories.

\subsubsection{A little more category theory}\label{subsubsec:highercat}

Throughout this subsection, we work with $n$-dimensional TFTs valued in a symmetric monoidal weak $k$-category $\fC_k$,
where $1\le k\le n$.\footnote{When $k<n$ we think of $\mathcal{C}_k$ as a Karoubi completion of the delooping of $\mathcal{C}_n$.} In this section, by the dimension $n$ of a TFT we refer to the spacetime dimension; in other sections, this is typically written as $d+1$.
We need $\fC_k$ to satisfy the following two properties.
\begin{enumerate}
	\item $\fC_k$ is \term{additive} (see Gaiotto-Johnson-Freyd~\cite[\S 4.3]{Gaiotto2019} for the definition of additive higher categories),
        and there is a finite path integral construction for TFTs valued in $\fC_k$ in the sense of Freed-Quinn~\cite{FQ:1991bn} for $k = 1$ and Freed~\cite{Fre94} and Freed-Hopkins-Lurie-Teleman~\cite{FHLT:2009qp} (see also~\cite{Fre93, Fre95, Fre99}) for $k > 1$.\footnote{There is not yet a construction of the finite path integral for arbitrary additive $\fC_k$, so its existence is part of our assumption (see work in progress by Scheimbauer-Walde for the case of $k=3$ \cite{ClaudiaLecture}). Compare~\cite[Hypothesis 1]{VD23}. In many cases, though, the finite path integral has been constructed: in addition to the above articles, see Morton~\cite{Mor15}, Trova~\cite{Tro16}, Carqueville-Runkel-Schaumann~\cite{CRS19}, Schweigert-Woike~\cite{SW19, SW20}, and Harpaz~\cite{Har20}, and see~\cite{CR12, BCP14a, BCP14b, SW18, CRS20, CMRSS21, Car23, CM23} for some related constructions.} 
        The finite path integral is the mathematical instantiation of the procedure of gauging a finite symmetry of a topological field theory; this symmetry could include fermion parity.
	\item \label{item:lifting} Let $\fC_k^\times$ denote the sub-$k$-category of invertible objects and invertible (higher) morphisms;
	$\fC_k^\times$ is a Picard $k$-groupoid, meaning its classifying space is canonically $\Omega^\infty$ of a $k$-connective spectrum $\abs{\fC_k^\times}$; see, e.g.,~\cite[\S 6.5]{freed2019lectures} for more information. We need $\abs{\fC_k^\times}$ to be homotopy
	equivalent to the connective cover of $I_{\U}$, the Pontrjagin dual of the sphere spectrum from \cref{Pontr_defn}.
This assumption on $\fC_k$ allows us to lift $\U$-valued bordism invariants into invertible TFTs valued
	in $\fC_k$; see~\cite{FH21InvertibleFT, freed2019lectures} for more information.
\end{enumerate}
If $\textbf{1}$ denotes the tensor unit in $\fC_k$, $\Omega\fC_k\coloneqq \End_{\fC_k}(\textbf{1})$ is a
symmetric monoidal $(k-1)$-category satisfying the above two conditions, so if we have chosen $\fC_k$, we define
$\fC_{k-1}\coloneqq\Omega\fC_k$.

Categories $\fC_k$ satisfying the above two properties are known for $k\le 2$: for $k = 1$, one may use
$\sVect_\C$, the category of complex super vector spaces, and for $k
= 2$, one may use $\sAlg_\C$, the Morita bicategory of complex superalgebras~\cite{Fre12, DG18}. For $k > 2$ our
choice of $\fC_k$ is therefore an ansatz: it is a conjecture of Freed-Hopkins~\cite[\S 5.3]{FH21InvertibleFT} that such $\fC_k$
exist (see work in progress by Freed-Scheimbauer-
Teleman \cite{TelemanLecture} for $k = 3$ and Johnson-Freyd-Reutter \cite{FreydLecture,reutter_slides} for $k > 3$). We will assume $\Omega^{k-2}\fC_k \simeq \sAlg_\C$, which implies $\Omega^{k-1}\fC_k \simeq \sVect_\C$. We
will often heuristically think of the objects of $\fC_k$ as algebras, monoidal categories, etc., by analogy with
the common use of monoidal $\C$-linear categories, resp.\ braided monoidal $\C$-linear categories to discuss 3d,
resp.\ 4d TFTs.


\subsubsection{$G$-symmetric theories as boundaries}\label{subsubsec:boundary}

Understanding the anomaly requires us to put the theory on the boundary of some (invertible) TFT. In this subsubsection we formalize this notion from the point of view of Atiyah-Segal-style TFT \cite{Atiyah:88, SegalCFT}.

Let $F_{\mathrm{triv}}\colon\Bord_{n+1}^\xi\to\fC_k$ be the \term{trivial $(n+1)$-dimensional TFT}, i.e.\ the symmetric monoidal
functor sending all objects to the tensor unit $\textbf{1}$ and all (higher) morphisms to the identity. The following definition is a formalization of the notion of an $n$-dimensional boundary.

\begin{defn}[{Freed-Teleman~\cite[Definition 2.1]{FT14}}]
Let $\alpha\colon\Bord_{n+1}^\xi\to\fC_k$ be a topological field theory, and for the $(n+1)$-d trivial TFT
$F_{\mathrm{triv}}\colon\Bord_{n+1}^\xi\to\fC_k$, let $\tau_{\le n}F_{\mathrm{triv}}$ be the
restriction of $F_{\mathrm{triv}}$ to the sub-$k$-category of $\Bord_{n+1}^\xi$ consisting of all objects and $j$-morphisms for $1\le j
< k$, and only the identity $k$-morphisms. A \term{topological field theory relative to $\alpha$} is a natural transformation
\begin{equation}\label{rel_nat_hom}
	{\mathfrak Z}\colon \tau_{\le n}F_{\mathrm{triv}}\longrightarrow \tau_{\le n}\alpha.
\end{equation}
\end{defn}
If we did not impose $\tau_{\le n}$, all relative TFTs would be equivalences between $F_{\mathrm{triv}}$ and
$\alpha$, as the category $\Fun^\otimes(\Bord_{n+1}, \fC_k)$ is a $k$-groupoid (see, e.g.,~\cite[Exercise
2.10.15]{etingof2016tensor} or~\cite[Lemma 2.13]{CR18} for the case $k = 1$, \cite[Remark 2.4.7(a), Theorem
2.4.18]{Lurie2009} for the case $k = n$, and~\cite[\S 6]{Fre13} for general $k$).

For an absolute TFT $Z$ and an $n$-manifold $M$, the partition function of $Z$ on $M$ is a complex number; but if
${\mathfrak Z}$ is a TFT relative to $\alpha$, ${\mathfrak Z}(M)$ is an element of the state space $\alpha(M)\in\sVect_\C$.
\begin{defn}
\label{relative_with_tangential_structures}
Fix two tangential structures $\xi\colon B_1\to B\O$ and $\eta\colon B_2\to B\O$, and a map $\Phi\colon \eta\to\xi$, i.e.\ a map $B_2\to B_1$ commuting with the maps to $B\O$. This induces a functor $f_\Phi\colon\Bord_m^\eta\to\Bord_m^\xi$ for all $m$, which is symmetric monoidal.

Let $\alpha\colon\Bord_{n+1}^\xi\to\fC_k$ be a topological field theory. Then an \term{$\eta$-structured TFT relative to $\alpha$} is a TFT relative to $f_\Phi\circ\alpha\colon\Bord_{n+1}^\eta\to\fC_k$.
\end{defn}
This allows us to give additional structure to boundary theories, e.g.\ spin boundary theories of an oriented TFT.

Let $G$ be a finite group, and given a tangential structure $\xi\colon B\to B\O$, let $\xi(G)$ denote the tangential structure $\xi\circ\pi_1\colon B\times BG\to B\O$, where $\pi_1$ is projection onto the first factor. Thus a $\xi(G)$-structure is data of a $\xi$-structure and a principal $G$-bundle.

There is an $(n+1)$-dimensional theory $F_G$ of unoriented manifolds, called \term{finite gauge theory}, obtained by performing the finite path integral to sum $F_{\mathrm{triv}}\colon\Bord_{n+1}^{\id(G)}\to\fC_k$
over $G$-bundles~\cite{DW90, FQ:1991bn}.
The following fact is
a consequence of the cobordism hypothesis:
\begin{lem}
\label{G_sym_boundary}
There is an equivalence of categories between TFTs $Z\colon\Bord_n^{\xi(G)}\to\fC_{n}$ and $\xi$-structured TFTs $\mathfrak Z$ relative to
$F_G$.
\end{lem}
See, e.g.,~\cite[Remark 3.10]{FMT22}, as well as~\cite{Fuchs:2002cm,Fuchs:2014ema,Fuchs:2012dt}.  \Cref{G_sym_boundary} has the concrete
interpretation that being a field theory with global $G$-symmetry is equivalent to being a boundary theory to $F_G$.
\begin{rem}[Partition functions relative to $F_G$]
Let $\mathfrak Z$ be an $n$-dimensional TFT relative to $F_G$ and $Z$ be the corresponding TFT of manifolds with a principal $G$-bundle under the equivalence of \cref{G_sym_boundary}. In this remark, we spell out the way in which the partition functions of $\mathfrak Z$ and $Z$ are ``the same,'' a perspective we learned from Gaiotto-Kulp~\cite[\S 2.1]{GK:2020iye}.

Let $M$ be a closed $n$-manifold. Then the natural transformation~\eqref{rel_nat_hom} implies the partition function ${\mathfrak Z}(M)$ is, rather than a complex number, a morphism $\varphi_M\colon \textbf{1}(M)\to F_G(M)$. Here both $\textbf{1}$ and $F_G$ are $(n+1)$-dimensional theories, $\textbf{1}(M)$ and $F_G(M)$ are state spaces: specifically, the complex vector spaces $\C$ and $\C[\cat{Bun}_G(M)]$,\footnote{If $\mathcal G$ is a groupoid, $\C[\mathcal G]$ means the vector space of functions on $\pi_0(\mathcal G)$, so $\C[\cat{Bun}_G(M)]$ is the vector space of functions on \emph{isomorphism classes} of principal $G$-bundles on $M$. Since $M$ is closed and $G$ is finite, this is a finite-dimensional vector space.} respectively. Therefore the morphism is determined by $\varphi_M(1)$, and we identify ${\mathfrak Z}(M)$ with the ``partition vector'' $\varphi_M(1)\in \C[\cat{Bun}_G(M)]$.

Meanwhile, $Z$ is an absolute TFT (i.e., an $n$-d TFT relative to the trivial $(n+1)$-d TFT) that requires the additional data of a principal $G$-bundle $P\to M$, such that isomorphic principal $G$-bundles have equal partition functions. Therefore $Z$ defines a function
\begin{equation}
    \begin{aligned}
        \psi_M\colon \pi_0(\cat{Bun}_G(M)) &\longrightarrow \C\\
            P &\longmapsto Z(M, P).
    \end{aligned}
\end{equation}
In other words, $\psi_M$ is an element of $\C[\cat{Bun}_G(M)]$.

The relationship between $Z$ and $\mathfrak Z$ espoused by \cref{G_sym_boundary} is that $\varphi_M(1) = \psi_M$.
\end{rem}
If objects of $\fC_k$ are algebras of some kind, $F_G(\pt)$ is the group algebra $\textbf{1}[G]$, and
\cref{G_sym_boundary} tells us that theories with a (nonanomalous) $G$-symmetry are equivalent data to (either left
or right) $\textbf{1}[G]$-modules.
\begin{exm}
When $n+1 = 3$ and the target category is $\cat{Cat}_\C[E_1]$,
the tricategory of monoidal $\C$-linear categories, the tensor unit is $\Vect_\C$ and $F_G(\pt) =
\Vect[G]$. In this case, the correspondence between $(2+1)$-d theories with a $G$-symmetry and
$\Vect[G]$-modules is spelled out by Freed-Teleman~\cite[\S 3.1]{FT22}.
\end{exm}
There is an analogue of $F_G$ called $F_\Spin$, an oriented TFT obtained by performing the finite path integral to
sum $F_{\mathrm{triv}}\colon\Bord_{n+1}^\Spin\to\fC_k$ over all spin structures inducing a fixed orientation. This is a well-studied idea, often studied under the name ``gauging fermion parity'' and going back to work of Seiberg-Witten~\cite{SW86} and Álvarez-Gaumé-Ginsparg-Moore-Vafa~\cite{AGGMV86} interpreting the Gliozzi-Sherck-Olive projection on a superstring worldsheet as a sum over spin structures. We are specifically interested in dimension $n = 3$, where $F_\Spin$ is the theory $\mathscr S$ discussed by Johnson-Freyd~\cite[\S 2.2]{Johnson-Freyd:2020twl}; we discuss this in more detail in \cref{theory_S}. \Cref{G_sym_boundary} generalizes to imply that $n$-dimensional spin TFTs are equivalent data to TFTs relative to
$F_\Spin$.

\subsubsection{The Jordan-Wigner transform for 2d TFTs}\label{subsubsec:2dExample}

As a simplified case to the contents later in this paper, in this subsubsection, we consider $n=2$, and discuss the contents of bosonization and Jordan-Wigner transformation from our point of view. The contents here have been studied elsewhere in the literature from the field theoretic point of view and on the lattice \cite{KT:2017jrc,Karch:2019lnn,Chen:2017fvr,GK:2020iye,Inamura:2022lun,Ji:2019ugf}. When $n=2$, a special feature is that we can avoid discussing nontrivial higher form symmetries after bosonization, making the discussion relatively clean. We will come back to this nontrivial point in \S\ref{subsubsec:higherd}.

Given a closed $2$-manifold $\Sigma$ with spin structure $\mathfrak s$, let
$a(\Sigma, \mathfrak s)\in\set{\pm 1}$ denote its \term{Arf invariant}; this is a spin bordism invariant and
defines an isomorphism $a\colon \Omega_2^\Spin\to\set{\pm 1}$. Spin structures on any manifold $M$ are a torsor
over $H^1(M;\Z/2)$; thus, given a principal $\Z/2$-bundle $P\to M$ and a spin structure $\mathfrak s$ on $M$, let
$\mathfrak s + P$ denote the spin structure given by acting on $\mathfrak s$ by the class $w_1(P)\in H^1(M;\Z/2)$.
Then the function taking a closed surface $\Sigma$, a spin structure $\mathfrak s$ on $\Sigma$, and a principal
$\Z/2$-bundle $P\to\Sigma$ to
\begin{equation}
	a_{\mathit{JW}}\colon (\Sigma, \mathfrak s, P)\longmapsto a(\Sigma, \mathfrak s + P)
\end{equation}
is a bordism invariant $a_{\mathit{JW}}\colon \Omega_2^{\Spin}(B\Z/2)\to\set{\pm 1}$, and therefore by property \ref{item:lifting} of $\fC_k$ lifts to define a $2$-dimensional invertible TFT
\begin{equation}
	z_c\colon\Bord_2^{\Spin\times\Z/2}\longrightarrow \sAlg_\C.
\end{equation}
Because this theory is defined for manifolds with a spin structure \emph{and} a principal $\Z/2$-bundle, it exists
relative to both $F_{\Z/2}$ and $F_\Spin$.
That is, if we fix the spin structure, $z_c$ is a theory on manifolds with a $\Z/2$-bundle, hence
by \cref{G_sym_boundary} is equivalent data to a spin TFT relative to $F_{\Z/2}$ (i.e.\ in the sense of \cref{relative_with_tangential_structures}), meaning a natural transformation from (the truncations of) the trivial spin TFT to $F_{\Z/2}$ regarded as a spin TFT. Since spin TFTs are equivalent to
theories relative to $F_\Spin$, $z_c$ is also equivalent to data of a homomorphism
\begin{subequations}
\label{weird_bimodule}
\begin{equation}
	z_c'\colon \tau_{\le 2}F_{\Spin}\longrightarrow \tau_{\le 2} F_{\Z/2\times\Spin}.
\end{equation}
Likewise, holding the $\Z/2$-bundle fixed and letting the spin structure vary, $z_c$ defines a homomorphism
\begin{equation}
	z_c''\colon \tau_{\le 2}F_{\Z/2\times\Spin}\longrightarrow \tau_{\le 2} F_{\Z/2}.
\end{equation}
\end{subequations}
Composing $z_c''$ and $z_c'$, we obtain a homomorphism $\tau_{\le 2}F_{\Spin}\to
\tau_{\le 2}F_{\Z/2}$, i.e.\ a \term{defect} between these two TFTs, akin to a bimodule between two algebras. In
general, given three $3$-d TFTs $A$, $B$, and $C$, a $(B, A)$-defect $Z'$, and a $(C, B)$-defect $Z''$, one can compose the
homomorphisms defining $Z'$ and $Z''$ as in Eq.~\eqref{weird_bimodule} above, to form a $(C, A)$-defect which we denote
$Z''\otimes_B Z'$.\footnote{Our notation is inspired by the fact that when the target category is $\sAlg_\C$ or
$\cat{Cat}_\C[E_1]$, composition of defects corresponds to tensor product of bimodules, which is composition in the
Morita category.} If $A = F_{\mathrm{triv}}$, a $(B, A)$-defect is the same thing as a TFT relative to $B$.

The point of all this is that tensoring with $\alpha$ exchanges theories relative to $F_{\Z/2}$ (i.e.\
$(F_{\mathrm{triv}}, F_{\Z/2})$-defects) with theories relative to $F_{\Spin}$ (i.e.\ $(F_{\Spin}, F_{\mathrm{triv}})$-defects):
\begin{defn}
\label{2d_BF_defn}
Let $\mc{Z}_f\colon\Bord_2^\Spin\to\sAlg_\C$ be a 2d spin TFT. The \term{bosonization} of $\mc{Z}_f$ is the TFT
\begin{equation}
	Z_b\coloneqq \mc{Z}_f\otimes_{F_\Spin} z_c\colon \Bord_2^{\SO\times\Z/2}\to\sAlg_\C.
\end{equation}
Likewise, given a TFT $W_b\colon\Bord_2^{\SO\times\Z/2}\to\sAlg_\C$, its \term{fermionization} is $\mc{W}_f\coloneqq
W_b\otimes_{F_{\Z/2}} z_c$. The result of the tensor product for bosonization is summarized in a sandwich construction in \cref{fig:BosSandwich}.
\end{defn}
Taking the tensor product of two theories $Z_1$ and $Z_2$ over $F_{\Spin}$, resp.\ over $F_{\Z/2}$ amounts to first forming the usual tensor product $Z_1\otimes Z_2$ of TFTs, then summing over spin structures, resp.\ principal $\Z/2$-bundles. Freed-Quinn~\cite[(2.9)]{FQ:1991bn} give a formula for the partition functions of finite path integral TFTs, allowing us to write down formulas for the partition functions of the bosonization or fermionization of a theory. First, the bosonization; let $\Sigma$ be a closed, oriented surface, $P\to \Sigma$ be a principal $\Z/2$-bundle, $\cat{Spin}(\Sigma)$ be the groupoid of spin structures on $\Sigma$, and $\mathcal Z_f\colon\Bord_2^\Spin\to\sAlg_\C$ be a TFT. Then~\cite[(2.3)]{Tho20}
\begin{subequations}
\label{2d_JW_formulas}
\begin{equation}
    Z_b(\Sigma, P) = \frac{1}{2^{\# \pi_0(\Sigma)}}\sum_{\mathfrak s\in\pi_0\cat{Spin}(\Sigma)} a_{\mathit{JW}}(\Sigma, P, \mathfrak s) \mathcal Z_f(\Sigma, \mathfrak s).
\end{equation}
Likewise, with $\Sigma$ as above, choose a spin structure $\mathfrak s$ on $\Sigma$ and a TFT $W_b\colon\Bord_2^{\SO\times\Z/2}\to\sAlg_\C$; then~\cite[(2.7)]{Tho20}
\begin{equation}
    \mathcal W_f(\Sigma, \mathfrak s) =\frac{1}{2^{\#\pi_0(\Sigma)}} \sum_{P\in \pi_0\cat{Bun}_{\Z/2}(\Sigma)} a_{\mathit{JW}}(\Sigma, P, \mathfrak s) W_b(\Sigma, P).
\end{equation}
\end{subequations}
The resemblance to the Fourier transform is no coincidence; we will return to this point in \S\ref{subsection:conjforInv}.
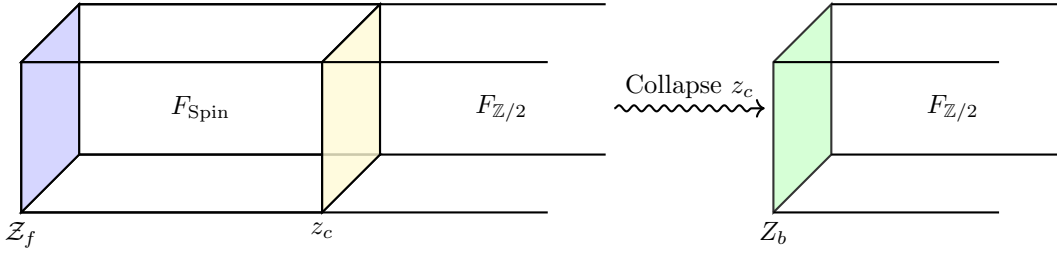
\begin{figure}[ht]
\centering
    \begin{tikzpicture}[thick]
    
        \def\Depth{4}
        \def\DepthTwo{3}
        \def\Height{2}
        \def\Width{2}
        \def\Sep{3}        
        
        \coordinate (O) at (0,0,0);
        \coordinate (A) at (0,\Width,0);
        \coordinate (B) at (0,\Width,\Height);
        \coordinate (C) at (0,0,\Height);
        \coordinate (D) at (\Depth,0,0);
        \coordinate (E) at (\Depth,\Width,0);
        \coordinate (F) at (\Depth,\Width,\Height);
        \coordinate (G) at (\Depth,0,\Height);
        \draw[black] (O) -- (C) -- (G) -- (D) -- cycle;
        \draw[black] (O) -- (A) -- (E) -- (D) -- cycle;
        \draw[black, fill=blue!20,opacity=0.8] (O) -- (A) -- (B) -- (C) -- cycle;
        \draw[black, fill=yellow!20,opacity=0.8] (D) -- (E) -- (F) -- (G) -- cycle;
        \draw[black] (C) -- (B) -- (F) -- (G) -- cycle;
        \draw[black] (A) -- (B) -- (F) -- (E) -- cycle;
        \draw[below] (0, 0*\Width, \Height) node{$\mc{Z}_f$};
        \draw[below] (\Depth, 0*\Width, \Height) node{$z_c$};
        \draw[midway] (\Depth/2,\Width-\Width/2,\Height/2) node {$F_{\text{Spin}}$};
        \draw[midway] (\Depth/2+\Depth,\Width-\Width/2,\Height/2) node {$F_{\mathbb{Z}/2}$ };
        
        \coordinate (O2) at (\Depth,0,0);
        \coordinate (A2) at (\Depth,\Width,0);
        \coordinate (B2) at (\Depth,\Width,\Height);
        \coordinate (C2) at (\Depth,0,\Height);
        \coordinate (D2) at (\Depth+\DepthTwo,0,0);
        \coordinate (E2) at (\Depth+\DepthTwo,\Width,0);
        \coordinate (F2) at (\Depth+\DepthTwo,\Width,\Height);
        \coordinate (G2) at (\Depth+\DepthTwo,0,\Height);
        \draw[black] (O2) -- (D2);
        \draw[black] (A2) -- (E2);
        \draw[black] (B2) -- (F2);
        \draw[black] (C2) -- (G2);
        
        \coordinate (O3) at (\Depth+\DepthTwo+\Sep,0,0);
        \coordinate (A3) at (\Depth+\DepthTwo+\Sep,\Width,0);
        \coordinate (B3) at (\Depth+\DepthTwo+\Sep,\Width,\Height);
        \coordinate (C3) at (\Depth+\DepthTwo+\Sep,0,\Height);
        \coordinate (D3) at (\Depth+2*\DepthTwo+\Sep,0,0);
        \coordinate (E3) at (\Depth+2*\DepthTwo+\Sep,\Width,0);
        \coordinate (F3) at (\Depth+2*\DepthTwo+\Sep,\Width,\Height);
        \coordinate (G3) at (\Depth+2*\DepthTwo+\Sep,0,\Height);
            \draw[black, fill=green!20,opacity=0.8] (O3) -- (C3) -- (B3) -- (A3) -- cycle;
        \draw[black] (O3) -- (D3);
        \draw[black] (A3) -- (E3);
        \draw[black] (B3) -- (F3);
        \draw[black] (C3) -- (G3);
        \draw[below] (\Depth+\DepthTwo+\Sep, 0*\Width, \Height) node{$Z_b$};
        \draw[midway] (\Depth/2+\Depth+\DepthTwo+\Sep,\Width-\Width/2,\Height/2) node {$F_{\mathbb{Z}/2}$};
        
        \draw[->,decorate,decoration={snake,amplitude=.4mm,segment length=2mm,post length=1mm}] (\Depth+\DepthTwo+\Width/4, \Width/2, \Height/2) -- (\Sep+\Depth+\DepthTwo-\Width/4,\Width/2,\Height/2) node[midway, above] {Collapse $z_c$};
        
    \end{tikzpicture}
 \caption{The figure depicts the procedure of bosonization, where $z_c$ is a right $F_{\Spin}$-module. The opposite procedure of fermionization starts with $Z_b$ and inserts $z_c$ as a right $F_{\Z/2}$-module.}
    \label{fig:BosSandwich}
\end{figure}

\begin{defn}[{Freed-Moore~\cite[\S 5.5]{FM06}}]
\label{euler_TFT_defn}
For any $\lambda\in\C^\times$ and $n\ge 0$, define the \term{Euler TFT} $e_\lambda$ to be the invertible, $n$-dimensional topological field theory whose partition function on a closed manifold $M$ is $\lambda^{\chi(M)}$, where $\chi$ denotes the Euler characteristic.
\end{defn}
\begin{rem}
Freed-Moore's definition is equivalent to \cref{euler_TFT_defn}, but more explicit; that \cref{euler_TFT_defn} suffices to define an invertible TFT follows from Freed-Hopkins-Teleman's classification of invertible field theories in terms of Reinhardt bordism invariants~\cite{FHT10} and the fact that the Euler number is a Reinhardt bordism invariant~\cite[Theorem 1]{Rei63}.
\end{rem}
\begin{lem}
The TFTs $z_c\otimes_{F_{\Z/2}} z_c$ and
$z_c\otimes_{F_{\Spin}} z_c$ are both isomorphic to $e_{1/2}$.
\end{lem}
\begin{proof}[Proof sketch]
First, show that $z_c\otimes_{F_{\Z/2}} z_c$ and
$z_c\otimes_{F_{\Spin}} z_c$ are invertible by using a theorem of
Schommer-Pries~\cite[Theorem 11.1]{Sch18} that reduces checking invertibility to checking the values of these TFTs on $S^1$ with either the bounding spin structure or the trivial $\Z/2$-bundle.\footnote{For
2d TFTs, Schommer-Pries' theorem requires an assumption on the tangential structure $\xi\colon B\to B\O(2)$ of the
theory: specifically, if $B$ is connected we need that $S^2$ admits such a structure. For both theories appearing
in this proof, the structure is $B\Spin_2\times B\Z/2\to B\O(2)$; the domain is connected and $S^2$ admits this
structure, so we may use Schommer-Pries' theorem.} To do so, use the description of the state spaces of a finite path integral theory given by Freed-Quinn in~\cite[(2.10)]{FQ:1991bn}.

Recall from the discussion around \cref{Pontr_defn} that the definition of $I_{\U}$ makes sense with an arbitrary injective abelian group in place of $\U$. There is a homotopy equivalence from $\abs{\sAlg_\C^\times}$ to the connective cover of $\Sigma^2
I_{\C^\times}$~\cite[Proposition 4.20]{DG18},\footnote{To obtain $I_{\U}$ instead of $I_{\C^\times}$, one should use a Hermitian analogue of $\sAlg_\C$~\cite[(1.38)]{Fre12}.} so invertible field theories $Z\colon\Bord_2^\xi\to
\sAlg_\C$ are equivalent data to homotopy classes of maps
\begin{equation}
\label{picard-comp}
	\abs{\overline{\Bord_2^\xi}}\to \Sigma^2 I_{\C^\times},
\end{equation}
where $\overline{\cat D}$ denotes the Picard $2$-groupoid completion of $\cat D$. The universal property of $I_{\C^\times}$ is a natural isomorphism $[E, \Sigma^m I_{\C^\times}]\overset\cong\to \Hom(\pi_m(E), \C^\times)$ just as in \cref{Pontr_defn}, so homotopy classes of maps of the form in Eq.~\eqref{picard-comp} (i.e.\ isomorphism classes of 2d
invertible TFTs) are equivalent data to their partition functions. Therefore it suffices to check that the partition functions of
$z_c\otimes_{F_{\Z/2}} z_c$ and $z_c\otimes_{F_{\Spin}}
z_c$ on a closed, connected, oriented surface $\Sigma_g$ of genus $g$ are both equal to $e_{1/2}(\Sigma_g) = 2^{2g-2}$, which can be done using Eq.~\eqref{2d_JW_formulas}.
\end{proof}

\begin{rem}
The spectrum
$\abs{\overline{\Bord_2^\xi}}$ is known to be the Madsen-Tillmann spectrum $\Sigma^2\mathit{MT\xi}$~\\ \cite{GMTW09,
Ngu17, SP17}, so one could alternatively explicitly identify the abelian group of 2d $\Spin\times\Z/2$ invertible
TFTs $\Hom(\pi_2(\Sigma^2\MTSpin_2\wedge (B\Z/2)_+), \U)$ and identify
$z_c\otimes_{F_{\Z/2}} z_c$ and $z_c\otimes_{F_{\Spin}}
z_c$ in that group. Randal-Williams~\cite[Figure 5, left]{RW14} runs
enough of the Adams spectral sequence for $\MTSpin_2$ to show that $\pi_0(\Sigma^2\MTSpin_2)\cong\Z$,
$\pi_1(\Sigma^2\MTSpin_2)\cong\Z/2$, and $\pi_2(\Sigma^2\MTSpin_2)\cong\Z\oplus\Z/2$, and that the map
$\Sigma^2\MTSpin_2\to\MTSpin$ is an isomorphism on $\pi_0$ and $\pi_1$ and surjective on $\pi_2$. This is enough to
set up the Atiyah-Hirzebruch spectral sequence for the $\Sigma^2\MTSpin_2$-homology of $B\Z/2$, i.e.\
$\pi_*(\Sigma^2\MTSpin_2\wedge (B\Z/2)_+)$, and solve it by comparing to the analogous spectral sequence for
$\Omega_*^\Spin(B\Z/2)$, with the conclusion that $\pi_2(\Sigma^2\MTSpin_2\wedge (B\Z/2)_+)\cong
\Z\oplus\Z/2\oplus\Z/2$, with generators $S^2$, $S_{\mathit{nb}}^1\times S_{\mathit{nb}}^1$ with trivial
$\Z/2$-bundle, and $S_{\mathit{nb}}^1\times S_b^1$ with $\Z/2$-bundle pulled back from the nontrivial double cover
on the second factor. Then one could check the isomorphism type of these two TFTs by checking only on these three generators.
\end{rem}

\begin{cor}
\label{2d_bosonization_invertible}
The bosonization of the fermionization of a TFT $Z_b\colon\Bord_2^{\SO\times\Z/2}\to\sAlg_\C$ is isomorphic to
$Z_b\otimes e_{1/2}$, and likewise the fermionization of the bosonization of $\mc{Z}_f\colon\Bord_2^\Spin\to\sAlg_\C$ is isomorphic to
$\mc{Z}_f\otimes e_{1/2}$.
\end{cor}
The factor of $e_{1/2}$ is the analogue of the factor of $1/(2\pi)$ in the Fourier inversion formula.

Tensoring with $z_c$ is an equivalence of
categories --- however, this equivalence is not monoidal. One way to see this is to compare the groups of
invertible objects on the two sides, which are not isomorphic. Instead, bosonization and fermionization behave like the
Fourier transform: the Fourier transform is not a ring homomorphism; rather, it exchanges multiplication on one
side with convolution on the other. There is a symmetric monoidal convolution product defined on 2d $\SO\times\Z/2$ TFTs, and bosonization sends tensor product to convolution~\cite[\S 2.2]{Tho20}.
\begin{rem}[More general tangential structures]
\label{pinm_rspin}
The Arf invariant of spin surfaces generalizes to the Arf-Brown-Kervaire invariant of \pinm surfaces~\cite{Bro71, KT90}, so the theory $z_c$ extends to a theory on \pinm surfaces with a $\Z/2$-bundle. One can therefore make the same definitions to define a bosonization-fermionization correspondence between 2d \pinm TFTs and 2d $\O\times\Z/2$ TFTs, and the analogue of \cref{2d_bosonization_invertible} holds. This is due to Thorngren~\cite[\S2.2]{Tho20}; see Kobayashi~\cite[\S 4]{Kobayashi2019} for an application and~\cite{Ste16, Tur18, MS23} for classification results of 2d \pinm TFTs.

Likewise, an \term{$r$-spin structure} on a surface $\Sigma$ is the tangential structure described by the $r$-fold cover $\SO(2)\to\SO(2)$ (so $1$-spin structures are orientations and $2$-spin structures are spin structures in the usual sense). The Arf invariant extends to $r$-spin surfaces~\cite{GG12, RW14}, so since $r$-spin structures on $\Sigma$ extending a given orientation are a torsor over $H^1(\Sigma; \Z/r)$~\cite[\S 2.3]{RW14}, one can follow a similar line of argument to define a correspondence between 2d $\SO\times\Z/r$ TFTs and 2d $r$-spin TFTs, albeit with some subtleties because for $r>2$ $r$-spin structures do not make sense above dimension $2$. The state spaces of $F_{r\text{-}\Spin}$ were constructed by Runkel~\cite[\S 6]{Run20}. See~\cite{Nov15, Ste16, CS21, LS21, SS22, CMS23, Sze23} for more work on 2d $r$-spin TFT. From the point of view of physics, spin structures are to fermions as $r$-spin structures are to \term{parafermions}, and the parafermionic version of the Jordan-Wigner transform is due to Fradkin-Kadanoff~\cite{FK80}; see also recent work of Radičević~\cite[\S 4.2]{Rad18}, Chen-Haghighat-Wang~\cite{CHW23}, and Duan-Jia-Lee~\cite[\S 3]{DJL23}.

\end{rem}

\subsubsection{Higher dimensions: bosonic shadows and bosonization conjectures}\label{subsubsec:higherd}

There has been a great deal of recent research generalizing the 2d bosonization/fermionization correspondence of the previous subsubsection to higher dimensions. Different generalizations adopt different perspectives; we will use a construction of Gaiotto-Kapustin~\cite{Gaiotto2016} in all dimensions, generalized by Tata-Kobayashi-Bulmash-Barkeshli \cite{Tata2021} to general twisted spin structures in spacetime dimension $4$. We encourage the reader to check out the related but different approaches of \cite{Tho20,Kobayashi:2022qh,Chen:2017fvr}.


%

We start with a fermionic symmetry presented by data $(G_b, s, \omega)$ as in \cref{def:fermionic_symmetry}, and recall from \cref{subsection:fermionicsym} the definitions of $(BG_b, s, \omega)$-twisted spin structures $\xi_{BG_b, s, \omega}$ and $(BG_b, s)$-twisted orientations $\xi_{BG_b, s}$. We will eventually focus on $n=3,4$. The basic story is pretty similar to before: there is a $3$-dimensional kernel theory $z_c$, which is a defect between two $4$-dimensional theories, and the bosonization/fermionization correspondence is implemented by tensoring with $z_c$. However, there are three key differences.
\begin{enumerate}
    \item On the fermionic side, we use $(BG_b, s, \omega)$-twisted spin structures, and the bosonic side must also take this generalization into account.
    \item Instead of using principal $\Z/2$-bundles to build $F_{\Z/2}$, one has to use $\Z/2$ (higher) gerbes. Another way to say this is that the ordinary $\Z/2$ symmetry on the bosonic side of the correspondence is replaced with an $(n-2)$-form $\Z/2$ symmetry, or a $B^{n-2}\Z/2$ symmetry.
    \item Moreover, $z_c$ carries an anomaly with respect to this $B^{n-2}\Z/2$ symmetry, as do the theories on the bosonic side of the correspondence. This means that rather than building $F_{\Z/2}$ by summing the trivial theory over $\Z/2$ higher gerbes, we must sum a nontrivial invertible theory, much like in the finite path integral construction of Dijkgraaf-Witten theory first constructed by Freed-Quinn~\cite{FQ:1991bn} and then generalized and extended in~\cite{Fre94, FHLT:2009qp, Mor15, Tro16, CRS19, SW19, SW20, Har20}.
\end{enumerate}

\begin{rem}
    We do not need a detailed understanding of higher gerbes in this paper; all we need is that for an abelian group $A$, $A$ $\ell$-gerbes are objects that can be defined over a topological space $X$ whose isomorphism classes are in natural bijection with $H^{\ell+1}(X; A)$ and which, like principal bundles, form a sheaf of $\infty$-groupoids, so that they are local objects in the sense of quantum field theory and can be background fields. Moreover, the addition on $H^{\ell+1}(X; A)$ refines to a tensor product $\odot$ on higher gerbes. See Breen~\cite{Bre04, Bre06}, Lurie~\cite[\S 7.2.2]{HTT}, and Nikolaus-Schreiber-Stevenson~\cite{NSS15} for precise definitions and more information.
\end{rem}


Recall that by property \ref{item:lifting}, bordism invariants $a\colon \Omega_n^\xi\to\U$ categorify to invertible TFTs $\alpha_a\colon \Bord_n^\xi\to\fC_k$ such that $a$ is the partition function of $\alpha_a$. Also recall $(BG_b, s)$-twisted orientations (\cref{twor_defn}) and the tangential structure $\xi_{BG_b,s}$ characterizing them. Then, let $\xi_{BG_b,s}\times B^{n-1}\Z/2$ denote the tangential structure which consists of a $\xi_{BG_b, s}$-structure and a map to $B^{n-1}\Z/2$.
\begin{defn}
\label{bos_an}
Given a fermionic symmetry $(G_b, s, \omega)$, let $\alpha_0\colon\Bord_{n+1}^{\xi_{BG_b,s}\times B^{n-1}\Z/2}\to\fC_k$ be the invertible TFT characterized by the property that its partition function is
\begin{equation}\label{bos_an_part}
    (M, f\colon M\to BG)\longmapsto \exp\left(\pi i\int_M\left( \Sq^2(\Lambda) + f^*(\omega)\Lambda\right)\right).
\end{equation}
Here $M$ is a closed $(n+1)$-manifold and $\Lambda$ is a $\Z/2$ $(n-2)$-gerbe, or equivalently a map $M\to B^{n-1}\Z/2$. Isomorphism classes of this data are in natural bijection with classess $[\Lambda]\in H^{n-1}(M;\Z/2)$.
\end{defn}
The theory $\alpha_0$ can be thought of as a ``classical higher Dijkgraaf-Witten theory'', where instead of using a finite group, we use a finite higher group.\footnote{Analogues of Dijkgraaf-Witten theory using more general targets than $BG$ were first studied by Yetter~\cite{Yet93} and Quinn~\cite{Qui95}, with additional work by Porter~\cite{Por98, Por07}, Faria Martins-Porter~\cite{FMP07, FMP23}, Porter-Turaev~\cite{PT08}, Staic-Turaev~\cite{ST10}, Turaev~\cite{Tur10},
Monnier~\cite{Mon15a}, Müller-Woike~\cite{MW20}, Windelborn~\cite{Win20}, Freed-Teleman~\cite{FT22}, Liu~\cite{Liu23}, and Sözer-Virelizier~\cite{SV23}.}

\begin{rem}\label{nvect}
When $n=3$, because the partition function of $\alpha_0$ is the integral of a cohomology class, rather than a more general bordism invariant, it is possible to construct $\alpha_0$ as valued in the $4$-category $\cat{Cat}_\C[E_2]$ of braided monoidal, $\C$-linear $2$-categories, the same target as the Crane-Yetter theory. This is because $\Sigma^4 H\C^\times$ is the $3$-connective cover of $\abs{\cat{Cat}_\C[E_2]^\times}$, e.g.\ because $\Omega^2(\cat{Cat}_\C[E_2])\simeq\Vect_\C$ and $\abs{\Vect_\C^\times}\simeq \Sigma H\C^\times$, so precomposing the $3$-connective covering map $\tau_{\ge 3}\colon \Sigma^4 H\C^\times\to \abs{\cat{Cat}_\C[E_2]^\times}$ with the map $\Sigma^4 H\Z/2\to \Sigma^4 H\C^\times$ given by the unique injective group homomorphism $\Z/2\hookrightarrow \C^\times$, we obtain an invertible field theory $\varepsilon$ valued in $\cat{Cat}_\C[E_2]$ from the degree-$4$ mod $2$ cohomology class defining $\alpha_0$. Moreover, $\omega \in H^2(BG_b; \Z/2)$ and the $\Z/2$ $2$-gerbe $\Lambda$ can be thought of as the generator of $H^2(B^2\Z/2; \Z/2)\cong \Z/2$, and the partition function given by the formula in \eqref{bos_an_part} can be thought of as the class
\begin{equation}
    \Sq^2(\Lambda) + \omega \cup \Lambda\in H^4(BG_b\times B^2\Z/2; \Z/2).
\end{equation}
\end{rem}

\begin{defn}
Let $F_B\colon \Bord_{n+1}^\SO\to\fC_k$ denote the quantum Dijkgraaf-Witten theory obtained from $\alpha_0$ by using the finite path integral to sum over $\Z/2$ $(n-2)$-gerbes.
\end{defn}

Like in \cref{nvect}, for $\fC_k$ the existence of the finite path integral is a hypothesis, but one can construct $F_B$ as a TFT with target $\mathit{n\Vect}$ independently of our hypothesis.

In Gaiotto-Kapustin's version of bosonization/fermionization, both the kernel theory $z_c$ and the bosonic theory $Z_b$ are anomalous with anomaly $\alpha_0$; that is, they are boundary theories for $F_B$. Our next step is to define $z_c$.

For the rest of this section we fix a choice of specific spaces in the homotopy types $B\O$ and $BG_b$. For $H\in\set{\O, G_b}$, use the geometric realization of the nerve of the topological category with a single object $*$ and $\Hom(*, *)\cong H$. Fix cocycles $W_1\in Z^1(B\O;\Z/2)$ and $W_2\in Z^2(B\O;\Z/2)$ representing the cohomology classes $w_1$, resp.\ $w_2$, and cocycles $S\in Z^1(BG_b;\Z/2)$ and $\Omega\in Z^2(BG_b;\Z/2)$ representing $s$, resp.\ $\omega$.


\begin{defn}[Gaiotto-Kapustin \cite{Gaiotto2016}]
\label{zc_defn}
Let $M$ be an $n$-manifold, possibly with boundary, together with data of
\begin{enumerate}
    \item a class $\Lambda\in Z^{n-1}(M; \Z/2)$,
    \item a map $f\colon M\to BG_b$,
    \item a spin structure $\zeta$ on $f^*(V)\oplus M$.
\end{enumerate}

In particular, we need to choose a map $\gamma \colon M \rightarrow BO$ representing the classifying map of $TM$, and the spin structure $\zeta$ contains the following two pieces of data,
\begin{enumerate}
    \item a cochain $\chi\in C^0(M;\Z/2)$ such that $\delta(\chi) = \gamma^*(W_1) + f^*(S)$, and 
    \item a cochain $\eta\in C^1(M;\Z/2)$ such that $\delta(\eta) = \gamma^*(W_2) + \gamma^*(W_1)\cup\gamma^*(S) + f^*(\Omega)$.
\end{enumerate}

Given these data, define
\begin{equation}\label{def:zc}
    z_c(M, f, \zeta; \Lambda) \coloneqq \sigma(M, \Lambda)(-1)^{\int_M \eta\smile\Lambda}\in\U,
\end{equation}
where $\sigma(M, \Lambda)$ denotes the \term{Gu-Wen Grassmann integral} of $\Lambda$ on $M$, whose explicit definition at the level of cocycles in given by Tata~\cite[\S 7]{Tata2020}.\footnote{See also~\cite{CR07, Cim09, GW14, Gaiotto2016, Kobayashi2019} for other works giving special cases of this definition and~\cite[\S IV]{Tata2021} for a comparison of different definitions.}

\end{defn}

\begin{rem}
A physical interpretation of $z_c(M, f, \zeta; \Lambda)$ can be given as follows. The vector space of a spin TFT associated with some $d$-dimensional manifold is $\mathbb{Z}/2$-graded, with the grading denoting the fermion number. More precisely, $\Lambda$ represents the grading of the vector spaces associated with every $(d-1)$-cycle in the following way: if the integration of $\Lambda$ on the $(d-1)$-cycle is 0 (or 1), the vector space associated with the $(d-1)$-cycle should have grading 0 (or 1). The Gu-Wen Grassmann integral represents the path integral of Grassmann variables on the total space. Different spin-structure $\zeta$ amounts to different order of integration of these Grassmann variables, which gives the extra factor $(-1)^{\int_M \eta\smile\Lambda}$ originating from the sign corresponding to the exchange of two fermionic operators. 
\end{rem}

%

The Gu-Wen Grassmann integral $\sigma(M, \Lambda)$ takes value in the abelian group of $4^{\mathrm{th}}$ roots of unity, i.e., $\{\pm 1, \pm i\}$. Specifically, if $M$ is orientable, $\sigma(M, \Lambda)$ is valued in $\set{\pm 1}$, but if $M$ is unorientable, $\sigma(M, \Lambda)$ can also take values in $\set{\pm i}$. When performing explicit calculations, it is sometimes more illuminating to use the Poincar\'e dual of $\Lambda$, which we denote as $\mathfrak{L}\in Z_1(M;\Z/2)$. 

Note that if one replaces $W_2$ with a different representative $W_2'$, then there is some cochain $A\in C^1(B\SO;\Z/2)$ with $\delta(A) = W_2 - W_2'$. One can then replace the data $(\Lambda, \gamma, \zeta)$ with $(\Lambda, f, \zeta + \gamma^*(A))$ to obtain the same value of $z_c$. Thus the choice of $W_2$, while important in order to have a definition, does not affect what follows; likewise for the chain-level choices of $W_1$, $W$, and $S$. 

The data $(\gamma, f, \chi, \eta)$ in \cref{zc_defn} induce a $(BG_b, s, \omega)$-twisted spin structure on $M$ refining the $(BG_b, s)$-twisted orientation picked out by $\gamma$ and $\chi$, by providing a trivialization of the cohomology classes $w_1(M) + f^*(s)$ and $w_2(M) + w_1(M)s + f^*(\omega)$. Moreover, if $M$ admits a $(BG_b, s, \omega)$-twisted spin structure, then the map from $(f, \eta)$ to the corresponding $(BG_b, s,\omega)$-twisted spin structure defines a homotopy equivalence between the space of data $(f, \eta)$ and the space of $(BG_b, s,\omega)$-twisted spin structures on $M$ refining the $(BG_b, s)$-twisted orientation defined by $\gamma$ and $\chi$. Similarly, there is a canonical homotopy equivalence between the space of data $(\Lambda, \gamma, f, \chi, \eta)$ on $M$ and the space of pairs of $(BG_b, s,\omega)$-twisted spin structures refining the $(BG_b, s)$-twisted orientation defined by $(\gamma, \chi)$ and a $\Z/2$ (higher) gerbe $\Lambda$. These equivalences are compatible with taking boundaries, meaning that we have data of a homotopy equivalence from the $n$-dimensional bordism category of manifolds with data of $(\Lambda, f, \gamma, \chi, \eta)$ to $\Bord_n^{\xi_{BG_b,s,\omega}\times B^{n-1}\Z/2}$, and we can ask whether the function $z_c$ is the partition function of a TFT.

Unfortunately, $z_c$ is not invariant enough to be a partition function. Given two equivalent data $(\Lambda_0, f_0,\gamma_0, \chi_0, \eta_0)$ and $(\Lambda_1, f_1, \gamma_1, \chi_1, \eta_1)$ we think of a path between these data as the data $(\overline \Lambda, \overline f, \overline \gamma, \overline\chi, \overline\eta)$ such that the pullback of this data to $M\times\set i$ is $(\Lambda_i, f_i,\gamma_i, \chi_i, \eta_i)$. Using this extension to $M\times[0, 1]$, which we think of as a bordism between our two choices of data,
Tata-Kobayashi-Bulmash-Barkeshli~\cite{Tata2021} show by direct computation that going from the data on $M\times\set 0$ to the data on $M\times\set 1$, $z_c$ is multiplied by
\begin{equation}
\label{anomaly_key}
    (-1)^{\int_{M\times [0,1]}\left(\overline f^*(\omega)\overline \Lambda + \Sq^2(\overline \Lambda)\right)}.
\end{equation}
When $G_b$ is the trivial group, this was previously shown by Gaiotto-Kapustin~\cite{Gaiotto2016}.

This bordism is pictured in \cref{fig:phase_of_z}.

\begin{figure}[ht]
\centering
    \begin{tikzpicture}[thick]
    
        \def\Depth{4}
        \def\Height{2}
        \def\Width{2}
        \def\Sep{3}        
        
        \coordinate (O) at (0,0,0);
        \coordinate (A) at (0,\Width,0);
        \coordinate (B) at (0,\Width,\Height);
        \coordinate (C) at (0,0,\Height);
        \coordinate (D) at (\Depth,0,0);
        \coordinate (E) at (\Depth,\Width,0);
        \coordinate (F) at (\Depth,\Width,\Height);
        \coordinate (G) at (\Depth,0,\Height);
        \draw[black] (O) -- (C) -- (G) -- (D) -- cycle;
        \draw[black] (O) -- (A) -- (E) -- (D) -- cycle;
        \draw[black,fill=green!20,opacity=0.8] (O) -- (A) -- (B) -- (C) -- cycle;
        \draw[black,fill=red!20,opacity=0.8] (D) -- (E) -- (F) -- (G) -- cycle;
        \draw[black] (C) -- (B) -- (F) -- (G) -- cycle;
        \draw[black] (A) -- (B) -- (F) -- (E) -- cycle;
        \draw[left] (0, 0*\Width, \Height) node{$z_c(M,f_0,\zeta_0;\Lambda_0)$};
        \draw[right] (\Depth, 0*\Width, \Height) node{$\,\,z_c(M,f_1,\zeta_1;\Lambda_1)$};
        \draw[above] (\Depth/2,\Width+\Width/4,\Height/2) node {$M\times[0,1]$};
        \draw[midway] (\Depth/2,\Width-\Width/2,\Height/2) node [scale=0.8]{$(-1)^{\int_{M\times [0,1]} \overline{f}^*(\omega)\overline{\Lambda}+ \Sq^2(\overline{\Lambda})}$};
        
        
        \coordinate (OT) at (\Sep+\Depth,0,0);
        \coordinate (AT) at (\Sep+\Depth,\Width,0);
        \coordinate (BT) at (\Sep+\Depth,\Width,\Height);
        \coordinate (CT) at (\Sep+\Depth,0,\Height);
    \end{tikzpicture}
    \caption{The red and green faces show the two partition functions are off by a phase when traversing the path that connects the two equivalent sets of data $(f_0,\zeta_0;\Lambda_0)$ and $(f_1,\zeta_1;\Lambda_1)$.}
    \label{fig:phase_of_z}
\end{figure}
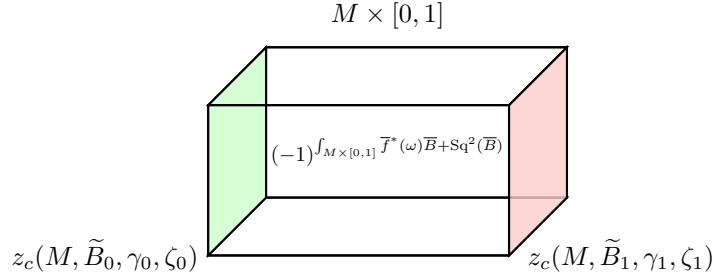

So $z_c$ is, like many irregular verbs, inconsistent in a consistent way. Recall Freed-Quinn's construction \cite{FQ:1991bn} of the state spaces of the Dijkgraaf-Witten theory $F_B$, in which one introduces a groupoid $\mathcal G(M)$ of data of cocycles representing a cohomology class, cycles representing the fundamental class, etc., and shows that the state spaces of the classical theory $\alpha_0$ are a line bundle over $\mathcal G(M)$, and that the state spaces of the quantum theory $F_B$, where we have summed over the classes $B$, are the sections of the line bundle. 

The conclusion is that a complex number that is an invariant of $(f, \zeta; \Lambda)$ but transforms as~\eqref{anomaly_key} is actually an element of the state space $F_B(M)$. That is:
\begin{prop}[{Gaiotto-Kapustin~\cite{Gaiotto2016}, Tata-Kobayashi-Bulmash-Barkeshli~\cite{Tata2021}}]
\label{it_is_anomalous}
$z_c$, regarded as a nonextended\footnote{The proposition is likely true with ``nonextended'' removed, but this is open: to our knowledge, $z_c$ has not been studied much in higher codimension.} spin theory, has the structure of a boundary theory for $F_B$.
\end{prop}
Gaiotto-Kapustin and Tata-Kobayashi-Bulmash-Barkeshli do not phrase their results in this way, but \cref{it_is_anomalous} follows from what they prove. In particular, their arguments apply to the case when $M$ has nonempty boundary.


Let $F_C$ be the result of applying the finite path integral to sum over $(BG_b, s,\omega)$-twisted spin structures with fixed $(BG_b, s)$-twisted orientation. Just like in the 2d case, we can encode the dependence of $z_c$ on the $(BG_b,s,\omega)$-twisted spin structure into the statement that $z_c$ is, as an oriented theory, an $(F_C, F_B)$-defect, and then bosonization and fermionization are hardly different from \cref{2d_BF_defn}.
\begin{defn}[{Tata-Kobayashi-Bulmash-Barkeshli~\cite[\S VI, \S VII]{Tata2021}}]\label{def:ZfZb}\hfill
\begin{enumerate}
    \item Let $\mc{Z}_f$ be an $n$-dimensional TFT on manifolds with a $(BG_b,s,\omega)$-twisted spin structure. The \term{bosonic shadow} or \term{bosonization} of $\mc{Z}_f$ is $Z_b\coloneqq \mc{Z}_f\otimes_{F_C} z_c$, which is a boundary theory for $F_B$.
    \item Let $Z_b$ be a boundary theory for $F_B$; then its \term{fermionization} is $\mc{Z}_f\coloneqq Z_b\otimes_{F_B} z_c$, which is an $n$-dimensional $(BG_b,s,\omega)$-twisted spin TFT.
\end{enumerate}
\end{defn}
So bosonization and fermionization exchange (nonanomalous) $(BG_b,s,\omega)$-twisted spin TFTs with TFTs with a $(BG_b,s)$-twisted orientation, an $(n-2)$-form $\Z/2$ symmetry, and the specific anomaly theory $\alpha_0$. In physics, we also say that bosonization is the process of keeping $G_b$ as the background gauge field while ``summing over spin structures'', which generates an extra $(n-2)$-form $\Z/2$ symmetry. 

\begin{rem}
In dimension $2$, $z_c$ is an invertible TFT, and one naturally wonders whether this is true in all dimensions. Because $z_c$ is not defined absolutely, but only relative to $\alpha_0$,\footnote{We have considered $z_c$ both as a boundary theory of $F_B$ for manifolds with a $(BG_b, s)$-twisted orientation, and as a boundary theory of $\alpha_0$ for manifolds with a $(BG_b, s)$-twisted orientation and a map to $B^{n-1}\Z/2$. To clarify, these two perspectives are equivalent, by summing over the maps to $B^{n-1}\Z/2$; though we have mostly thought of $z_c$ as an $F_B$-boundary, in this remark we will think of it as an $\alpha_0$-boundary.} it is less clear how to define invertibility, because if $M$ and $N$ are two boundary theories to the same theory $Z$, $M\otimes N$ is a $(Z\otimes Z)$-boundary condition, and we need extra data to obtain an absolute theory.

However, because the bordism invariant used to define $\alpha_0$ has order $2$, there is an equivalence $\alpha_0\otimes\alpha_0\cong F_{\mathrm{triv}}$; after choosing such an equivalence, the tensor product of two $\alpha_0$ boundary theories $M$ and $N$ becomes an absolute TFT (i.e.\ it is a boundary theory of the trivial theory), and therefore we may ask whether $M\otimes N$ is trivial. For $M = N = z_c$, ultimately because $z_c$ is built from $\Z/2$-valued cocycles, and so in this sense $z_c$ is an invertible (anomalous) field theory.

Some other works have studied invertible boundary theories, including~\cite{Ina21, que21}; see also~\cite[\S 4.1]{Etingof2009}, \cite[Definition 1.3.1]{DY23}, and~\cite[\S 4.1]{Dec23}.
\end{rem}
\begin{rem}[More general tangential structures]
In addition to the generalization in Tata-Kobayashi-Bulmash-Barkeshli that we have just surveyed, several other works have studied generalizations of Gaiotto-Kapustin's construction to other tangential structures.
Bhardwaj~\cite[\S 3.3]{Bha17} studies bosonization of 3d \pinp TFTs, and Kobayashi~\cite{Kobayashi2019, Kobayashi:2022qh} generalizes to both \pinp and \pinm TFTs in all dimensions. Gukov-Hsin-Pei~\cite[\S 6]{GHP21}, Hsin-Ji-Jian~\cite[\S 5]{HJJ22}, and Kobayashi~\cite{Kob22} study analogues of bosonization and fermionization for field theories on manifolds with ``Wu structure,'' i.e.\ a trivialization of a Wu class.
\end{rem}


\subsubsection{Adding symmetries and the conjecture}\label{bos_conj_sec}
We are interested in using the bosonic shadow procedure to compute anomaly field theories, essentially by reducing the more complicated fermionic case to the better-understood bosonic case. In this subsubsection, we will often say ``category'' when referring to $k$-categories; whenever we do this, the value of $k$ will either be clear or can be understood from context.


Let $\widetilde\alpha\colon\Bord_4^{\xi_{BG_b,s,\omega}}\to\fC_k$ be an invertible TFT, and suppose that we have data of a trivialization $\tau$ of the restriction
\begin{equation}
    \widetilde\alpha|_{\Spin}\colon \Bord_4^\Spin\longrightarrow \Bord_4^{\xi_{BG_b,s,\omega}}\overset{\widetilde\alpha}{\longrightarrow}\fC_k,
\end{equation}
where the first map is induced by regarding a spin structure as a $(BG_b, s, \omega)$-twisted spin structure with trivial $G_b$-bundle.
Change of tangential structure induces a forgetful functor $\Phi_{\widetilde\alpha}$ from the category of boundary theories to $\widetilde\alpha$ to the category of boundary theories to $\widetilde\alpha|_{\Spin}$ --- which, thanks to $\tau$, is the category of $3$-dimensional spin TFTs.

Let $\mc{Z}_f$ be an $3$-dimensional spin TFT, and suppose that $\mc{Z}_f$ is in the image of $\Phi_{\widetilde\alpha}$, i.e.\ that $\mc{Z}_f$ can be extended to an anomalous TFT on $(BG_b,s,\omega)$-twisted spin manifolds with anomaly $\widetilde\alpha$. Choose such an extension $\widetilde{\mc{Z}}_f$ of $\mc{Z}_f$. If $F_{\widetilde\alpha}\colon \Bord_n^{\xi_{BG_b,s}\times B^{n-1}\Z/2}\to\fC_k$ denotes the theory obtained from $\widetilde \alpha$ by summing over twisted spin structures with fixed principal $G$-bundle, then the $F_\Spin$-boundary theory $\mc{Z}_f$ extends to the $F_{\widetilde\alpha}$-boundary theory $\widetilde{\mc{Z}}_f$.

Let $\widetilde\beta$ be the bosonization of $\widetilde\alpha$ in the sense of \cref{def:ZfZb}, and let $F_{\widetilde\beta}$ be the theory produced by summing $\widetilde\beta$ over $\Z/2$ $2$-gerbes.
We would like to say ``the anomaly of the bosonization is the bosonization of the anomaly''. One might hope to make that precise by asking that $z_c$ extends from an $(F_\Spin, F_B)$-defect to an $(F_{\widetilde\alpha}, F_{\widetilde\beta})$-defect. However, $F_{\widetilde\beta}$ is anomalous: it is a bosonization, so carries the anomaly theory $\alpha_0$ from \cref{bos_an}. $F_{\widetilde\alpha}$ does not carry this anomaly.
This means that without some kind of additional data, it does not make sense to refer to $(F_{\widetilde\alpha}, F_{\widetilde\beta})$ defects: see \cref{defective}. By treating the combined module ${}_{\tilde{z}_{c}}(F_{\widetilde{\beta}})$ as a single unit, one can couple to it $F_{\widetilde{\alpha}}$ in such a way that the anomaly only residing on $F_{\widetilde{\beta}}$ is trivialized by $\tilde{z}_c$. In fact, $\widetilde{z}_{c}$ must contain an anomaly because the entire system of $F_{\widetilde{\alpha}} \boxtimes_{\widetilde{z}_c} F_{\widetilde{\beta}}$ must in the end be nonanomalous. 

\begin{figure}[h!]
\begin{tikzpicture}
    \fill[MidnightBlue!20!white] (0, 4) rectangle (4, 0);
    \draw[very thick, BrickRed] (-4, 0) -- (0, 0);
    \draw[very thick, MidnightBlue] (4, 0) -- (0, 0);
    \fill (0, 0) circle (0.6mm); 
    \node[below] at (-2, 0) {$F_{\widetilde\alpha}$};
    \node[below] at (2, 0) {$F_{\widetilde\beta}$};
    \node at (2, 2) {$\alpha_0$};
    \node at (-2, 2) {$\textbf{1}$};
    \node[below] at (0, 0) {$\widetilde z_c$};
    \node[rotate=90] at (0, 2) {Problem here!};
\end{tikzpicture}
\caption{We want to state \cref{upstream_conjecture} implements the slogan ``the bosonization of the anomaly is the anomaly of the bosonization,'' but since bosonizations typically have nontrivial anomalies, this cannot be done naïvely: one needs extra data to reconcile the bulk theories $\textbf{1}$ and $\alpha_0$, respectively the (trivial) anomaly of $F_{\widetilde \alpha}$ and the anomaly of $F_{\widetilde\beta}$.}
\label{defective}
\end{figure}
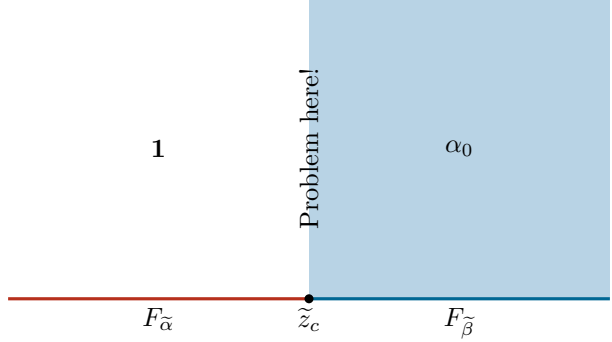

To solve this, we do something which may look odd: we introduce a $(BG_b, s, \omega)$-twisted spin structure as an additional background field. The partition function~\eqref{bos_an_part} vanishes on $(BG_b, s,\omega)$-twisted spin manifolds --- in fact, the $(BG_b, s, \omega)$-twisted spin structure provides a trivialization and therefore trivializes the anomaly. Though it may seem strange to introduce a twisted spin structure after bosonizing and getting rid of twisted spin structures, this is fine from the point of view of calculating an anomaly indicator on a $4$-manifold $X$: $X$ already has a twisted spin structure, so $\alpha_0$ is trivial on $X$ and therefore using this procedure works.
\begin{conj}[Bosonization conjecture]
\label{upstream_conjecture} 
Let $z_c$ denote Gaiotto-Kapustin's $3$-dimensional $(F_\Spin, F_B)$ defect (i.e.\ Tata-Kobayashi-Barkeshli-Bulmash's construction for $G_b = 1$). Then, \emph{as theories of $(BG_b, s,\omega)$-twisted spin manifolds}, $z_c$ canonically extends to an $(F_{\widetilde\alpha}, F_{\widetilde\beta})$ defect $\widetilde z_c$.
\end{conj}
The importance of $\tilde{z}_c$ will be the main focus of the next subsection. Its existence is crucial for realizing the anomaly of a fermionic topological order as a cobordism invariant. A physically intuitive motivation for the existence of $\tilde{z}_c$ can be given as follows: consider the bulk boundary system involving the invertible theory labeled by $\widetilde \beta$ and $\mathcal{T}_b$, and gauge the diagonal $G_b$ symmetry to obtain $\frac{\mathcal{T}_b \times \widetilde \beta}{G_b}$. The bulk is then given by a so called ``gauged SPT'' associated to $\widetilde \beta$, and is formally presentable as a Dijkgraaf-Witten theory with action given by $\widetilde \beta$. One can do the same on the side of $\mathcal T_f$ for $\widetilde{\alpha}$. From the point of view of the gauged SPTs, the theory $\tilde{z}_c$ can be defined as a gapped  domain wall theory that interpolates between the two gauged SPTs.

\subsubsection{Invertibility of the fermionic anomaly from bosonization}\label{subsection:conjforInv}

In this sub-subsection, we want to establish that the theory $\widetilde{\alpha}$ is actually an invertible theory and encodes the anomaly of the fermionic topological order through the lens of the bosonization conjecture. 

For bosonic topological orders, according to \cite{brochier2021invertible}, a fully nondegenerate braided fusion category is an invertible element in the Morita 4-category of braided fusion categories denoted $\cat{Mor}_2(2\Vect)$. Via the cobordism hypothesis, this gives rise to an invertible 4-dimension topological field theory and encodes the anomaly of bosonic topological orders. This invertible field theory gives a bordism invariant for oriented 4-manifolds, and is conjectured to coincide with the Crane-Yetter theory \cite{Crane1993}, as discussed in \cite{Brochier18,brochier2021invertible}.\footnote{Invertibility of the Crane-Yetter theory directly from its state-sum construction is established in unpublished works of Freed-Teleman \cite{FreedLecturea,FreedLectureb} and Walker \cite{Walker}, as well as \cite{Sch18}.}

When $\widetilde{\alpha}$ is the associated anomaly theory of a fermionic topological order $\mathscr{T}_f$, invertibility under the tensor product given by stacking is not obvious at first glance, because $\widetilde{\alpha}$ is the anomaly for a \emph{slightly degenerate} braided fusion category. In principle, one would want to generalize the result in the nondegenerate case by checking if slightly degenerate braided fusion categories are invertible in  $\cat{Mor}_2(2\sVect)$, this category is constructed rigorously in \cite{Decoppet:2024htz}.\footnote{2$\Vect$ and $2\sVect$ are the 2-categories $\Sigma(\Vect)$, resp.\ $\Sigma(\sVect)$, the Karoubi completions of the $2$-categories of $\Vect$-, resp.\ $\sVect$-module categories, as defined in~\cite[\S 1.4]{Gaiotto2019}.} This would imply that a super MTC gives rise to a fully-extended invertible 4d spin TFT, which can be regarded as a ``spin Crane-Yetter" theory. One can then use this to construct $\widetilde{\alpha}$ as an invertible fermionic TFT.
Instead of taking this approach, we will elaborate on how \cref{upstream_conjecture} can be used as credence for the fact that $\widetilde{\alpha}$ is invertible, but we do not prove this. 

We begin at the level of $z_c$ which maps between $\mathscr{T}_f$ and its bosonized theory $\mathscr{T}_b$; this is the bottom map in \cref{fig:alphabeta}. While $z_c$ is a $(F_{\Spin},F_{\Z/2})$ bimodule, it can be described by a spin-$\Z/2^f$ gauge theory in (3+1)-d. Its
dynamical field is a $\Z/2$-valued $1$-cocycle $\eta$ that solves $\delta\eta = w_2$, aka a spin structure.

 We will use the language of fusion 2-categories for the purpose of discussing invertibility of the anomaly of the (2+1)-d fermionic TFT. The fusion 2-category has the right property to serve as a bulk theory for the fermionic TFT in consideration, and it cures the slight non-degeneracy on the boundary. In the following definition, we define the \term{fusion $2$-category} corresponding to a particular 4d TFT. The notion of fusion $2$-categories is due to Douglas-Reutter~\cite[Definition 2.1.6]{DR18}; see there for the definition.
\begin{defn}[{\cite[\S 2.2]{Johnson-Freyd:2020twl}}]
\label{theory_S}
    The (3+1)-d spin-$\Z/2^f$ gauge theory, denoted by $\mathscr{S}$, is a nondegenerate braided fusion 2-category with two components: the component of the identity is given by $2\sVect$, and a \textit{magnetic} component that contains two objects.  
    \begin{itemize}
        \item The surface operators $\mathbf{1}$ and $c$ form the identity component. $c$ is called the \textit{Cheshire} string or Kitaev chain as it is the \textit{condensation} of the fermion $\psi$ in $\sVect$ along a surface, and satisfies fusion rule $c^2 \cong \mathbf{1}$. $\mathbf{1}$ is the condensation of the vacuum 1 in $\sVect$ along a surface. 
        \item The non-identity component has a magnetic object $m$, which is required for detecting $c$, and another object $m' = m\otimes c$. Under fusion the magnetic object obeys $m^2\cong \mathbf{1}$. 
    \end{itemize}
\end{defn}

\begin{rem}
There is an abundance of references for fusion 2-categories. See \cite{DY23,Johnson-Freyd:2021tbq} for a further summary and applications of the Cheshire object, as well as \cite{Gaiotto2019,Kong:2024ykr} for reviews on condensation in higher categories. Moreover, the (3+1)-d spin-$\Z/2^f$ gauge theory $\mathscr{S}$ has been explored in great detail in physics literature, which goes under the name of ``(3+1)-d fermionic $\Z/2$ gauge theory'' or ``(3+1)-d fermionic toric code'' \cite{Hansson2004,Burnell2014,Chen2021,Fidkowski2021,Barkeshli2023}.
\end{rem}

\noindent The potential degeneracy of $\mathscr{T}_f$ is cured when coupled to $z_c$ because the line $\psi$ condenses in the bulk to $c$, and $m$ only existed in the bulk.

There is a choice of isomorphism between $Z_{(1)}(\cat{Mod}(\mathscr{T}_f))$, the Drinfel'd center of $\cat{Mod}(\mathscr{T}_f)$, and $\mathscr{S}$ which corresponds to a choice of minimal modular extension.
For each choice of isomorphism of $Z_{(1)}(\cat{Mod}(\mathscr{T}_f)) \cong \mathscr{S}$
one can assign to $\mathscr{T}_f\otimes_{F_{\Spin}} z_c = \mathscr{T}_b$ its anomaly theory $\widetilde{\beta}$; this is moving into the bulk in \cref{fig:alphabeta}. From the point of view of the tensor product on the side of $\widetilde{\alpha}$, the theory $\widetilde{\beta}$ is not invertible. As we have been thinking of bosonization as an analogue of the Fourier transform, invertibility with respect to the tensor product ought to be exchanged with invertibility with respect to some sort of convolution.
\begin{defn}[Convolution kernel]
Suppose $n = 2m$ and define $\kappa_{\mathrm{conv}}\colon\Bord_n^{\O\times B^m\Z/2\times B^m\Z/2}\to\fC_k$ to be the invertible TFT characterized by the property that its partition function on a closed $n$-manifold $M$ with $\Z/2$ $(m-2)$-gerbes $P$ and $Q$ is
\begin{equation}\label{kconv}
    \kappa_{\mathrm{conv}}(M, P, Q) \coloneqq \exp\paren*{\pi i\int_M [P]^2 + [P][Q]},
\end{equation}
where $[P]\in H^m(M;\Z/2)$ denotes the cohomology class classifying the higher gerbe $P$.\footnote{Since this invertible TFT was defined as the integral of a cohomology class, an analogue of \cref{nvect} applies to the construction of $\kappa_{\mathrm{conv}}$ and the convolution product.}
\end{defn}
\begin{defn}[{Thorngren~\cite[(2.27)]{Tho20}}]
\label{convolution_defn}
With $m$ and $n$ as above, let $Z_1$ and $Z_2$ be TFTs of manifolds with a $\Z/2$ $(m-2)$-gerbe. The \term{convolution} of $Z_1$ and $Z_2$ is the TFT 
\begin{equation}Z_1\star Z_2\colon\Bord_n^{\O\times B^m\Z/2}\to\fC_k\,,
\end{equation}
defined by summing the TFT
\begin{equation}
    W(M, P, Q) \coloneqq Z_1(M, P) \otimes Z_2(M, P\odot Q) \otimes \kappa_{\mathrm{conv}}(M, P, Q)
\end{equation}
over the first $(m-2)$-gerbe $P$, where $\odot$ is the tensor product of gerbes, which adds their cohomology classes.
\end{defn}
\begin{exm}
The partition function of $Z_1\star Z_2$ on an $n$-manifold $M$ with $(m-2)$-gerbe $Q$ is
\begin{equation}
    (Z_1\star Z_2)(M, Q) \coloneqq \chi_\infty(\cat{Gerbe}^{m-2}_{\Z/2}(M)) \,\smashoperator{\sum_{P\in\cat{Gerbe}^{m-2}_{\Z/2}(M)}}\, Z_1(M, P) Z_2(M, P\odot Q) \exp\paren*{\pi i \int_M [P]^2 + [P][Q]},
\end{equation}
where $\cat{Gerbe}^{m-2}_{\Z/2}(M)$ denotes the $m$-groupoid of $(m-2)$-gerbes, generalizing the $m = 1$ case of groupoid of principal $\Z/2$-bundles. The function $\chi_\infty$ is the \term{(higher) groupoid cardinality} of a higher groupoid; see, for example,~\cite{Qui95, Lei08, BHW10}.
\end{exm}
\begin{rem}[Symmetric monoidality of the convolution product]
We predict, but do not attempt to prove, that the convolution product extends to a symmetric monoidal structure on the category of TFTs on $n$-manifolds with $\Z/2$ $(m-2)$-gerbes, and that the bosonization functor should admit a symmetric monoidal structure with respect to the usual tensor product on the fermionic side and the convolution product on the bosonic side. (For the latter conjecture, one must address somehow the anomaly on the bosonic side (\cref{bos_an}).)
\end{rem}
\begin{rem}[Generalizations of \cref{convolution_defn}]
One can straightforwardly generalize \cref{convolution_defn} to other finite cyclic groups $A$, albeit restricted to oriented manifolds with higher $A$ gerbes to ensure the integral in~\eqref{kconv} is defined. It would be interesting to explore generalizations to (finite) higher abelian groups, similarly to the perspectives taken by Freed-Teleman~\cite[\S 9]{FT22}, Liu~\cite{Liu23}, and Freed-Moore-Teleman~\cite[\S 3.5]{FMT22} on electric-magnetic duality. It would also be interesting to generalize \cref{convolution_defn} to odd-dimensional TFTs.
\end{rem}
In coupling to a bulk $z_c$ we have gained nondegeneracy at the price of the product structure on $\widetilde{\beta}$ changing. This leads us to conjecturing invertibility of $\widetilde{\alpha}$ as a step in the bosonization conjecture:
\begin{conj}\label{conj:inv}
The data of $\widetilde{\beta}$ being invertible under the convolution product is ported through the conjectured existence of $\widetilde{z}_c$ to $\widetilde{\alpha}$ which is invertible with respect to the regular tensor product.
\end{conj}
 The heart of \cref{upstream_conjecture} is that composing the maps 
$$
\begin{tikzcd}
    \mathscr{T}_f \arrow[r,"z_c"] & \mathscr{T}_b  \arrow[r] & \widetilde{\beta}  \arrow[r,"\widetilde{z}_c",dotted] &\widetilde{\alpha}
\end{tikzcd}
$$
should be equivalent to going into the bulk for $\mathscr{T}_f$ immediately. Therefore, we phrase it as a slogan ``the bosonization of the anomaly is the anomaly of the bosonization''. This means that traversing  between $\mathscr{T}_f$ and $\mathscr{T}_b$  in \cref{fig:alphabeta} automatically maps the corresponding content of bulks, and vice versa when mapping between $\widetilde{\alpha}$ and $\widetilde{\beta}$ for the boundaries.
 Therefore, to find the anomaly field theory of a $(BG_b,s,\omega)$-twisted spin extension of a super MTC, we bosonize, and compute the bosonic anomaly indicator of the corresponding $(BG_b,s)$-twisted oriented theory with its $B\Z/2$ symmetry. Once one has calculated the bulk theory for the bosonic theory, upon fermionizing one obtains the anomaly field theory of the original super MTC. This is the strategy we will use to find anomaly indicators of spin TFTs.

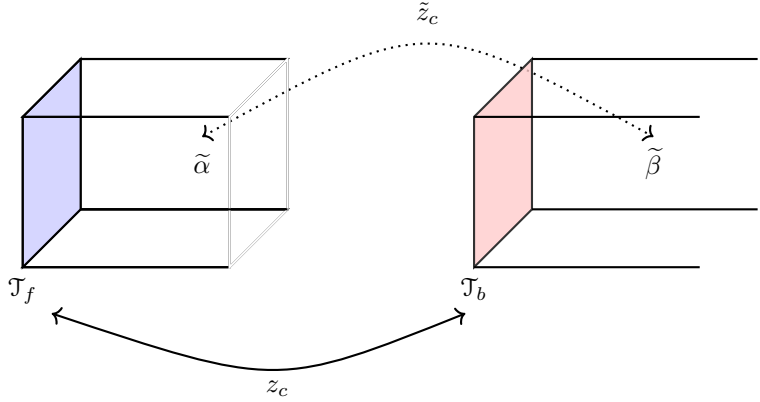
\begin{figure}
    \centering
     \begin{tikzpicture}[thick]
    
        \def\Depth{4}
        \def\DepthTwo{3}
        \def\Height{2}
        \def\Width{2}
        \def\Sep{3}        
        
        \coordinate (O) at (0,0,0);
        \coordinate (A) at (0,\Width,0);
        \coordinate (B) at (0,\Width,\Height);
        \coordinate (C) at (0,0,\Height);
        \coordinate (D) at (\Depth-1.25,0,0);
        \coordinate (E) at (\Depth-1.25,\Width,0);
        \coordinate (F) at (\Depth-1.25,\Width,\Height);
        \coordinate (G) at (\Depth-1.25,0,\Height);
        \draw[black] (O) -- (C) -- (G) -- (D) -- cycle;
        \draw[black] (O) -- (A) -- (E) -- (D) -- cycle;
        \draw[black, fill=blue!20,opacity=0.8] (O) -- (A) -- (B) -- (C) -- cycle;
        \draw[black] (C) -- (B) -- (F) -- (G) -- cycle;
        \draw[black] (A) -- (B) -- (F) -- (E) -- cycle;
        \draw[below] (0, 0*\Width, \Height) node{$\mathscr{T}_f$};
        \draw[midway] (\Depth/2,\Width-\Width/2,\Height/2) node {$\widetilde{\alpha}$};
       \draw[midway] (\Depth/2+6,\Width-\Width/2,\Height/2) node {$\widetilde{\beta}$};

      \draw[<->,black,dotted ] (\Depth/2,\Width-\Width/2+.35,\Height/2)  .. controls (\Depth/2+3,\Width-\Width/2+2,\Height/2).. (\Depth/2+6,\Width-\Width/2+.35,\Height/2);

      \draw[<->,black ] (\Depth/2-2,-\Width-\Width/2+2,\Height/2)  .. controls (\Depth/2+1,-\Width-\Width/2+1,\Height/2).. (\Depth/2+3.5,-\Width-\Width/2+2,\Height/2);

      \draw[] (\Depth/2+3,\Width-\Width/2+2,\Height/2) node{$\tilde{z}_c$};

      \draw[] (\Depth/2+1,-\Width-\Width/2+1,\Height/2) node{$z_c$};

        \coordinate (O2) at (\Depth,0,0);
        \coordinate (A2) at (\Depth,\Width,0);
        \coordinate (B2) at (\Depth,\Width,\Height);
        \coordinate (C2) at (\Depth,0,\Height);
        \coordinate (D2) at (\Depth+\DepthTwo,0,0);
        \coordinate (E2) at (\Depth+\DepthTwo,\Width,0);
        \coordinate (F2) at (\Depth+\DepthTwo,\Width,\Height);
        \coordinate (G2) at (\Depth+\DepthTwo,0,\Height);
        
        \coordinate (O3) at (\Depth+\DepthTwo+\Sep-4,0,0);
        \coordinate (A3) at (\Depth+\DepthTwo+\Sep-4,\Width,0);
        \coordinate (B3) at (\Depth+\DepthTwo+\Sep-4,\Width,\Height);
        \coordinate (C3) at (\Depth+\DepthTwo+\Sep-4,0,\Height);
        \coordinate (D3) at (\Depth+2*\DepthTwo+\Sep-4,0,0);
        \coordinate (E3) at (\Depth+2*\DepthTwo+\Sep-4,\Width,0);
        \coordinate (F3) at (\Depth+2*\DepthTwo+\Sep-4,\Width,\Height);
        \coordinate (G3) at (\Depth+2*\DepthTwo+\Sep-4,0,\Height);
            \draw[black, fill=red!20,opacity=0.8] (O3) -- (C3) -- (B3) -- (A3) -- cycle;
        \draw[black] (O3) -- (D3);
        \draw[black] (A3) -- (E3);
        \draw[black] (B3) -- (F3);
        \draw[black] (C3) -- (G3);
        \draw[below] (\Depth+\DepthTwo+\Sep-4, 0*\Width, \Height) node{$\mathscr{T}_b$};

        \draw[white] (D)--(E)--(F)--(G)--cycle;
        \draw[] (O)--(D);
        
    \end{tikzpicture}
    \caption{
    The figure shows the bulk boundary systems on the fermionic and bosonic side.
    The bottom of the figure displays two (2+1)-d topological orders related by $z_c$, while the top of the figure shows the corresponding (3+1)d bulk theories. By \cref{upstream_conjecture}, $z_c$ extends to  $\widetilde{z}_c$ which relates a noninvertible theory with an invertible one.}
    \label{fig:alphabeta}
\end{figure}

\subsection{Partition Functions for super-MTC and Anomaly Indicators}\label{subsec:recipe}

In the context of bosonic topological orders, \cite{walker2019,Bulmash2021,Ye:2022bkx} presented the bulk TFT in terms of a generalized version of the Crane-Yetter model \cite{Crane1993}. The bulk TFT can be identified by calculating the partition function on the generating manifolds for the relevant bordism group. These partition functions on the complete list of generating manifolds give a list of \textit{anomaly indicators}. This is a gadget that characterizes the anomaly in terms of a specific element inside the relevant cobordism group, and is presented in terms of the data of the topological order and symmetry action at hand. 

We wish to directly adapt these methods in the bosonic case to compute the anomaly indicators for fermionic topological orders. 
Specifically, we need to calculate the partition function of the anomaly theory on a complete list of generators of $(BG_b, s, \omega)$-twisted spin bordism groups, and the anomaly can be accordingly identified as an element in the (Pontryagin dual) cobordism group, $\mho^4_\xi$. According to \cref{upstream_conjecture}, we can use the same Crane-Yetter model to obtain the partition function of the bosonized theory, denoted as $\widetilde \beta$ in \S\ref{subsection:unpack} and \S\ref{subsection:conjforInv}. We are now tasked with calculating the partition function of $\widetilde{\alpha}$ that hosts the fermionic topological order on its boundary, as first proposed in \cite{Tata2021}.

As discussed around Eq.~\eqref{eq:requirement}, for a fermionic symmetry given by data $(G_b, s, \omega)$, the data of a generating manifold contains three pieces of data: a manifold $M$, a map $f\colon M\rightarrow BG_b$, and a spin-structure $\zeta$ on $f^*(V) \oplus TM$. In this subsection, we will present the full data $M, f, \zeta$ as the argument of $\mc{Z}_f$, but later when dealing with specific examples we will usually omit $f$ and $\zeta$ to avoid clutter. For such a generating manifold, using the bosonic shadows of \cref{def:ZfZb}, the partition function for the spin TFT can be decomposed as 
\begin{equation}\label{eq:summing_bosonic_shadow}
    \mc{Z}_f(M, f, \zeta) = \frac{1}{\sqrt{|{H}^2(M; \Z/2)|}}\sum_{[\mathfrak{L}]\in H_1(M;\Z/2)}Z_b(M, f; \mathfrak{L})z_c(M, f, \zeta; \mathfrak{L})\,.
\end{equation}
Here we use a slightly different notation for the arguments of $z_c$ that were introduced in \cref{zc_defn}: instead of using $\Lambda$ we use $\mathfrak{L}$ which represents the Poincaré dual of $\Lambda$ such that $[\mathfrak{L}]\in H_1(M;\Z/2)$, and the partition function is written as the summation running over all elements $[\mathfrak{L}]\in H_1(M;\Z/2)$. 

Each summand is the multiplication of the bosonic shadow $Z_b$ and $z_c$. The bosonic shadows $Z_b(M, f; \mathfrak{L})$ are insensitive to $\zeta$, and are constructed in terms of the super-MTC data and the symmetry action. Moreover, we need to insert an extra fermion loop into the cycles of ${M}$ represented by $\mathfrak{L}$ when calculating $Z_b(M, f; \mathfrak{L})$. In contrast, $z_c(M, f, \zeta; \mathfrak{L})$ is explicitly dependent on $\zeta$ and on neither the super-MTC data nor the symmetry action data. With regards to  \cref{upstream_conjecture}, $z_c$ extends to a map on anomaly theories, implemented by $\tilde{z}_c$. However, $z_c$ strictly has more data associate with it than $\tilde{z}_c$ does, as it is a bosonization on the theories and not just the anomalies.


Now we describe the practical recipes to calculate $Z_b$ and $z_c$. 
To calculate $z_c$ for every $\mathfrak{L}$, we can start with a set of generators $[\mathfrak{L}_i]\in H_1(M;\Z/2)$. We assign a certain phase $z_c(M, \mathfrak{L}_i, \zeta)$ to each $\mathfrak{L}_i$ according to $\zeta$ in the following way,
       \begin{align}
        \label{eq:sigmaLbc}
            z_c(M, f, \zeta; \mathfrak{L}_i) = \begin{cases}
            \pm 1 \text{ if $\mathfrak{L}_i$ is orientable, \quad i.e., $w_1(\mathfrak{L}_i) = 0 \mod 2$} \\
            \pm i \text{ if $\mathfrak{L}_i$ is unorientable, i.e., $w_1(\mathfrak{L}_i) = 1 \mod 2$}\,.
            \end{cases}
        \end{align}
In particular, for orientable $\mathfrak{L}_i$, the $+$ or $-$ sign tracks whether the cycle $\mathfrak{L}_i$ has a bounding or non-bounding spin-structure, respectively. According to \cite{Gaiotto2016,Kobayashi2019,Tata2021}, $z_c$ is a quadratic refinement of a higher cup product pairing:
\begin{equation}
z_c(M, f, \zeta; \mathfrak{L}_i + \mathfrak{L}_j) = z_c(M, f, \zeta; \mathfrak{L}_i)z_c(M, f, \zeta; \mathfrak{L}_j)(-1)^{\int_M\Lambda_i\cup_2 \Lambda_j}
\end{equation}
where again $\Lambda_{i,j}$ are Poincaré dual to $\mathfrak{L}_{i,j}$ such that $[\Lambda_{i,j}]\in H^3(M; \Z/2)$. Here the formula involves higher cup products $\cup_i$, which are introduced in \cite{Steenrod1947ProductsOC} and interpreted from a combinatorial point of view in \cite{thorngren2018combinatorial,Tata2020,Chen:2021ppt}. In this way, for a given $\zeta$, the extra phase $z_c(M, f, \zeta; \mathfrak{L})$ can be identified for every  $[\mathfrak{L}]\in H_1(M, \Z/2)$. 

To calculate $Z_b(M, f; \mathfrak{L})$, we can simplify the calculation a lot by identifying the \emph{handle decomposition} of the manifold $M$. For bosonic theory described by a unitary-MTC, \cite{Ye:2022bkx} develops the recipe to calculate the partition function of the anomaly theory, given the handle decomposition of $M$. For a super-MTC, the recipe to calculate $Z_b(M, f; \mathfrak{L})$ is almost the same, except that we need to take into account the contribution of $\mathfrak{L}$ by introducing an extra fermion loop according to $\mathfrak{L}$, as discussed in e.g. \cite{Tata2021}. Hence, we can directly port the recipe in \cite{Ye:2022bkx} to calculate $Z_b$. When $G_b$ is a finite group, the recipe is detailed in \cite[Section III.D]{Ye:2022bkx}. We repeat it here and refer the reader back to that section for more details. We will subsequently use it to calculate the bosonic shadow $Z_b(M, f; \mathfrak{L})$, given the data of a super-MTC and some symmetry action on it. We will go to Lie group symmetries in \cref{app:U1_time}. 

\begin{figure}[!htbp]
\includegraphics[width=\textwidth]{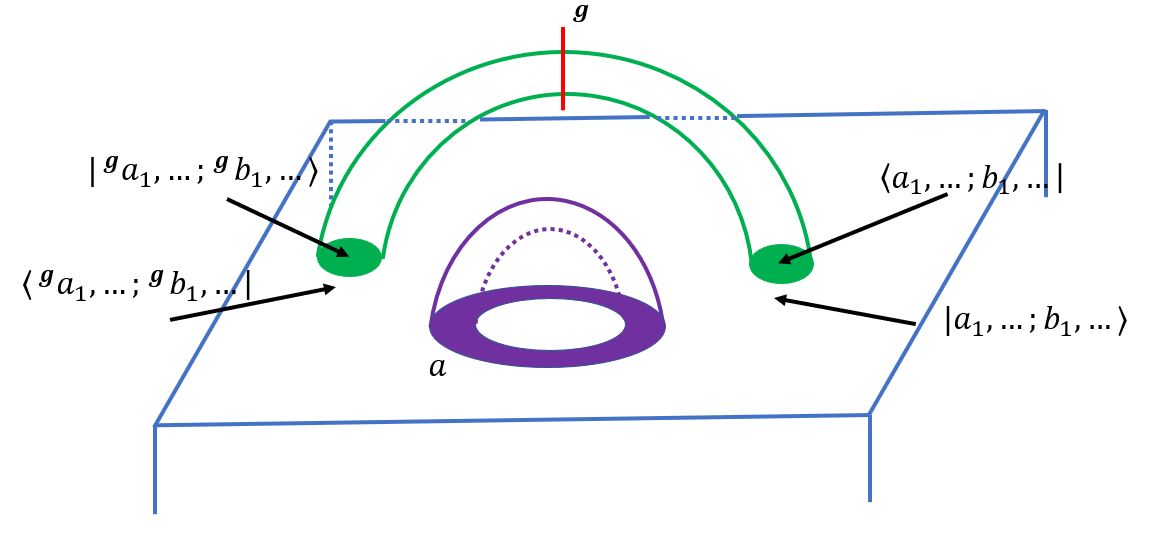}
\caption{Illustration of a blue 0-handle, a green 1-handle and a purple 2-handle together with labels assigned to their attaching regions. The green shaded regions are the attaching regions $S^0\times D^3$ of the 1-handle, and the purple shaded regions are the attaching regions $S^1\times D^2$ of the 2-handle. The red line displays a holonomy, which crosses the 1-handle with the section being $D^3$. We associate an anyon $a$ to the 2-handle. We also associate a vector $|a_1,\dots;b_1,\dots\rangle$ and a dual vector $\langle a_1,\dots;b_1,\dots|$ to the attaching regions living on the 0-handle side and 1-handle side, respectively (these two sides are identified by the embedding map that attaches the 1-handle to the 0-handle).}
\label{fig:Handle}
\end{figure}

\begin{enumerate}

    \item Identify a handle decomposition of the manifold $M$, which, for a 4-dimensional manifold $M$, is captured by the \emph{Kirby diagram}. Jumping ahead, we can then translate these Kirby diagrams into anyon diagrams by labeling lines with anyons and balls with morphisms (states) in a proper fashion. With nontrivial $\mathfrak{L}$, we also need to add an extra $\psi$ fermion loop that follows the path of $\mathfrak{L}$. We refer the reader ahead to \cref{fig:Kirby_RP4} and \cref{fig:anyon1_RP4,fig:anyon2_RP4} for a comparison between the Kirby diagram and its corresponding anyon diagrams.
    
    \item On each 1-handle add appropriate holonomy\footnote{This holonomy is called \emph{defects} in \cite{Ye:2022bkx}.} according to the $G_b$-bundle $\mc{G}$ determined by $f\colon M\rightarrow BG_b$. Specifically, if the holonomy along a 1-handle is $\bf{g}$, we slice the 1-handle by a plane labeled by $\bf{g}$, as illustrated in \cref{fig:Handle}. 
       
    \item Now we start labeling the handle decomposition. 
    
    For the 2-handles, the $S^1$ boundary of each 2-handle is separated by the holonomies into segments. Associate an anyon $a$ to an arbitrary segment on the $S^1$ boundary of each 2-handle. The anyons on the other segments are related to $a$ by $G_b$-actions from the holonomies. 

    For the 1-handles, associate a dual vector $\langle a_1,\dots;b_1,\dots|_{\mu\dots}K^{s(\bf g)}$ and a vector\\ $|^{\bf g}a_1,\dots;^{\bf g}b_1,\dots\rangle_{\tilde{\mu}\dots}$ to the two $D^3$ planes of the attaching region $S^0\times D^3$ of every 1-handle, where $a_1,\dots$ and $b_1,\dots$ are labels of anyons running out of and into the lower $D^3$ plane of the attaching region of the 1-handle. This is illustrated in \cref{fig:Handle}. In the presence of nontrivial $\mathfrak{L}$, we should include a local fermion $\psi$ loop according to $\mathfrak{L}$. 

    The labels on the 0-handle should be completely determined by the labels on the 1-handles and 2-handles. For a connected $M$, we can always choose a handle decomposition such that there is only one 0-handle. 
    
    \item Now we write down the contribution from each 2-handle, 1-handle and 0-handle separately, and we refer to the results as $\eta$-factors, $U$ factors and $\langle K \rangle$, respectively.    
    
    $\eta$-factor: It is the phase for $a$ from the natural isomorphism that connects the functor of successive $G_b$-actions $\rho_{\bf{g}_1}\circ \rho_{\bf{g}_2}\circ \cdots$ to the identity functor. This is explained in detail in \cite[Section III.D, Comment g]{Ye:2022bkx}.
    
    $U$-factors: The $U$-factors are given by
    \begin{equation}\label{eq:U-factor}
    \langle a_1,\dots;b_1,\dots|_{\mu \dots}K^{s(\bf g)}\rho_{\bf g}^{-1}|^{\bf g}a_1,\dots;^{\bf g}b_1,\dots\rangle_{\tilde \mu\dots}\\
    =U_{\bf g}^{-1}\left(^{\bf g}a_1,\dots;^{\bf g}b_1,\dots\right)_{\tilde\mu\dots,\mu\dots}
    \end{equation}
    
    $\langle K \rangle$: Since we label everything on the Kirby diagram and turn it into an anyon diagram, the contribution from the 0-handle is simply the evaluation of the anyon diagram from the explicit data of the given super-MTC.
    
    \item Assemble the result for the bosonic shadow as follows:
    \begin{equation}
    \begin{aligned}\label{eq:main_computation}
     {Z}_b\left(M, f; \mathfrak{L}\right) = D^{-\chi+2(N_4-N_3)}\times \sum_{\text{labels}}   &\Bigg(\frac{\displaystyle \prod_{ \text{2 handle $i$}} d_{a_i}}{{\displaystyle \prod_{\text{1-handle $x$}} \left(\prod_{\text{2-handle $j$ across $x$}}{d_{a_j}}\right)^{1/2}}}\\
     &\times \big(\prod_i(\eta\text{-factors})_i\big)\times \big(\prod_x  (U\text{-factors})_x \big) \times \langle K \rangle \Bigg)\,.
    \end{aligned}
    \end{equation}
    Here $N_k$ is the number of $k$-handles in the handle decomposition, and $\chi\equiv N_0-N_1+N_2-N_3+N_4$ is the Euler number of $M$.
\end{enumerate}
\begin{lem}
The expression Eq.~\eqref{eq:main_computation} is independent of the exact form of the handle decomposition, position of holonomies, and various gauge transformations.
\end{lem}
\begin{proof}[Proof sketch]
The proof is the same as in~\cite[\S C.1, \S C.2, \S C.3]{Ye:2022bkx}, because the proof there simply uses the fact that the category under consideration is a pre-modular tensor category and does not reference the modularity property. 
\end{proof}
Moreover, the proof in \cite[\S C.4]{Ye:2022bkx} also shows that $Z_1(M, [0])$ (when there is no insertion of the fermion loop) is an invariant for oriented bordism.

The invariants we have just defined are the partition functions of a topological field theory, because we built them by fermionization. From \cref{upstream_conjecture}, we think this TFT is an anomaly theory, meaning it should be invertible. We do not have a general proof, though it is true in all examples we computed.
\begin{cor}\label{usu_bordism}
Assuming \cref{upstream_conjecture} and \cref{conj:inv}, the fermionic anomaly indicators are $(BG_b, s,\omega)$-twisted spin bordism invariants.
\end{cor}
\begin{proof}
Recall that $\xi_{BG_b,s,\omega}$ denotes the tangential structure corresponding to a $(BG_b,s,\omega)$-twisted spin structure, and let $\xi_{BG_b,s,\omega}(n)$ denote the pullback of $\xi_{BG_b,s,\omega}$ along $B\O(n)\to B\O$. Freed-Hopkins-Teleman~\cite{FHT10} showed that the partition functions of invertible TFTs are \emph{SKK invariants}, meaning the anomaly indicators define a homomorphism
\begin{equation}
    \phi \colon \pi_4(\Sigma^4 \mathit{MT\xi}_{BG,s,\omega}(4)) \longrightarrow \C^\times.
\end{equation}
(The relationship between homotopy groups of Madsen-Tillmann spectra and SKK invariants is a combination of work of Ebert~\cite[\S 2.5]{Ebe13}, Bökstedt-Svane~\cite{BS14}, and Kreck-Stolz-Teichner (unpublished); see Szegedy~\cite{Sze23}.)
We would like to obtain ordinary bordism invariants. The obstruction to lifting an SKK bordism invariant $\phi$ to an ordinary bordism invariant is $\phi(S^4)$, where $S^4$ carries the tangential structure arising from the boundary of $D^5$ (in particular, it has a trivial $G_b$-bundle). This is shown for $\xi = \SO$ and $\xi = \O$ by Karras-Kreck-Neumann-Ossa~\cite[Theorem 4.2]{KKNO73}, and the proof for general tangential structures is analogous. Anomaly indicators on $S^4$ with this tangential structure are equal to $1$, so we obtain an actual bordism invariant.
\end{proof}

\section{\texorpdfstring{Warmup: $\Z/4^{Tf}$}{}}\label{sec:Z/2^Timerev}

As a warmup, in this section we consider the $\Z/4^{Tf}$ symmetry and rederive the anomaly indicator for the $\Z/4^{Tf}$ symmetry, which was first proposed in \cite{Wang:2016qkb}. In light of \cref{def:fermionic_symmetry}, the triple which defines what we call the $\Z/4^{Tf}$ symmetry is $G_b=\Z/2$ with both $s$ and $\omega$ nontrivial. This symmetry can also be represented by the algebra $\mc{T}^2=(-1)^F$, where $\mc{T}$ is the generator of $\Z/2$ time-reversal and $(-1)^F$ denotes fermion parity. 
In the 10-fold way classification, this symmetry is in ``class DIII''.

It is well-known that the relevant bordism group for the $\Z/4^{Tf}$ symmetry is $\Omega_4^{\Pin+}\cong \Z/16$, generated by $\mathbb{RP}^4$ with nontrivial $\Z/2$-bundle and either of its two \pinp structures~\cite[\S 2, Theorem 3.4(a)]{Gia73a}.

\begin{figure}[!htbp]
\includegraphics[width=0.3\textwidth]{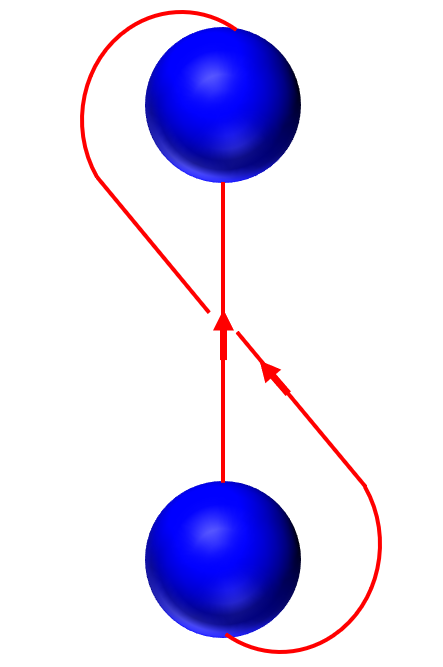}
\caption{The Kirby diagram of $\mathbb{RP}^4$. The two blue spheres illustrate the attaching region of the 1-handle and the red lines illustrate the attaching region of the 2-handle. The 1-handle is unorientable.}
\label{fig:Kirby_RP4}
\end{figure}

We can proceed to derive the anomaly indicator for the $\Z/4^{Tf}$ symmetry by calculating the partition function of the generating manifold $\mathbb{RP}^4$ according to the procedure outlined in \S\ref{subsec:recipe}. The result is summarized in the following proposition.

\begin{prop}\label{prop:anomaly_indicator_Z2}
The anomaly indicator of fermionic topological orders with the $\Z/4^{Tf}$ symmetry is given by 
\begin{equation}\label{eq:indicator_Z2T}
\mc{I} = \frac{1}{\sqrt{2}D}\sum_{a}d_a\theta_a\eta_a\,,
\end{equation}
where
\begin{equation}\label{eq:extraU_Z2T}
\eta_a = 
\left\{
\begin{array}{lr}
\eta_a(\mc{T}, \mc{T}), & \,^{\mc{T}}a = a\\
i\, \eta_a(\mc{T}, \mc{T}) U_{\mc{T}}(a,\psi;a\times\psi)F^{a,\psi,\psi}\,, & \,^{\mc{T}}a = a\times \psi\\
0, & \text{otherwise.}\\
\end{array}
\right.
\end{equation}

\end{prop}

This expression is first proposed in \cite{Wang:2016qkb} and derived in \cite[Appendix J]{Tata2021}, although the derivation there involves cell decomposition instead of handle decomposition and is hence rather involved.

\begin{proof}[Proof sketch of \cref{prop:anomaly_indicator_Z2}]
We derive the formula for the anomaly indicator by calculating the partition function of $\mathbb{RP}^4$. According to Eq.~\eqref{eq:summing_bosonic_shadow}, $\mathbb{RP}^4$ has  $H_1(\mathbb{RP}^4; \Z/2)\cong \Z/2$, hence the partition function can be decomposed as the sum of two bosonic shadows. corresponding to whether or not we insert a fermion loop into the noncontractible cycle. $H^2(\mathbb{RP}^4;\Z/2)\cong\Z/2$ and thus there is a $\frac{1}{\sqrt{2}}$ factor in Eq.~\eqref{eq:summing_bosonic_shadow}.

The minimal handle-decomposition of $\mathbb{RP}^4$ contains 1 0-handle, 1 1-handle, 1 2-handle, 1 3-handle, and 1 4-handle, and its Kirby diagram is given in \cref{fig:Kirby_RP4}. We also need to put a $\mc{T}$ holonomy on the 1-handle. This concludes steps 1 and 2 of the recipe in \cref{subsec:recipe}.


\begin{figure}[!htbp]
  \centering
  
  \begin{subfigure}{0.45\textwidth}
  \centering  
  \includegraphics[page=9]{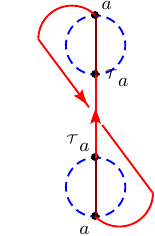}
  \caption{\phantom{a}}
  \label{fig:anyon1_RP4}
  \end{subfigure}
  \hfill    
  \begin{subfigure}{0.45\textwidth}
  \centering
  \includegraphics[page=10]{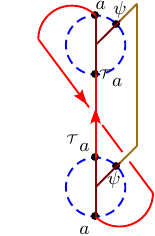}
  \caption{\phantom{b}}
  \label{fig:anyon2_RP4}
  \end{subfigure}

  \caption{Anyon diagrams from the Kirby diagram in Fig.~\ref{fig:Kirby_RP4}, with no extra fermion loop (left) or one extra sanddune-colored fermion loop in the blue 1-handle (right). The red line illustrates the 2-handle, the blue circles illustrate the 1-handle, and the dark red lines illustrate morphisms. Note that in comparison with Fig.~\ref{fig:Kirby_RP4} where both segments flow upward, here one segment flows upward and another segment flows downward due to the nonorientable cycle. }
\end{figure}

When translating this Kirby diagram into anyon diagrams, we start with inserting no fermion loop into the diagram. As in step 3, then we need to label the 2-handle and the 1-handle by anyons and morphisms in a proper way. First, we label the 2-handle by anyons. Because of the nontrivial $\Z/2$-bundle on $\mathbb{RP}^4$, anyons and morphisms are acted upon by $\mc{T}$ when crossing the 1-handle, hence we label two red segments in \cref{fig:Kirby_RP4} by $a$ and $\,^\mc{T}a$, respectively. Moreover, because the cycle is unorientable, in comparison to \cref{fig:Kirby_RP4} we need to flip the flow of one red segment when drawing the anyon diagram, as shown in \cref{fig:anyon1_RP4}. Next, we label the 1-handle by morphisms. On the 1-handle we need to associate a morphism in $\Hom(a, \,^\mc{T}a)$, which is nonempty only when $a=\,^\mc{T}a$. In this way, the Kirby diagram can be translated to the anyon diagram in \cref{fig:anyon1_RP4}. 

Now we calculate the $\eta$-factors, $U$-factors and $\langle K \rangle$ according to step 4. The $\eta$-factor associated to this diagram comes from $\rho_{\mc{T}}^{-1}\circ \rho_{\mc{T}}^{-1}$ acting on $\,^\mc{T}a$, which gives $\eta_{\,^\mc{T}a}(\mc{T}, \mc{T})^* = \eta_a(\mc{T}, \mc{T})$. The $U$-factor associated to this diagram is simply 1. Finally, the anyon diagram in \cref{fig:anyon1_RP4} evaluates to $d_a\theta_a$. After carefully counting all the other factors involving quantum dimensions as in Eq.~\eqref{eq:main_computation}, we arrive at the expression of the first bosonic shadow $Z_1$,
\begin{equation}\label{eq:bosonic_shadow1_Z2}
Z_1 = \frac{1}{D}\sum_{\substack{a\\ 
\{\,^{\mc{T}}a = a\}}}d_a\theta_a\times\eta_a(\mc{T}, \mc{T}).
\end{equation}
The bracket denotes the condition on the anyon $a$ that goes into the sum.

Then we insert a fermion loop into the 1-handle/noncontractible cycle. 
We label the 2-handle and the 1-handle in a similar fashion, and we obtain the anyon diagram in \cref{fig:anyon2_RP4}. From the 1-handle we have the constraint $\,^\mc{T}a = a\times \psi$, such that the morphism associated to the 1-handle is nonempty. The $\eta$-factor associated to this diagram is also $\eta_a(\mc{T}, \mc{T})$, while the $U$-factor associated to this diagram is $U_{\mc{T}}(a,\psi;a\times\psi)$. Finally, the anyon diagram in \cref{fig:anyon2_RP4} evaluates to $d_a\theta_a F^{a,\psi,\psi}$. After carefully counting all the other factors involving quantum dimensions as in Eq.~\eqref{eq:main_computation}, we arrive at the expression of the second bosonic shadow $Z_2$,
\begin{equation}\label{eq:bosonic_shadow2_Z2}
Z_2 = \frac{1}{D}\sum_{\substack{a\\ 
\{\,^{\mc{T}}a = a\times \psi\}}}d_a\theta_a\times\eta_a(\mc{T}, \mc{T})U_{\mc{T}}(a,\psi;a\times\psi)F^{a,\psi,\psi}.
\end{equation}

At the very end, we need to sum over the two bosonic shadows, weighted by the phase factor $z_c$. According to Eq. \eqref{eq:sigmaLbc}, we can choose the phase factor in front of $Z_2$ to be $+i$, which amounts to choosing a \pinp-structure on $\mathbb{RP}^4$ among the two choices. The partition function resulted from the other choice will be related to this one by complex conjugation. Therefore, we have
\begin{equation}
\mathcal{I} =\mc{Z}_f(\mathbb{RP}^4) = \frac{1}{\sqrt{2}}(Z_1+i Z_2)\,.
\end{equation}
Plugging into Eqs.~\eqref{eq:bosonic_shadow1_Z2} and \eqref{eq:bosonic_shadow2_Z2}, we arrive at the partition function of $\mathbb{RP}^4$ taking the form of Eq. \eqref{eq:indicator_Z2T}. This is our desired anomaly indicator for the $\Z/4^{Tf}$ symmetry. It is straightforward to check that this expression is invariant under the vertex basis transformation, Eq.~\eqref{eq:FRvertex_basis_trans} and Eq.~\eqref{eq:Uvertex_basis_trans}, as well as the symmetry action gauge transformation, Eq.~\eqref{eq:Gaction_gauge}.
\end{proof}
As a straightforward application, by directly plugging into Eq.~\eqref{eq:indicator_Z2T} the data of fermionic topological orders $\U_2\times \U_{-1}$ and $\SO(3)_3$ with the $\Z/4^{Tf}$ symmetry (collected in Appendix~\ref{app:data}), we have 
\begin{prop}
The anomaly of fermionic topological orders $\U_2\times \U_{-1}$ and $\SO(3)_3$ with the $\Z/4^{Tf}$ symmetry has anomaly $\nu = 2, 3\in \mho^4_{\Pin^+} \cong\Z/16$, respectively.
\end{prop}

\begin{rem}
    There is also the $\Z/2^T\times \Z/2^f$ symmetry defined by the triple $(G_b, s, \omega)$ such that $G_b=\Z/2$ with $s$ nontrivial and $\omega$ trivial, with symmetry algebra $\mc{T}^2=1$. It is in ``class BDI'' of the 10-fold way classification of fermionic symmetries. The relevant bordism group is $\Omega_4^{\Pin-}\cong 0$~\cite[Theorem 5.1]{ABP69}, hence the partition function on any \pinm manifold is 1 and there is no associated anomaly indicator.
\end{rem}

\section{\texorpdfstring{$\Z/4^T\times \Z/2^f$}{}}\label{sec:anomalycomp}

In this section, we go to the $\Z/4^T\times \Z/2^f$ symmetry and derive the anomaly indicator of any fermionic topological order with the $\Z/4^T\times \Z/2^f$ symmetry. In light of \cref{def:fermionic_symmetry}, the triple $(G_b, s, \omega)$ is given by $G_b=\Z/4$ with $s$ nontrivial and $\omega$ trivial.
The symmetry algebra is $\mc{T}^2=(-1)^F C$, with $\mc{T}$ the time-reversal generator, $C$ charge conjugation and $(-1)^F$ fermion parity. Such $\Z/4^T\times \Z/2^f$ symmetry shows up in many interesting fermionic topological orders, especially $\U_k$, with $k=5,13,17,25,\dots$, as discussed in \cite{Delmastro:2019vnj}. 

In order to obtain the anomaly indicator and eventually the anomaly of these fermionic topological orders with the $\Z/4^T\times \Z/2^f$ symmetry, we first need to identify the relevant tangential structure, which is called EPin structure in the literature \cite{WWZ20}, and calculate the relevant bordism/cobordism group. Despite the simplicity of the symmetry group, this bordism group and its generator have not been calculated before (though see~\cite{BG97, WWZ20} for some partial progress), thanks to tricky extension problems in both the Atiyah-Hirzebruch and Adams spectral sequences. So we undertake this calculation in \cref{subsec:smith_cal}, where we collect the necessary information here
 \begin{thm}\hfill
$\Omega_4^{\EPin}\cong\Z/4$. Let $\mathcal M$ denote the manifold we construct in \cref{thm:generator}, which is the total space of a Klein bottle bundle over $S^2$, then $\mathcal M$ generates $\Omega_4^{\EPin}$.
\end{thm}
Hence, anomalies of $\Z/4^T\times \Z/2^f$ symmetries in (2+1)-d are also classified by $\Z/4$. The additional information we use to resolve the extension question and compute the bordism group comes from a long exact sequence built from the \term{Smith homomorphism}. The Smith homomorphism was first studied by Conner-Floyd~\cite[Theorem 26.1]{CF64}, then later generalized to many situations by many authors; see~\cite{COSY20, HKT20, DDKLPTT, DDKLPTT2} for discussions aimed at a mathematical physics audience. The use of the long exact sequence associated to the Smith homomorphism and its cofiber, identified explicitly in~\cite{DDKLPTT, DDKLPTT2}, is a newer technique, but has already proven helpful to resolve differentials and extension questions in several bordism computations in~\cite{DDHM23, Deb23, DL23}.

Having obtained the necessary topological information, we can calculate the partition function of a manifold representative of a generator of the bordism group, following the recipe outlined in \S\ref{subsec:recipe}, and obtain the anomaly indicator with the $\Z/4^T\times \Z/2^f$ symmetry. 
We present the anomaly indicator for the $\Z/4^T\times \Z/2^f$ symmetry in Proposition \ref{prop:anomaly_indicator}. 

\begin{prop}\label{prop:anomaly_indicator}
The anomaly indicator of fermionic topological order with the $\Z/4^T\times \Z/2^f$ symmetry is given by 

\begin{equation}\label{eq:indicator_Z4T}
\begin{aligned}
\mc{I} = \frac{1}{2D^2}&\sum_{\substack{a,b,y,z,u\\ \mu\nu\rho\sigma\tilde{\mu}\tilde{\nu}\tilde{\rho}\tilde{\sigma}\alpha\beta\mu'}}
{d_a}\frac{\theta_u}{\theta_b}
\left(R_{\,^{\mc{T}}y}^{\,^{\mc{T}}a,\,^{\mc{T}^3}a}\right)^*_{\tilde{\rho}\sigma}  \left(F_{\,^{\mc{T}}z}^{b,a,\,^{\mc{T}^2}a}\right)^*_{(u,\alpha,\beta)(y,\tilde{\sigma},\mu')}
\left(F_{\,^{\mc{T}}z}^{b,a,\,^{\mc{T}^2}a}\right)_{(u,\alpha,\beta)(\,^{\mc{T}^2}y,\rho,\mu)}\\
&\times U_{\mc{T}}^{-1}(\,^{\mc{T}}a,\,^{\mc{T}^3}a;\,^{\mc{T}}y)^*_{\sigma\tilde{\sigma}} 
U_{\mc{T}}^{-1}(b,\,^{\mc{T}^2}y;\,^{\mc{T}}z)_{\mu\tilde{\mu}}
U_{\mc{T}}^{-1}(a,\,^{\mc{T}^2}a;\,^{\mc{T}^2}y)_{\rho\tilde{\rho}} \\
&\times \eta_a(\mc{T}, \mc{T})^*\eta_{\,^{\mc{T}^2}a}(\mc{T}, \mc{T})^*\eta_{\,^{\mc{T}^2}a}(\mc{T}^2, \mc{T}^2)^*\times \mc{U}_{byz}\,,
\end{aligned}
\end{equation}
where
\begin{equation*}\label{eq:extraU_Z4T}
\text{\footnotesize{\(
\mc{U}_{byz} = 
\left\{
\begin{array}{lr}
\delta_{\tilde{\mu}\nu} \delta_{\mu' \tilde{\nu}}U_{\mc{T}}^{-1}(b,\,^{\mc{T}}y;\,^{\mc{T}}z)^*_{\nu\tilde{\nu}} \,,
&\,^{\mc{T}}b = b~\&~\,^{\mc{T}}z = z\\
\phantom{aa}\\
-\left(F_{z\times \psi}^{\psi,\,^{\mc{T}}b,y}\right)^*_{(b,-,\mu')(z,\tilde{\nu},-)}
\left(F_{z\times \psi}^{\psi,\,^{\mc{T}}b,\,^{\mc{T}}y}\right)_{(b,-,\nu)(z,\tilde{\mu},-)}U_{\mc{T}}^{-1}(b,\,^{\mc{T}}y;\,^{\mc{T}}z)^*_{\nu\tilde{\nu}}\,, 
&\,^{\mc{T}}b = b\times \psi ~\&~\,^{\mc{T}}z = z\times \psi\\
\phantom{aa}\\
i\left(F_{z}^{\psi,b,y}\right)_{(b\times\psi,-,\tilde\nu)(z\times\psi,\mu',-)}
\left(F_{z\times \psi}^{\psi,b,\,^{\mc{T}^2}y}\right)_{(b\times\psi,-,\nu)(z,\tilde{\mu},-)}\left(F^{\psi,\psi,z}_z\right)^*\,, 
&\,^{\mc{T}}b = b~\&~\,^{\mc{T}}z = z \times \psi\\
\quad\quad\quad\quad\times U_{\mc{T}}^{-1}(\psi, b, b\times \psi)^* \,U_{\mc{T}}^{-1}(b\times \psi,\,^{\mc{T}}y;\,^{\mc{T}}z)^*_{\nu\tilde{\nu}}\\
\phantom{aa}\\
i\delta_{\tilde{\mu}\nu} \delta_{\mu' \tilde{\nu}}\left(F^{\,^{\mc{T}}b, \psi, \psi}_{\,^{\mc{T}}b}\right)U_{\mc{T}}^{-1}(b, \psi, b\times \psi)^* \,U_{\mc{T}}^{-1}(b\times \psi,\,^{\mc{T}}y;\,^{\mc{T}}z)^*_{\nu\tilde{\nu}}\,,
&\,^{\mc{T}}b = b\times \psi~\&~\,^{\mc{T}}z = z \\
\phantom{aa}\\
0, & \text{otherwise.}\\
\end{array}
\right.    \)}}
\end{equation*}

Roman letters in the formulas denote anyons, and Greek letters denote bases of fusion spaces. The use of a dash in the subscript of $F$-symbols is when the fusion involves the fermion and therefore only has a single (unique) channel. 
\end{prop}

\begin{figure}[!htbp]
    \begin{subfigure}{0.45\textwidth}
        \centering
        \includegraphics[page=11]{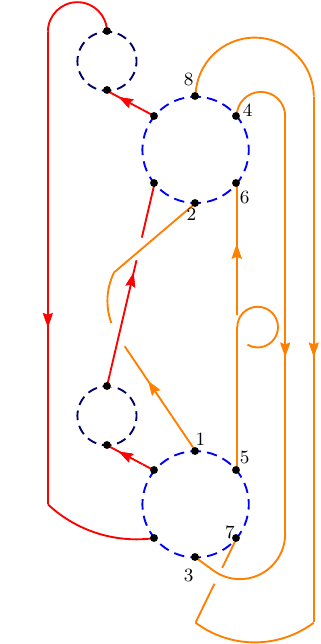}
    \caption{}
    \label{fig:Orderfiber}
    \end{subfigure}     \hfill
    \begin{subfigure}{.45\textwidth}
     \centering
     \includegraphics[page=12]{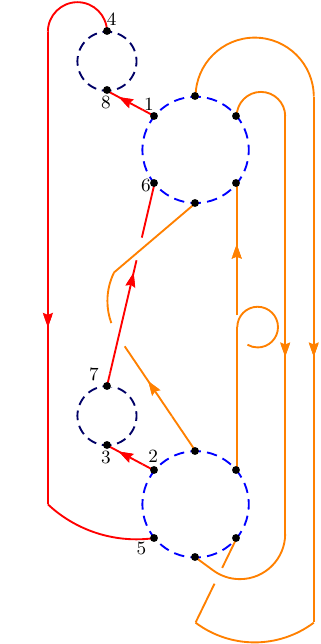}
    \caption{}
    \label{fig:OrderKlein}
\end{subfigure}
\label{fig:Order}
    \caption{Illustration of how anyons travel in the orange (left) and red (right) lines. Anyons will travel from 1 to 8 and back to 1. We remark that some arrows on the diagrams are reversed with respect to the ordering of vertices. The anyons at position 1 on the orange and red lines are labeled by $a$ and $b$, respectively.}
    \end{figure}

\input{AnyonDiagram}

\begin{proof}[Proof sketch of \cref{prop:anomaly_indicator}]
Even though the expression of the anomaly indicator is relatively complicated, the derivation still follows closely the procedure outlined in \S\ref{subsec:recipe}. We just need to calculate the partition function of the generating manifold $\mc{M}$ of $\Omega_4^\EPin$, and its important properties are presented in \S\ref{app:manifold_generator}. In particular $\mc{M}$ has  $H_1(\mc{M}, \Z/2)\cong\Z/2\oplus\Z/2$, hence the partition function can be decomposed as the sum of four bosonic shadows; and
$H^2(\mc{M}; \Z/2) \cong \Z/2\oplus\Z/2$, thus there is a $\frac{1}{2}$ factor in Eq.~\eqref{eq:summing_bosonic_shadow}. The minimal handle-decomposition of $\mc{M}$ contains 1 0-handle, 2 1-handles, 2 2-handles, 2 3-handles, and 1 4-handle, and its Kirby diagram is given in \cref{fig:Kirby}. 

To translate this Kirby diagram into anyon diagrams, we start with inserting no fermion loop into the diagram. We need to label the 2-handles and 1-handles by anyons and morphisms according to the recipe. First we need to label segments of the orange and red loop by anyons. We follow the ordering of the two anyon loops in \cref{fig:Orderfiber} and \ref{fig:OrderKlein}, and label the anyons at position 1 by $a$ and $b$, respectively. According to the $\Z/4$ bundle structure of $\mc{M}$ given in Appendix~\ref{app:manifold_generator}, the blue 1-handle has a nontrivial $\Z/4$ bundle put on it. In particular, we take the convention that when an anyon crosses from the top blue circle to the bottom one, it receives an action by $\mc{T}^{-1} = \mc{T}^3$; crossing from the bottom blue circle to the top results in an action by $\mc{T}$. This produces the labels in \cref{fig:anyon1}. To draw the anyon diagram \cref{fig:anyon1}, we also need to reverse the flow of anyons on several segments because of the presence of the unorientable blue 1-handle. Finally we label the morphisms by $x,y,z$ (together with $\mu,\nu$ labels that we omit in the figures for clarity), and the morphisms are also acted upon by $\mc{T}$ when crossing the blue 1-handle. The morphisms are nonempty only when $\,^{\mc{T}}b = b$ and~${}^{\mc{T}}z = z$. This gives \cref{fig:anyon1}. 

Now we start inserting some fermion loops into the diagram. Since $H_1(\mc{M}; \Z/2)\cong \Z/2\oplus\Z/2$, there are three inequivalent possibilities: adding one extra fermion loop to the dark blue handle, adding one extra fermion loop to the blue handle, and adding one extra fermion loop crossing both the blue and dark blue handle. The labels can then be obtained in a similar fashion, and we obtain \cref{fig:anyon2,fig:anyon3,fig:anyon4}.

Then we can directly translate these anyon diagrams into compact expressions in terms of the data for a super-MTC according to the standard rules of computing anyon diagrams \cite{bakalov2001lectures,Selinger2011,barkeshli2014,Ye:2022bkx}. After adding the correct factors of $U$-symbols, $\eta$-symbols and quantum dimensions according to Eq.~\eqref{eq:main_computation}, we obtain the compact form of individual bosonic shadows $Z_i,i=1,\dots,4$. For example, we have 
\begin{equation}\label{eq:bosonic_shadow1}
\begin{aligned}
Z_1 = \frac{1}{D^2}&\sum_{\substack{a,b,y,z,u\\ \mu\nu\rho\sigma\tilde{\mu}\tilde{\nu}\tilde{\rho}\tilde{\sigma}\alpha\beta\\
\{\,^{\mc{T}}b = b\\
\,^{\mc{T}}z = z\}}}
{d_a}\frac{\theta_u}{\theta_b}
\left(R_{\,^{\mc{T}}y}^{\,^{\mc{T}}a,\,^{\mc{T}^3}a}\right)^*_{\tilde{\rho}\sigma}  \left(F_{z}^{b,a,\,^{\mc{T}^2}a}\right)^*_{(u,\alpha,\beta)(y,\tilde{\sigma},\tilde{\nu})}
\left(F_{z}^{b,a,\,^{\mc{T}^2}a}\right)_{(u,\alpha,\beta)(\,^{\mc{T}^2}y,\rho,\mu)}\delta_{\tilde\mu\nu}\\
&\times U_{\mc{T}}^{-1}(\,^{\mc{T}}a,\,^{\mc{T}^3}a;\,^{\mc{T}}y)^*_{\sigma\tilde{\sigma}} U_{\mc{T}}^{-1}(b,\,^{\mc{T}}y;z)^*_{\nu\tilde{\nu}} 
U_{\mc{T}}^{-1}(b,\,^{\mc{T}^2}y;z)_{\mu\tilde{\mu}}
U_{\mc{T}}^{-1}(a,\,^{\mc{T}^2}a;\,^{\mc{T}^2}y)_{\rho\tilde{\rho}} \\
&\times \eta_a(\mc{T}, \mc{T})^*\eta_{\,^{\mc{T}^2}a}(\mc{T}, \mc{T})^*\eta_{\,^{\mc{T}^2}a}(\mc{T}^2, \mc{T}^2)^*\,,
\end{aligned}
\end{equation}
which is the first case for $\mathcal{U}_{byz}$ in Proposition \ref{prop:anomaly_indicator}. The brackets denote conditions on the anyon $b$ and $z$ that go into the sum. The remaining three bosonic shadows correspond to the other three nontrivial cases for $\mathcal{U}_{byz}$.

Finally, we explain the phase, denoted by $z_c$ in Eq.~\eqref{eq:summing_bosonic_shadow}, in front of each bosonic shadow $Z_i,i=1,\dots,4$. This directly comes from the spin-structure we choose for $\mc{M}$ in \cref{thm:generator}. In particular, the phase in front of $Z_2$ is $-1$ and reflects the fact that the orientable cycle of $\mc{M}$ corresponds to $S^1$ with non-bounding spin-structure. Fermions therefore pick up a minus sign when traversing this $S^1$. 
We also choose the spin-structure of $\mc{M}$ such that the phase in front of $Z_3$ and $Z_4$ is $+i$.\footnote{If we choose a different spin-structure for $\mc{M}$ such that the phase in front of $Z_3$ and $Z_4$ is $-i$, then $\mc{M}$ with the new spin-structure is simply the inverse of $\mc{M}$ with the old spin-structure in the bordism group $\Omega_4^\EPin$. It is also straightforward to see that compared to $\mc{I}$, the partition function of $\mc{M}$ with the new spin-structure is simply complex-conjugated.}
Summing the expressions of bosonic shadows $Z_i$ up with the correct phase in the front according to Eq.~\eqref{eq:summing_bosonic_shadow}, we have
\begin{equation}\label{eq:summing_bosonic_shadow_Z4T}
\mathcal{I} = \mc{Z}_f(\mc{M}) = \frac{1}{2}(Z_1-Z_2+iZ_3+iZ_4)\,.
\end{equation}
Plugging into the specific expression of $Z_i, i=1,\dots,4$, we finally arrive at the compact expression of the anomaly indicator $\mc{I}$ for the $\Z/4^T\times \Z/2^f$ symmetry, given in \cref{prop:anomaly_indicator}.
\end{proof}

Now we plug into the data of some simple fermionic topological orders with the $\Z/4^T\times \Z/2^f$ action to obtain the actual value of their anomalies, which are summarized in \cref{prop:actual_anomaly}. 

\begin{thm}\label{prop:actual_anomaly}
The anomaly of fermionic topological orders $\U_5$, $\U_2\times \U_{-1}$ and $\SO(3)_3$ with the $\Z/4^T\times \Z/2^f$ symmetry has anomaly $\nu = 0, 2, 3\in \mho^4_\EPin \cong\Z/4$, respectively.
\end{thm}
\begin{proof}
After directly plugging the data of $\U_5$, $\U_2\times \U_{-1}$ and $\SO(3)_3$ (collected in Appendix~\ref{app:data}) into the formula in Eq.~\eqref{eq:indicator_Z4T}, the result is 
\begin{equation}
    \mc{I} = 1, -1, -i
\end{equation}
Thus, we immediately see that the anomalies of these fermionic topological orders correspond to  $\nu=0, 2, 3$, respectively, in $\mho^4_\EPin\cong\Z/4$.
\end{proof}
To give more credence to our computation, in Appendix~\ref{app:cascade}, we use the anomaly cascade developed in \cite{Bulmash:2021ryq} to rederive the anomaly of these fermionic topological orders.\footnote{For $\U_5$ and $\U_2\times \U_{-1}$, the calculation in the appendix reproduces the result $\nu=0,2$ obtained by the anomaly indicator. For $\SO(3)_3$, naïvely the anomaly cascade can only tell us that the anomaly $\nu$ is odd.} The calculation of $\U_5$ can be easily generalized to $\U_k,k=5,13,\dots$, which all have the $\Z/4^T\times \Z/2^f$ symmetry as discussed in \cite{Delmastro:2019vnj}. It turns out that just like $\U_5$ they all have anomaly $\nu=0$. 

\begin{rem}\label{rem:solving_extension}
From the partition function of $\SO(3)_3$, we see that the anomaly indicator does take values in $\{i^k,k=0,\dots,3\}$, as dictated by $\mho^4_\EPin\cong \Z/4$ we obtain in \cref{subsec:smith_cal} from Smith homomorphism. This is not at all obvious from the explicit formula in Eq.~\eqref{eq:indicator_Z4T}. This fact gives yet another way to solve the extension problem (assuming \cref{conj:inv}). The manifolds in $\Omega^{\EPin}_4$ include the K3 surface and $\mathcal{M}$.
Evaluating the partition function for $\SO(3)_3$ on $\mathcal{M}$ and obtaining the value $-i$ indicate that indeed the extension exists and $\mathcal{M}$ is indeed the generating manifold. This method is in similar spirit to the method of using $\eta$-invariants to solve extension problems, such as in~\cite{Gil84, Gil85, BG87a, Gil87, Gil88a, Gil88b, Sto88, Gil89, BG95, Gil96, BGS97, BY99, BYG99, BY06, Mal11, MR15, Hsi18, KPMT20, DGL22, HTY22, DDHM23}.
\end{rem}

\section{Anomaly Indicators with Lie group Symmetry: 10-fold way}\label{app:U1_time}

In this section, we give examples of anomaly indicators for fermionic symmetries involving Lie groups. These include seven out of ten symmetries in the 10-fold way classification of fermionic symmetries \cite{Wang2014,FH21InvertibleFT}, as listed in \cref{tab:definition_of_groups}.
The anomaly indicators for class AII and class AIII involving $\U$ symmetry were first presented in \cite{Lapa2019}, and a derivation for abelian topological orders was given in \cite{Kobayashi2021}. \cite{Ning2021} proposed anomaly indicators for all these symmetries, with the help of replacing the time-reversal symmetry with the mirror symmetry under the crystalline equivalence principle. We will obtain the same expressions in the most general setting by following the recipe in \S\ref{subsec:recipe}. In doing so we showcase how to apply our general recipe to Lie group symmetries, and write down anomaly indicators of them without resorting to any additional assumptions. Moreover, for symmetries in class A or class C, even though there is no 't Hooft anomaly associated to them in (2+1)-d, we calculate the partition functions on some generating manifolds as well, which correspond to an element in the cobordism group, and interpret the result as the formula for thermal and $\U$ Hall conductance. For class CI and class CII symmetry, we also demonstrate that certain elements in the cobordism group which classifies the anomaly can never be realized by any fermionic topological order, demonstrating the phenomenon of ``symmetry-enforced gaplessness'' \cite{Wang2014,Ning2021}.

\begin{table}[!htbp]
\centering
\renewcommand{\arraystretch}{1.3}
\resizebox{\columnwidth}{!}{%
\begin{tabular}
{c c c c c c c c }
\toprule
Class & $G_f$ & $G_b$ & $s$ & $\omega$ & Tangential & Bordism Group & Generator \\ \midrule
D & $\Z/2^f$ & 1 &  &  &  Spin & $\Z$ & K3 \\
DIII & $\Z/4^{Tf}$ & $\Z/2$ & $x$ & $x^2$ & Pin\textsuperscript{$+$} & $\Z/16$ & $\RP^4$   \\
BDI & $\Z/2^T\times\Z/2^f$ & $\Z/2$ & $x$ & $0$ & Pin\textsuperscript{$-$} & $0$ & --   \\ 
A & $\mathrm{U}_f(1)$ & $\U$ & $0$ & $w_2$ & Spin\textsuperscript{$c$} & $(\Z)^2$ & $\CP^2$, $S^2\times S^2$  \\ 
AI & $\mathrm{U}_f(1)\rtimes \Z/2^T$ & $\O(2)$ & $w_1$ & $w_2$ & Pin\textsuperscript{$\tilde c-$} & $\Z/2$ & $\CP^2$ \\ 
AII & $\mathrm{U}_f(1)\rtimes_{\Z/2} \Z/4^{Tf}$ & $\O(2)$ & $w_1$ & $w_2+w_1^2$ & Pin\textsuperscript{$\tilde c+$} & $(\Z/2)^3$ & $\RP^4$, $\CP^2$, $S^2\times S^2$ \\ 
AIII & $\mathrm{U}_f(1)\times_{\Z/2} \Z/4^{Tf}$ & $\U\times \Z/2$ & $x$ & $w_2+x^2$ & Pin\textsuperscript{$c$} & $\Z/8\oplus\Z/2$ & $\RP^4$, $\CP^2$  \\
C & $\SU_f(2)$ & $\SO(3)$ & $0$ & $w_2$ & Spin\textsuperscript{$h$} & $(\Z)^2$ & $\CP^2$, $S^4$  \\ 
CI & $\mathrm{SU}_f(2)\times_{\Z/2} \Z/4^{Tf}$ & $\O(3)$ & $w_1$ & $w_2$ & Pin\textsuperscript{$h+$} & $\Z/4\oplus\Z/2$ & $\RP^4$, $\CP^2$  \\
CII & $\mathrm{SU}_f(2)\times \Z/2^T$& $\O(3)$ & $w_1$ & $w_2+w_1^2$ & Pin\textsuperscript{$h-$} & $(\Z/2)^3$ & $\RP^4$, $\CP^2$, $S^4$ \\
\bottomrule
\end{tabular}
}
\caption{We list the fermionic symmetries of the 10-fold way classification and present them in terms of both $G_f$ and $(G_b, s, \omega)$ as in \cref{def:fermionic_symmetry}. Here, $x$ is the generator of $H^1(\Z/2;\Z/2)$ and $w_{1,2,3}$ are the Stiefel-Whitney classes of (special) orthogonal groups $\U = \SO(2)$, $\O(2)$, $\SO(3)$ or $\O(3)$. The table also gives the corresponding tangential structures, corresponding bordism groups in (3+1)-d and the generating manifolds. Class DIII and class BDI with $G_b=\Z/2$ are discussed in \S\ref{sec:Z/2^Timerev} and see \cref{item:K3} for some comments about class D. We discuss the last seven cases involving $\U$ and $\SO(3)$ symmetry in this appendix. }
\label{tab:definition_of_groups}
\end{table}

In particular, we should pay special attention to how to write down the $\eta$-symbols in the presence of Lie group symmetries. This was discussed in detail in \cite{Ye:2022bkx} in the context of bosonic topological order and we repeat it here. Consider a manifold $M$ with a $G_b$-bundle on it defined by the map $f\colon M\rightarrow BG_b$. We want to obtain the $\eta$-factor for some 2-handle $h$ if we label $h$ by an anyon $a$. We write down this $\eta$-factor in terms of the (fractional) charge of $a$.
\begin{defn}\label{defn:charge}
   Suppose $G_b$ is a connected Lie group. Denote $e^{2\pi i q_a}\in\U$ as the phase factor obtained from pairing $[\eta_a]\in H^2(BG_b, \U)$ with the generator of $H_2(BG_b; \Z)\cong \Z$. Then $q_a\in [0,1)$ is defined as the \textit{(fractional) charge} of the anyon $a$ under the symmetry $G_b$.
\end{defn}
Note that in our convention, for the symmetries in the 10-fold way classification involving $\U$, the local fermion $\psi$ carries charge $\frac{1}{2}$. This is with respect to the subgroup $\mathrm{U}_b(1)$ of the bosonic symmetry group $G_b$. Compared with the subgroup $\mathrm{U}_f(1)$ of $G_f$, which is the double cover of $\mathrm{U}_b(1)$, the charge differs by a factor of 2, and the local fermion $\psi$ carries charge 1.\footnote{In this appendix, we use the subscript $f$ to emphasize that $\mathrm{U}_f(1)$ is a subgroup of the fermionic symmetry group $G_f$, while $\U$ with no subscript is a subgroup of $G_b$. Similarly for $\mathrm{SU}_f(2)$.} This is a convention adopted by many physics papers. We use our convention so that all symmetry groups are discussed on equal footing. We also mention that for the symmetries in the 10-fold way classification involving $\SO(3)$, the local fermion $\psi$ carries a projective representation, or is a spinor, under $\SO(3)$.

The $\eta$-factor can be expressed in terms of the charge $q_a$ as follows \cite{Ye:2022bkx}. In the presence of a connected Lie group symmetry $G_b$, $f$ maps a 2-chain $[h]$ in $M$, which represents the 2-handle $h$, to a 2-chain $f_*[h]$ in $BG_b$, which represents $n\in H_2(BG_b; \Z)\cong\Z$. Then the desired $\eta$-factor is simply $e^{2n \pi i q_a}$. Intuitively, such a phase factor can be viewed as the phase the anyon $a$ experiences when traveling along the $S^1$ boundary given the nontrivial background $G_b$-bundle structure, hence the expression is written in terms of the charge of $a$. 

Even though the calculation of the relevant bordism groups is known in the literature \cite{FH21InvertibleFT}, we also need to find the generating manifolds of them, some of which are not explicitly written down in the literature. We present the generating manifolds for each bordism group we consider, and some of the proof relies on the Smith long exact sequence reviewed in \cref{subsec:smith_cal}. 

\subsection{Class A and class C}\label{subapp:anomaly_U(1)}

In this subsection, we start by considering the fermionic symmetries corresponding to ``class A'' and ``class C''. The necessary information of the two symmetries is collected in \cref{tab:definition_of_groups}. A special feature of ``class A'' and ``class C'' is that there is no 't Hooft-like anomaly for the two symmetries, but there can still be a nontrivial partition function that gives various Hall conductance.
We will also see later in this appendix that the calculation of anomaly indicators for some other 10-fold way symmetries reduce to these two cases by restricting to a $\U$ or $\SO(3)$ subgroup. 

Let us start with class A. The associated tangential structure is well-known to be \spinc. We have $\Omega_4^{\Spin^c}\cong\Z\oplus\Z$ \cite[Chapter XI]{Sto68}, generated by,\footnote{We can also say that the two generating manifolds are $\mathbb{CP}^2$ and $S^2\times S^2$ with the $\U$-bundle induced from their complex structure.} 
\begin{itemize}
\item $\mathbb{CP}^2$, with the tautological $\U$ bundle,
\item $S^2\times S^2$, with a $\U$-bundle on it whose classifying map is identified with $(2, 2)$ in the abelian group $[S^2, B\U] \times [S^2, B\U] = (\pi_2(B\U))^2 \cong \Z^2$.
\end{itemize}
The Pontrjagin dual of the bordism group is $\mho^4_{\Spin^c}\cong\U\oplus\U$. Therefore, there is no $(2+1)$-d anomaly associated to class A. Still, in this case, the partition function defined in \S\ref{subsec:recipe} identifies an element in $\mho^4_{\Spin^c}$ as well. In the physics literature, the two $\U$ pieces are interpreted as two theta terms in $(3+1)$-d, which give the thermal Hall conductance and $\U$ Hall conductance.

Given an element $(\Theta_1, \Theta_2) \in \mho^4_{\Spin^c}\cong\U\oplus\U$, the partition function on a manifold $M$ with chosen $\U$-bundle structure and spin-structure can be written as follows \cite{lawson1989spin,Seiberg2016,Wang2014} 
\begin{equation}\label{eq:spinCbasis}
\mc{Z}_f(M) = \exp\left(i(\Theta_1 I_1 + \Theta_2 I_2)\right),
\end{equation}
where
\begin{equation}
I_1 = \frac{1}{8}\left(-\text{Sign}(M) + \int_M\left(c_1\right)^2\right),
\end{equation}
\begin{equation}
I_2 = \int_M\left(c_1\right)^2.
\end{equation}
Here $\text{Sign}(M)$ is the signature of $M$, and $c_1$ is the first Chern class of the \textit{bosonic} $\U$ bundle.\footnote{In terms of the Riemann curvature tensor $R$ and the $\U$ field strength $F$, i.e.\ from the point of view of Chern-Weil theory, $\text{Sign}(M) = \frac{1}{192\pi^2}\int_M \text{tr}(R\wedge R)$ and $c_1 = \frac{F}{2\pi}$.} 

\begin{rem} \label{rem_Hall_conductance}
In physics, $(\Theta_1, \Theta_2)$ is related to the thermal Hall conductance $\kappa$ and $\mathrm{U}_f(1)$ Hall conductance $\sigma_H$ in the following way,
\begin{equation}
    \kappa = \frac{\Theta_1}{2\pi} \pmod 1,\quad \sigma_H = \frac{8\Theta_2 + \Theta_1}{2\pi} \pmod 1\,.
\end{equation}
Here we need an important fact/convention in physics: a (2+1)-d fermionic invertible state with class A symmetry has integer $\kappa$ and $\sigma_H$. Physically, we can stack invertible states with given fermionic topological order without changing the anyon content, and thus $\kappa$ and $\sigma_H$ for a fermionic topological order can be determined only up to contributions from invertible states. Therefore, for class A symmetry only the fractional part of $\kappa$ and $\sigma_H$ can be determined from the anyon content/super-MTC.
\end{rem}

Given a fermionic topological order with class A symmetry, by calculating the partition function on the two manifold representatives, we have
\begin{prop}\label{prop:anomaly_indicator_U(1)case0}
\begin{equation}
    e^{i\Theta_1} = \frac{\mc{Z}_f\left(S^2\times S^2\right)}{\mc{Z}_f\left(\mathbb{CP}^2\right)^8}\,,\quad \exp(i\Theta_2) = \mc{Z}_f\left(\mathbb{CP}^2\right)\,
\end{equation}
where
\begin{equation}\label{eq:indicator_U(1)1}
\mc{Z}_f\left(\mathbb{CP}^2\right) = \frac{1}{\sqrt{2}D}\sum_a d_a^2 \theta_a e^{2\pi i q_a}\,,
\end{equation}
\begin{equation}\label{eq:indicator_U(1)2}
 \mc{Z}_f\left(S^2\times S^2\right) = \frac{1}{2D}\sum_{a,b} d_a d_b S_{ab} e^{4\pi i q_a}e^{4\pi i q_b}.
\end{equation}
Here $q_a$ is the fractional charge of anyon $a$ defined in \cref{defn:charge}.
\end{prop}

\begin{proof}[Proof sketch of \cref{prop:anomaly_indicator_U(1)case0}]
Since both manifolds are simply connected, in both calculation there is only one bosonic shadow to sum over, and hence the calculation and the final expressions are greatly simplified.

The partition function of $\mathbb{CP}^2$ can be calculated as follows. The minimum handle decomposition of $\mathbb{CP}^2$ contains 1 0-handle, 1 2-handle and 1 4-handle. The Kirby diagram can be found in~\cite{gompf1994}; we draw it in Eq.~\eqref{eq:Kirby_cp2}. The topological twist reflects the $+1$ intersection number of $\mathbb{CP}^2$. 
Now we label the 2-handle by anyon $a$. From the $\U$ bundle structure on $\mathbb{CP}^2$, the $\eta$-factor is simply $e^{2\pi i q_a}$, where $q_a\in [0,1)$ is the fractional charge of anyon $a$ as in \cref{defn:charge}. The anyon diagram associated to the Kirby diagram is evaluated as
\begin{equation}\label{eq:Kirby_cp2}
\left\langle
 \raisebox{-0.5\height}{ \includegraphics[page=17]{Draft2-pics} }
\right\rangle = d_a \theta_a\,.
\end{equation}
Assembling all factors as in Eq.~\eqref{eq:main_computation} and Eq.~\eqref{eq:summing_bosonic_shadow}, we have
\begin{equation}
\mc{Z}_f\left(\mathbb{CP}^2\right) = \frac{1}{\sqrt{2}D}{\displaystyle\sum_a d_a^2 \theta_a e^{2\pi i q_a}}.
\end{equation}

The partition function of $S^2\times S^2$ can be calculated in a very similar fashion. The minimum handle decomposition of $S^2\times S^2$ contains 1 0-handle, 2 2-handles and 1 4-handle, and the Kirby diagram is given in~\cite{gompf1994} and drawn in Eq.~\eqref{eq:Kirby_S2S2}. In particular, the two circles correspond to the equators of the two $S^2$ pieces. Now we label the red and orange 2-handle by anyon $a$ and $b$, respectively. From the $\U$ bundle structure on $S^2\times S^2$, the $\eta$-factors are $e^{4\pi i q_a}$ and $e^{4\pi i q_b}$, respectively, where again $q_{a,b}$ is the fractional charge of anyon $a$ and $b$, respectively. The anyon diagram associated to the Kirby diagram is evaluated as
\begin{equation}\label{eq:Kirby_S2S2}
\left\langle
 \raisebox{-0.5\height}{ \includegraphics[page=18]{Draft2-pics} }
\right\rangle = D S_{ab}
\end{equation}
Assembling all factors as in Eq.~\eqref{eq:main_computation} and Eq.~\eqref{eq:summing_bosonic_shadow}, we have
\begin{equation}
\mc{Z}_f(S^2\times S^2) = \frac{1}{2D}\sum_{a,b} d_a d_b S_{ab} e^{4\pi i q_a}e^{4\pi i q_b}.
\qedhere
\end{equation}
\end{proof}
The discussion of class C is very similar to the discussion of class A. The associated tangential structure is spin$^h$. Freed-Hopkins~\cite[Theorem 9.97]{FH21InvertibleFT} showed $\Omega_4^{\Spin^h}\cong\Z\oplus\Z$ (see also~\cite{BM23, Mil23}), and Hu~\cite[Appendix A]{Hu23} found the following set of generators:\footnote{We can also say that the two generating manifolds are $\mathbb{CP}^2$ and $S^4$ with the $\SO(3)$-bundle induced from their almost quaternionic structure. Note that $S^2\times S^2$ with spin$^h$ structure induced from its \spinc structure is not a generating manifold but is bordant to two copies of a generator of $\Omega_4^{\Spin^h}$.}
\begin{itemize}
\item $\mathbb{CP}^2$, with the tautological $\U\subset \SO(3)$ bundle,
\item $S^4$, with an $\SO(3)$-bundle over it, whose classifying map is identified with $f\colon S^4\cong \mathbb{HP}^1 \subset \mathbb{HP}^\infty\cong B\SU(2) \overset{p_*}{\rightarrow} B\SO(3)$, where $p\colon \SU(2)\rightarrow \SO(3)$ is the natural projection.\footnote{This $\SO(3)$-bundle has an interesting property: its spin cobordism Euler class is nonzero, even though its $\Z$-cohomology Euler class vanishes. This means that the caveat raised in \cref{smith_caveat} applies to the Smith homomorphism $\Omega_4^{\Spin^h}\to\Omega_1^\Spin(B\SO(3))$: using Euler classes in ordinary cohomology does not correctly compute the Smith homomorphism. See~\cite[Appendix B]{DDKLPTT2} for details of the computation of this spin cobordism Euler class and its consequences.}
\end{itemize}
The Pontrjagin dual of the bordism group is $\mho^4_{\Spin^h}\cong\U\oplus\U$. Therefore, again there is no $(2+1)$-d anomaly associated to class C, but we can obtain (the fractional part of) the thermal Hall conductance and $\SO(3)$ Hall conductance from the partition functions.

Given an element $(\Theta_1, \Theta_2) \in \mho^4_{\Spin^h}\cong\U\oplus\U$, the partition function on a manifold $M$ with chosen $\SO(3)$-bundle structure and spin-structure can be written as follows \cite{Hu23,Wang2014} 
\begin{equation}\label{eq:spinHbasis}
\mc{Z}_f(M) = \exp\left(i(\Theta_1 I_1 + \Theta_2 I_2)\right),
\end{equation}
where
\begin{equation}
I_1 = \frac{1}{4}\left(-\text{Sign}(M) + \int_M p_1\right),
\end{equation}
\begin{equation}
I_2 = \int_M p_1.
\end{equation}
Here $\text{Sign}(M)$ is the signature of $M$, and $p_1$ is the first Pontrjagin class of the \textit{bosonic} $\SO(3)$ bundle. 
\begin{rem} 
In physics, $(\Theta_1, \Theta_2)$ is related to the thermal Hall conductance $\kappa$ and $\mathrm{U}_f(1)$ Hall conductance $\sigma_H$ in the following way,
\begin{equation}\label{eq:Hall_conduct_SO(3)}
    \kappa = \frac{\Theta_1}{\pi} \pmod 2,\quad \sigma_H = \frac{4\Theta_2 + \Theta_1}{\pi} \pmod 2\,.
\end{equation}
A (2+1)-d fermionic invertible state with class C symmetry has even integer $\kappa$ and $\sigma_H$. Therefore, as discussed in \cref{rem_Hall_conductance}, for class C symmetry $\kappa$ and $\sigma_H$ can only be determined mod 2 from the anyon content/super-MTC.
\end{rem}


Given a fermionic topological order with $\SO(3)$ action, by calculating the partition function on the two manifold representatives, we have
\begin{prop}\label{prop:anomaly_indicator_SO(3)case0}
\begin{equation}\label{eq:theta_SO(3)}
    e^{i\Theta_1} = \frac{1}{\mc{Z}_f\left(\mathbb{CP}^2\right)^4}\,,\quad \exp(i\Theta_2) = \mc{Z}_f\left(\mathbb{CP}^2\right)\,
\end{equation}
where
\begin{equation}\label{eq:indicator_SO(3)}
\mc{Z}_f\left(\mathbb{CP}^2\right) = \frac{1}{\sqrt{2}D}\sum_a d_a^2 \theta_a e^{2\pi i q_a}\,.
\end{equation}
Here $q_a\in \{0, \frac{1}{2}\}$ is the fractional charge of anyon $a$ defined in \cref{defn:charge}, and labels whether $a$ carries
integer ($q_a=0$) or spinor ($q_a = \frac{1}{2}$) representation under $\SO(3)$.
\end{prop}

\begin{proof}[Proof sketch of \cref{prop:anomaly_indicator_SO(3)case0}]
The calculation of the partition function of $\mathbb{CP}^2$ completely parallels the calculation in the proof of \cref{prop:anomaly_indicator_U(1)case0}, and we immediately obtain the result in Eq.~\eqref{eq:indicator_SO(3)}. The only subtlety is that here $q_a$ only takes value in $\{0, \frac{1}{2}\}$ because $H^2(B\SO(3);\U)\cong \Z/2$.\footnote{This follows from the Bockstein long exact sequence associated to $0\to\Z\to\R\to\U\to 0$ and the facts that $H^2(B\SO(3);\R) = 0$, $H^3(B\SO(3);\R) = 0$, and $H^3(B\SO(3);\Z)\cong\Z/2$~\cite[Theorem 1.5]{Bro82}.}

We just need to focus on the partition function of $S^4$. Again, since $S^4$ is simply connected, we just have one bosonic shadow to sum over. Moreover, the handle decomposition of $S^4$ is extremely simple, i.e., it just contains 1 0-handle and 1 4-handle glued together along the boundary $S^3$. Hence following the formula in Eq.~\eqref{eq:main_computation}, we immediately have
\begin{equation}
    \mc{Z}_f(S^4) = 1\,.
\end{equation}
By directly evaluating $I_1$ and $I_2$ for the two generating manifolds, we obtain Eq.~\eqref{eq:theta_SO(3)}.
\end{proof}

Here we see an interesting phenomenon: the partition function on some manifold representing a nontrivial class is always 1 for any fermionic topological order with given symmetry. From this we can derive some interesting physical consequences. In the current example of symmetry in class C, by inspecting Eq.~\eqref{eq:theta_SO(3)} and Eq.~\eqref{eq:Hall_conduct_SO(3)}, we have

\begin{cor}\label{cor:SO(3)_Hall}
    Any fermionic topological order with class C symmetry action must have $\SO(3)$ Hall conductance given by an even integer.
\end{cor}

\subsection{Class AI, AII, AIII}
\label{tilde_c_plus}

Now we go to class AI, AII and AIII, whose fermionic symmetry groups all contain $\U$ as a subgroup. The definitions of these fermionic symmetries are in \cref{tab:definition_of_groups}. It turns out that the anomaly indicators for these symmetries can all be obtained from the anomaly indicators of class A and class DIII ($\Z/4^{Tf}$ symmetry), whose anomaly indicators have been obtained in \cref{prop:anomaly_indicator_U(1)case0} and \cref{prop:anomaly_indicator_Z2}. See \cite[Section VI]{Ye:2022bkx} for a similar calculation in the context of bosonic topological order. We list the result below.

\begin{prop}\label{prop:classA_anomaly}
The classification of anomaly and anomaly indicators of fermionic topological orders with symmetries in class AI, AII and AIII are given by
\begin{itemize}
\item Class AI. The anomaly is classified by $\Z/2$, with the anomaly indicator $\mc{I} = \mc{Z}_f(\CP^2)$.
\item Class AII. The anomaly is classified by $(\Z/2)^3$, with the anomaly indicator $\mc{I}_1 = \mc{Z}_f(\RP^4)$, $\mc{I}_2 = \mc{Z}_f(\CP^2)$ and $\mc{I}_3 = \mc{Z}_f(S^2\times S^2)$.
\item Class AIII. The anomaly is classified by $\Z/8\oplus\Z/2$, with the anomaly indicator of the $\Z/8$ piece $\mc{I}_1 = \mc{Z}_f(\RP^4)$, and the $\Z/2$ piece $\mc{I}_2 = \mc{Z}_f(\CP^2)$.
\end{itemize}
The partition functions of $\CP^2$ and $S^2\times S^2$ are calculated in \cref{prop:anomaly_indicator_U(1)case0} and the partition function of $\RP^4$ is calculated in \cref{prop:anomaly_indicator_Z2}. 
\end{prop}

\begin{rem}
    Even though these anomaly indicators have the same expressions as the expressions for class A or class DIII symmetries, because the classification of anomaly is different, they actually take values in different sets. For example, for class AI, $\mc{I} = \mc{Z}_f(\CP^2)$ takes values only in $\set{\pm 1}$.
\end{rem}

These results are straightforward if we know the generating manifolds of the corresponding bordism groups. Hence, we end this subsection by commenting on the generating manifolds for these symmetries. 

First consider class AI. The corresponding tangential structure is \cite{FH21InvertibleFT, Ste21} $\Pin^{\tilde c-}\coloneqq (\Pin^-\ltimes\Spin(2))/\set{\pm 1}$. Here $\Pin^-$ acts on $\Spin(2)$ by $\Pin^-\to\O\overset{\det}{\to} \{\pm 1\}$ and $\set{\pm 1}$ acts on the circle group $\Spin(2)\cong \U$ by complex conjugation; then, to obtain $\Pin^{\tilde c-}$, quotient by the diagonal $\set{\pm 1}$ subgroup. Similarly for $\Pin^{\tilde c+}$ below. $\Omega_4^{\Pin^{\tilde c-}}\cong\Z/2$~\cite{FH21InvertibleFT}, generated by $\CP^2$ with tautological $\U$ bundle.\footnote{One way to see this is to observe that $\int w_2^2$ is a bordism invariant of pin\textsuperscript{$\tilde c-$} manifolds and is nonvanishing on $\CP^2$.}

Next we consider class AII. Freed-Hopkins showed $\Omega_4^{\Pin^{\tilde c+}}\cong (\Z/2)^{\oplus 3}$~\cite[Theorem 9.87]{FH21InvertibleFT}. Because both $\U$ (class AI) and $\Z/2^{T}$ (class DIII) are subgroups of $\O(2)^T$, \spinc manifolds and \pinp manifolds all have canonically induced pin$^{\tilde c+}$ structures. 
Therefore, a natural candidate set of manifold representatives consists of $\CP^2$, $S^2\times S^2$, and $\RP^4$, with induced pin$^{\tilde c+}$ structures. 
However, it is not explicitly proven in the literature that these three manifolds are linearly independent in $\Omega_4^{\Pin^{\tilde c+}}$. Here we explicitly present the proof, which utilizes the Smith long exact sequence reviewed in \cref{subsec:smith_cal}. 

\begin{prop}\label{prop:pinctplus}
The classes of $\CP^2$, $S^2\times S^2$, and $\RP^4$ are linearly independent in $\Omega_4^{\Pin^{\tilde c+}}$, hence form a generating set. Here $\CP^2$ and $S^2\times S^2$ have pin\textsuperscript{$\tilde c+$} structures induced from their \spinc structures, and $\RP^4$ has its pin\textsuperscript{$\tilde c+$} structure induced from either of its two \pinp structures.
\end{prop}

\begin{proof}
Observe that the bordism invariant $\int w_1^4\colon \Omega_4^{\Pin^{\tilde c+}}\to\Z/2$ vanishes on $\CP^2$ and $S^2\times S^2$, but does not vanish on $\RP^4$. Hence we mainly need to prove that
$\CP^2$ and $S^2\times S^2$ are linearly independent in $\Omega_4^{\Pin^{\tilde c+}}$, which amounts to proving that the map $\Omega_4^{\Spin^c}\to\Omega_4^{\Pin^{\tilde c+}}$ induced by the inclusion $\U\hookrightarrow \O(2)$ maps two generators in $\Omega_4^{\Spin^c}$ to two generators in $\Omega_4^{\Pin^{\tilde c+}}$.

 Let $V_t\to B\O(2)$ be the tautological rank-$2$ vector bundle and $\sigma\coloneqq\mathrm{Det}(V_t)$. A pin\textsuperscript{$\tilde c+$} structure is equivalent to a $(B\O(2), 3V_t)$-twisted spin-structure\footnote{If the reader is comfortable with virtual vector bundles, pin\textsuperscript{$\tilde c+$} structures are also equivalent to $(B\O(2), -V_t)$-twisted spin structures. This is how pin\textsuperscript{$\tilde c+$} structures are presented as twisted spin structures in~\cite[(10.2)]{FH21InvertibleFT}.} \cite[(10.2)]{FH21InvertibleFT} (see~\cite[Lemma D.8]{SSGR18} for a related but different characterization). Consider the Smith long exact sequence from \cref{Smith_LES} with $X = B\O(2)$, $V = 3V_t$, and $W = \sigma$; by \cref{which_sphere_bundle}, $S(W)\to B\O(2)$ is homotopy equivalent to the map $B\U\to B\O(2)$. Therefore we have a long exact sequence
\begin{equation}
\label{1st_look_pintilde_LES}
    \dotsb\to
    \Omega_k^\Spin((B\U)^{3V_t-6}) \to
    \Omega_k^\Spin((B\O(2))^{3V_t-6}) \overset{S_\sigma}{\to}
    \Omega_{k-1}^\Spin((B\O(2))^{3V_t+\sigma-7})\to\dotsb
\end{equation}
Using this, we interpret the pieces of~\eqref{1st_look_pintilde_LES} as follows:
\begin{itemize}
    \item Because $2V_t\to B\U$ is spin, $\Omega_k^\Spin((B\U)^{3V_t-6})\cong\Omega_k^\Spin((B\U)^{V_t-2})$, which is identified with \spinc bordism~\cite{BG87a, BG87b}.
    \item The map $\Omega_k^\Spin((B\U)^{3V_t-6}) \to
    \Omega_k^\Spin((B\O(2))^{3V_t-6})$ can be identified with the map $\Omega_k^{\Spin^c}\to\Omega_k^{\Pin^{\tilde c+}}$ given by the induced pin\textsuperscript{$\tilde c+$} structure described above, because both are induced by the inclusion $\U\hookrightarrow \O(2)$.
    \item The characteristic-class data for a $(B\O(2), 3V_t+\sigma -7)$-twisted spin structure is $w_1(3V_t+\sigma) = 0$ and $w_2(3V_t+\sigma) = w_2$. This tangential structure corresponds to the fermionic symmetry defined by the triple $(G_b = \O(2), s= 0, \omega = w_2)$, which is often called a spin-$\O(2)$ structure. Spin-$\O(2)$ structures are also studied in~\cite{Nak13, DDHM22, HHLZ22, LS22, Ste21, DDHM23, DYY25}.
\end{itemize}
Thus~\eqref{1st_look_pintilde_LES} becomes
\begin{subequations}
\begin{equation}
\label{2nd_look}
    \dotsb\longrightarrow \Omega_5^{\Pin^{\tilde c+}} \longrightarrow \Omega_4^{\Spin\text{-}\O(2)}\longrightarrow \Omega_4^{\Spin^c}\longrightarrow \Omega_4^{\Pin^{\tilde c+}} \longrightarrow \Omega_3^{\Spin\text{-}\O(2)} \longrightarrow\dotsb
\end{equation}
so plugging in $\Omega_4^{\Spin^c}\cong\Z^2$~\cite[Chapter XI]{Sto68}, $\Omega_4^{\Pin^{\tilde c+}}\cong (\Z/2)^{\oplus 3}$ and $\Omega_5^{\Pin^{\tilde c+}} = 0$~\cite[Theorem 9.87]{FH21InvertibleFT}, and $\Omega_3^{\Spin\text{-}\O(2)} \cong \Z/2$~\cite[\S 4.1]{Ste21} and $\Omega_4^{\Spin\text{-}\O(2)} \cong \Z^2$~\cite[Proposition 3.47]{DYY25}, \eqref{2nd_look} simplifies to
\begin{equation}
    0\longrightarrow \Z^2\longrightarrow
    \eqnmarkbox[blue]{node1}{\Z^2}
    \longrightarrow
    \eqnmarkbox[red]{node2}{(\Z/2)^{\oplus 3}}
    \longrightarrow \Z/2.
    \annotate[yshift=-0.3em]{below, left}{node1}{$\Omega_4^{\Spin^c}$}
    \annotate[yshift=-0.3em]{below, right}{node2}{$\Omega_4^{\Pin^{\tilde c+}}$}
    \vspace{1.2em}
\end{equation}
\end{subequations}
Exactness implies that any generating set of $\Omega_4^{\Spin^c}$ is still linearly independent (over $\Z/2$) in $\Omega_4^{\Pin^{\tilde c+}}$. We conclude.
\end{proof}

Finally we consider class AIII. The relevant tangential structure is called \pinc in the literature, and Bahri-Gilkey~\cite[Theorem 0.2(b)]{BG87a} show that 
 there is an isomorphism $\varphi\colon \Omega_4^{\Pin^c}\overset\cong\to \Z/8\oplus\Z/2$, such that the two \pinc structures\footnote{If a manifold $M$ admits a \pinc structure, then its set of \pinc structures is a torsor over $H^2(M;\Z)$, analogous to \spinc structures. $\RP^4$ admits a \pinc structure, since the obstruction $\beta(w_2)$ lives in $H^3(\RP^4;\Z) = 0$, so the set of \pinc structures is a torsor over $H^2(\RP^2;\Z)\cong\Z/2$.} on $\RP^4$ are sent by $\varphi$ to $(\pm 1, 0)\in\Z/8\oplus\Z/2$ and the \pinc structure on $\CP^2$ induced by its \spinc structure from \S\ref{subapp:anomaly_U(1)} is sent to $(0, 1)$. Thus we may take $\RP^4$ and $\CP^2$ as our generators.

\subsection{Class CI, CII}\label{tilde_h_plus}

In this last subsection, we consider class CI and class CII. Both of their fermionic symmetry groups contain $\SO(3)$ as a subgroup. The necessary information of these symmetries is listed in \cref{tab:definition_of_groups}. The anomaly indicators for these symmetries can also be obtained from the anomaly indicators of class C and class DIII, obtained in \cref{prop:anomaly_indicator_SO(3)case0} and \cref{prop:anomaly_indicator_Z2}. For these two symmetries, it turns out that certain element in the group that classifies the anomaly can never be realized by any fermionic topological order. Hence, any system that saturates this anomaly can only be gapless. This is the phenomenon of ``symmetry-enforced gaplessness'', as discussed in e.g. \cite{Wang2014,Ning2021}. We list the results below.

\begin{prop}\label{prop:classC_anomaly}
The classification of anomaly and anomaly indicators of fermionic topological orders with symmetries in class CI and CII are given by
\begin{itemize}
\item Class CI. The anomaly is classified by $\Z/4\oplus \Z/2$, with the anomaly indicator of the $\Z/4$ piece $\mc{I}_1=\mc{Z}_f\left(\mathbb{RP}^4\right)$, and the $\Z/2$ piece $\mc{I}_2 = \mc{Z}_f(\CP^2)$. However, despite the $\Z/4$ classification, $\mc{I}_1$ can only take values in $\{\pm 1\}$. 
\item Class CII. The anomaly is classified by $(\Z/2)^3$, with the anomaly indicator $\mc{I}_1=\mc{Z}_f\left(\mathbb{RP}^4\right)$, $\mc{I}_2 = \mc{Z}_f\left(\mathbb{CP}^2\right)$ and $ \tilde{\mc{I}} = \mc{Z}_f\left(S^4\right)$. However, $\tilde{\mc{I}}$ is identically 1 from \cref{prop:anomaly_indicator_SO(3)case0}.
\end{itemize}
The partition function of $\CP^2$ is calculated in \cref{prop:anomaly_indicator_SO(3)case0} and the partition function of $\RP^4$ is calculated in \cref{prop:anomaly_indicator_Z2}. 
\end{prop}

Again, these results are straightforward if we determine the generating manifolds of the corresponding bordism groups, $\Omega_4^{\Pin^{h\pm}}$. The bordism groups were computed by Freed-Hopkins~\cite[Theorem 9.97]{FH21InvertibleFT} to be $\Z/4\oplus\Z/2$ for pin\textsuperscript{$h+$} and $(\Z/2)^{\oplus 3}$ for pin\textsuperscript{$h-$}. To describe the generators, we use the fact that spin\textsuperscript{$h$} and pin\textsuperscript{$+$} structures naturally induce pin\textsuperscript{$h\pm$} structures. The standard inclusion of $\SO(3)$ into $\O(3)$ suggests that spin\textsuperscript{$h$} structures can define either kind of pin\textsuperscript{$h\pm$} structure. We can embed the nonzero element of $\Z_2^T$ to diag$(-1, -1, -1)$ in $\O(3)$ such that a pin\textsuperscript{$+$} structure induces a pin\textsuperscript{$h+$} structure. Via the canonical embedding $\O(1)\cong \Z/2 \hookrightarrow \O(3)$, a pin\textsuperscript{$+$} structure also induces a pin\textsuperscript{$h-$} structure. We let the reader check that the two maps pull back the classes $s$ and $\omega$ of class CI and class CII correctly into corresponding elements in class DIII, as in \cref{def:map_symmetry} to define a map of fermionic symmetry groups.

\begin{lem}[{Guo-Putrov-Wang~\cite[Claims 3 and 6]{GPW18}}]\hfill
\begin{enumerate}
    \item There is an isomorphism $\phi\colon \Omega_4^{\Pin^{h+}}\overset\cong\to\Z/4\oplus\Z/2$, such that
    \begin{enumerate}
    \item $\phi([\RP^4]) = (1, 0)$, where $\RP^4$ has the pin\textsuperscript{$h+$} structure induced from its \pinp structure.
    \item $\phi([\CP^2]) = (0, 1)$, where $\CP^2$ has the pin\textsuperscript{$h+$} structure induced from its spin\textsuperscript{$h$} structure from above.
    \end{enumerate}
    \item There is an isomorphism $\psi\colon\Omega_4^{\Pin^{h-}}\overset\cong\to (\Z/2)^{\oplus 3}$, such that the bordism classes of $\RP^4$, $\CP^2$, and $S^4$ are a set of $\Z/2$-basis for $\Omega_4^{\Pin^{h-}}$. Here $\CP^2$ and $S^4$ have pin\textsuperscript{$h-$} structures induced from their spin\textsuperscript{$h$} structures and $\RP^4$ has its pin\textsuperscript{$h-$} structure induced from either of its \pinp structures.
\end{enumerate}
\end{lem}

Finally, we observe that for both class CI and class CII symmetries, certain elements in the classification of anomaly can never be realized by any fermionic topological order with given symmetry action. 
\begin{itemize}
\item For class CI, $\mc{I}_1=\pm i$ can never be realized by any fermionic topological order. 

\begin{proof}
Here we repeat the proof in \cite{Ning2021} of this statement. Note that the summation in $\mc{I}_1=\mc{Z}_f\left(\mathbb{RP}^4\right)$ involves two types of anyons, one type satisfying $\,^\mc{T}a = a$ and another type satisfying $ \,^\mc{T}a = a\times \psi$. However, for any fermionic topological order with class CI symmetry, the second type of anyons actually does not exist. This is because by inspecting Eq.~\eqref{eq:etaConsistency}, $a$ and $\,^\mc{T} a$ must have the same $\SO(3)$ charge, i.e., either they both have $q = 0$ (integer spin) or both have $q=\frac{1}{2}$ (half-integer spin), while $\psi$ must have $q=\frac{1}{2}$. Hence, there is no way such that $ \,^\mc{T}a = a\times \psi$ is satisfied by some anyon $a$. For the first type of anyons, from Eqs.~\eqref{eqn:UFURConsistency} and \eqref{eq:etaConsistency} we see that $\theta_a$ and $\eta_a({\mc{T}}, {\mc{T}})$ can only take values in $\pm 1$. Therefore, $\mc{I}_1$ can only take real values rather than $\pm i$.
\end{proof}

\item For class CII, $\tilde{I}=-1$ can never be realized by any fermionic topological order. 
\end{itemize}
Therefore, any state with these anomalies can never be a fermionic topological order, and hence must be gapless. This phenomenon is called ``symmetry-enforced gaplessness'', first observed in \cite{Wang2014}.

\section{Conclusion and Discussion}\label{sec:discussion}
The main focus of our work was on detecting anomalies for fermionic topological orders, especially time-reversal anomalies of $\Z/4^{Tf}$ symmetry and $\Z/4^T\times \Z/2^f$ symmetry. The $\Z/4^T\times \Z/2^f$ symmetry is realized already in $\U_k$ theories, and we hope that for low values of $k$ our result can be interesting for applications to the fractional quantum Hall effect. Going beyond just discrete symmetries, in \S\ref{app:U1_time} we showcase even more examples of our techniques by computing the anomaly indicators for all symmetries in the tenfold way involving Lie group symmetries. The mathematical underpinning that makes our results sensible arises from how we built up the bosonization conjecture, and the invertibility conjecture.  
We then spelled out how to make the formal mathematics explicit by calculating the partition function of the anomaly theory using bosonic shadows and techniques from geometric topology.

We wrap up by giving a quick summary of interesting future directions.

\begin{enumerate}
    \item \label{item:K3} One of the most important manifolds we need to consider is the K3 surface. The bordism class of the K3 surface generates $\Omega^\Spin_4\cong\Z$, and K3 equipped with a trivial $G_b$-bundle often appears as a generator of bordism groups associated to many different fermionic symmetry groups. These include fermionic symmetry $\Z/k^T\times \Z/2^f$ when $8\mid k$ with the associated bordism group $\Omega_4^{\EPin[k]}$, according to \cref{thm:generator}. Hence, the K3 surface is relevant to the calculation of anomaly indicators for many fermionic symmetries. Furthermore, the partition function of K3 gives the formula for the chiral central charge of a super-MTC, similar to how the partition function of $\mathbb{CP}^2$ gives the formula for the chiral central charge of a unitary-MTC \cite{Crane1993c}, which is called the Gauss-Milgram formula in the literature.
     Such a formula is very important in understanding the properties of the corresponding fermionic/bosonic topological order. For example, it gives (the fractional part of) the thermal Hall conductance as discussed in \S\ref{subapp:U(1)5}, and it is relevant to understanding the boundary properties of the topological order \cite{Kong2020,You23}.
     
     We can anticipate that the formula for the partition function on the K3 surface we get by directly reading the Kirby diagram is very complicated, akin to the complication of the formula in Eq.~\eqref{eq:indicator_Z4T}.
     Aasen-Jones-Walker approach this problem from the point of view of \term{characteristic bordism}~\cite{KT90}; see~\cite{Walker}.
  
    \item It would be interesting to study anomaly indicators for other symmetries, e.g., dihedral group symmetries. Dihedral groups appear as the point groups of many 2d wallpaper groups, hence anomaly indicators of dihedral group symmetries can have potential application in understanding the ``emergibility'' of topological orders in these lattice systems. Furthermore, dihedral group symmetries have also been found in abelian Chern-Simons theories~\cite[Tables 1--4]{Delmastro:2019vnj}. Therefore, it would be interesting to explore the anomalies of these symmetries. The papers~\cite{Gia76, Ped17, GOPWW20, KPMT20, WWZ20, Deb21, DDHM23} have collectively calculated most of the degree-$4$ twisted spin bordism groups of $BD_{2n}$ controlling these anomalies. Moreover, it would be desirable to extend the calculation to other symmetry groups appearing in \cite{Geiko2022} as symmetries of Chern-Simons-Witten theories.
            
    \item Having gained a comprehensive understanding of the anomaly of fermionic topological orders, it becomes evident that a thorough understanding of the gapped boundaries of invertible/SPT states is essential. This is termed ``Clay's problem'' by Freed-Teleman, following Córdova-Ohmori's work~\cite{Cordova2019a,Cordova2019}. We hope that the insights obtained from our paper can serve as a building block for this problem. 
    
    For example, when $k\equiv 0 \mod 8$, we have $\mho^4_{\EPin[k]}\cong \Z/2\oplus\Z/2$. From \cref{prop:map} and the anomaly of $\SO(3)_3$ for $\Z/4^T\times \Z/2^f$ symmetry, we see that $\SO(3)_3$ with $\Z/k^T\times \Z/2^f$ symmetry action has an anomaly with value $(1, 0)$. It will be interesting to identify if there is any fermionic topological order that has $\Z/k^T\times \Z/2^f$ symmetry and an anomaly with value $(0,1)$, or that there is no such fermionic topological order and we have another example of symmetry-enforced gaplessness, i.e., certain element of ’t Hooft anomaly being not realized by any symmetry-enriched topological order.

    Moreover, we also hope that our result may help generate a necessary and sufficient condition of symmetry-enforced gaplessness. It is worth mentioning that there
    have already been a lot of attempts in this direction, including \cite{Cordova2019a,Cordova2019,Ning2021,Brennan23,Brennan23a,Yang23}.
    
    \item The bosonization conjecture used to argue invertibility of $\widetilde{\alpha}$  in \S\ref{subsection:unpack} and \S\ref{subsection:conjforInv} is reasonable from a physical perspective, but could be improved if one can offer a rigorous mathematical proof using a spin Crane-Yetter construction. There are some discussions of Crane-Yetter construction which generates a fully-extended framed or oriented TFT \cite{Crane1993,Crane1993b,Brochier18,brochier2021invertible,Tham2021}, and we wish to extend the construction to generate a fully-extended spin TFT, which will eventually prove the bosonization conjecture and invertibility conjecture we have. 

    \item Our method of calculating bordism groups and generating manifolds using the Smith homomorphism can be very helpful in calculating bordism groups associated to other symmetries. Specifically, there are many 3-dimensional space groups whose point groups contain the fourfold rotoreflection symmetry $S_4$ \footnote{Note that $S_4$ here has nothing to do with the permutation group of four elements.} as a subgroup, and the $S_4$ symmetry corresponds to an order 4 anti-unitary symmetry according to the crystalline equivalence principle \cite{Thorngren2016,Deb21,Zhang2022}. Our result of $\Omega_4^\EPin\cong\Z/4$ will help in the classification and construction of topological crystalline states protected by these symmetries. Indeed, after the paper is posted on arXiv, version 2 of \cite{Zhang2022} appears, which corroborates the $\Z/4$ result from classifying topological superconductors with spin-1/2 fermions protected by $S_4$. 
    It will be interesting to see how this will eventually lead to a full classification of topological crystalline states protected by space group symmetries.
    
    \item In \cite{Bulmash:2021ryq} the anomaly has been interpreted as the obstruction to extending certain data associated to the symmetry action on the super-MTC to a unitary-MTC. We adopt this interpretation in Appendix~\ref{app:cascade} and see that the anomalies obtained via this approach agree with the anomaly indicator computations. It would be desirable to have a mathematically rigorous connection between the two approaches. 
       
\end{enumerate}

\appendix

\section{\texorpdfstring{The Power of Smith: $\Omega_4^\EPin$ and $\Omega_4^{\EPin[k]}$}{}}\label{subsec:smith_cal}

In this section, we calculate the bordism group involved in the anomaly of the $\Z/4^T\times \Z/2^f$ symmetry in (2+1)-d fermionic systems and identify a generating manifold. Interestingly, we can easily generalize the calculation to the $\Z/k^T\times \Z/2^f$ symmetry,\footnote{Similar to the $\Z/4^T\times \Z/2^f$ symmetry, in light of \cref{def:fermionic_symmetry}, the $\Z/k^T\times \Z/2^f$ symmetry discussed here can be defined in terms of the triple $(G_b, s, \omega)$ where $G_b=\Z/k$, with $s$ nontrivial and $\omega$ trivial. Since $H^1(B\Z/k;\Z/2)\cong\Z/2$, this uniquely specifies a fermionic symmetry.} with $k$ a multiple of 4. In \cref{subsubsec:smithxtn}, we derive results of the bordism groups for general $k$. In \cref{app:manifold_generator}, we present the generating manifolds of the bordism groups. This sets up the calculation of the anomaly indicator in \S\ref{sec:anomalycomp}.

\begin{defn}
\label{epin_defn}
Let $k$ be divisible by $4$. An \term{epin$[k]$ structure} is a $(B\Z/k, \sigma)$-twisted spin structure. When $k = 4$ we will also refer to an epin$[4]$ structure as an \term{epin structure}.
\end{defn}
As discussed after Eq.~\eqref{eq:qdef}, here $\sigma$ is a line bundle on $B\Z/k$, defined as the pullback of the tautological bundle on $B\Z/2\cong B\O(1)$ across the nontrivial classifying map $Bs\colon B\Z/k \rightarrow B\Z/2$. 
The name ``epin'' is due to Wan-Wang-Zheng~\cite{WWZ20}, who studied epin$[4]$ structures.\footnote{Instead of defining epin structures as twisted spin structures, as we did in \cref{epin_defn}, Wan-Wang-Zheng define them using a group $\EPin\cong\Z/4\ltimes \Spin$ and a map $\EPin\to\O$~\cite[(1.2), (1.6)]{WWZ20}. It follows from the discussion in (\textit{ibid.}, \S 1) that the two definitions agree.}

According to \cref{prop:tangential_general}, the tangential structure involved in the classification of anomaly is an epin$[k]$ structure. It is straightforward to see that $w_1(\sigma)$ is nontrivial while $w_2(\sigma)$ is trivial, and that this uniquely characterizes them in $H^*(B\Z/k;\Z/2)$, so that the requirement in Eq.~\eqref{eq:requirement} about realizing $s$ and $\omega$ as Stiefel-Whitney classes is indeed satisfied. 

The anomaly of the $\Z/k^T\times \Z/2^f$ symmetry in (2+1)-d is hence classified by the Pontrjagin dual of $\Omega^{\EPin[k]}_4$. Remarkably, we will see in \cref{epink_thm_middle} that for all $k$ the associated Atiyah-Hirzeburch spectral sequences (and even Adams spectral sequences) have identical entries on all pages, yet the extension problems on the $E_\infty$-pages differ for different $k$.

To solve the extension problem, the Smith homomorphism serves a crucial role, and here we give a brief review of the Smith homomorphism together with the long exact sequence associated to it. We start with a simple lemma; recall the definition of $(X, V)$-twisted $\xi$-structures from \cref{def:twisted}.
\begin{lem}
\label{lemma:how_normal_twists}
Let $V,W\to X$ be vector bundles of ranks $r_V$, $r_W$, respectively, and suppose $M$ is a closed $n$-manifold with an $(X, V)$-twisted $\xi$-structure. If $i\colon N\hookrightarrow M$ is a closed $(n-r_W)$-submanifold of $M$ such that the mod $2$
fundamental class $i_*(N)\in H_{n-r_W}(M;\Z/2)$ is Poincaré dual to the mod $2$ Euler class $e(W)$, then the $(X, V)$-twisted $\xi$-structure on $M$ induces an $(X, V\oplus W)$-twisted $\xi$-structure on $N$.
\end{lem}
This follows directly from the fact that, since $[N]$ is Poincaré dual to the Euler class of $W$, the
normal bundle $\nu\to N$ of $N\hookrightarrow M$ is isomorphic to $W$; and $TM|_N\cong TN\oplus\nu$.

The conditions in \cref{lemma:how_normal_twists} typically do not uniquely determine the diffeomorphism class of $N$. However, with a little care, the assignment from $M$ to $N$ can be made compatible with \emph{bordism}.

Let $\Omega_\xi^*$ denote ``$\xi$-cobordism''~\cite{Ati61}, the generalized cohomology theory defined by the spectrum $\mathit{MT\xi}$ whose generalized homology theory is $\xi$-bordism. $\Omega_\xi^*$ is different from $\mho^*_\xi$, as the latter was built using Pontrjagin duality; the values of these two theories are very different even when evaluated on the point.
\begin{defn}
\label{Smith_defn}
Let $V,W\to X$ be vector bundles of ranks $r_V$, resp.\ $r_W$ and let $e^\xi(W)\in\Omega_\xi^{r_W}(X^{W-r_W})$ be the $\xi$-cobordism Euler class of $W$. Taking the cap product with $e^\xi(W)$ defines a homomorphism
\begin{equation}
S_W\colon \Omega_n^{\xi}(X^{V-r_V}) \longrightarrow \Omega_{n-r_W}^\xi(X^{V\oplus W - (r_V+r_W)}).
\end{equation}
This is called a \term{Smith homomorphism}.
\end{defn}
For the details of the definition of twisted $\xi$-cobordism Euler classes and the bordism-invariance of $S_W$, see~\cite[\S 4.1]{DDKLPTT2}. In particular, one needs stronger results than \cref{lemma:how_normal_twists} to get the theory of the Smith homomorphism off of the ground; we included \cref{lemma:how_normal_twists} to provide intuition for the more general construction in \cref{Smith_defn}.
\begin{rem}
\label{smith_caveat}
It would be nice to have a simpler description of $S_W$, due to the abstruseness of Euler classes in twisted generalized cohomology. This is often possible.
\begin{enumerate}
    \item Euler classes are natural in maps of spectra, so given a map of spectra $\phi\colon \mathit{MT\xi}\to E$, the image of the Smith homomorphism under $\phi$ is the cap product with the $E$-cohomology Euler class. Usually one chooses $E$ to be $H\Z$ or $H\Z/2$, sending $e^\xi(W)$ to the usual Euler class, resp.\ top Stiefel-Whitney class, of $W$, in order to only worry about cap products in ordinary homology.
    \item Both bordism classes and ordinary homology classes of a manifold $M$ can often be represented by maps of manifolds $N\to M$. The Smith homomorphism can then be recast as asking, given $M$ and $W\to M$, find a manifold $M$ and a map $f\colon N\to M$ whose bordism or homology class is Poincaré dual to the Euler class of $W$. Then the Smith homomorphism sends the $V$-twisted $\xi$-bordism class of $M$ to the $(V\oplus W)$-twisted $\xi$-bordism class of $N$.
\end{enumerate}
These two facts lead to the usual interpretation of the Smith homomorphism as ``taking the Poincaré dual of the Euler class/of the top Stiefel-Whitney class''.

In most cases, using Euler classes in $\Z$ or $\Z/2$ cohomology, rather than in cobordism, suffices; this includes all Smith homomorphisms studied in this paper. But there are examples where one must use a better approximation to $\xi$-cobordism to correctly define the Smith homomorphism. One such example appears in~\cite[Appendix B]{DDKLPTT2}.
\end{rem}
\begin{thm}[\cite{DDKLPTT2}]
\label{Smith_LES}
Let $V, W\to X$ be vector bundles of ranks $r_V$, $r_W$, respectively, and $p\colon S(W)\to X$ be the sphere bundle of $W$.
Then there is a long exact sequence
\begin{equation}\label{eq:LESsmith}
	\dotsb \rightarrow \Omega_k^\xi(S(W)^{p^*V-r_V}) \overset{p_*}{\rightarrow} \Omega_k^\xi(X^{V-r_V})
	\overset{S_W}{\rightarrow} \Omega_{k-r_W}^\xi(X^{V\oplus W -(r_V+r_W)})\rightarrow
	\Omega_{k-1}^\xi(S(W)^{p^*V-r_V})\rightarrow \dotsb
\end{equation}
\end{thm}
This long exact sequence connects bordism groups with different twists and in different dimensions, hence if one case is easier to determine, the long exact sequence can be very helpful to deriving results in other cases. We will apply \cref{Smith_LES} a few times in this paper with $X=B\Z/k$, and we need the following lemma regarding its sphere bundle.

\begin{lem}
\label{which_sphere_bundle}
Given a short exact sequence, $1\rightarrow \hat{G}\rightarrow G\rightarrow \Z/2 \rightarrow 1$, let $\sigma$ be the 1-dimensional line bundle on $G$ defined as the pullback of the tautological bundle on $B\Z/2\cong B\O(1)$. The map $S(\sigma)\to BG$ is homotopy equivalent to the map $B\hat{G}\to BG$ induced by the inclusion
$\hat{G}\hookrightarrow G$.
\end{lem}
Compare~\cite[Lemma C.2]{DL23}.
\begin{rem} 
If the projection $G\to \Z/2$ corresponds to an element $s\in H^1(BG;\Z/2)$ as in Eq.~\eqref{eq:qdef}, then $\hat{G}$ is the subgroup of \textit{unitary symmetries}.
\end{rem}
\begin{proof}[Proof of \cref{which_sphere_bundle}]
The sphere bundle of $\sigma$ pulls back from the sphere bundle of the universal line bundle $L\to B\Z/2$ across
the classifying map $f\colon BG\to B\Z/2$ for $\sigma$. But $S(L) = E\Z/2$, which is contractible. The homotopy pullback of a diagram $B\overset f\to D\overset g\gets C$ such that $C$ is contractible is the homotopy fiber of $f$, so
$S(\sigma)\to BG$ is the homotopy fiber of the map $BG\to B\Z/2$ induced by the quotient $G\to\Z/2$. The
classifying space functor turns short exact sequences into fiber sequences, and applying this to $ 1\to \hat{G}\to G \to \Z/2\to 1$, we can conclude.
\end{proof}

\subsection{\texorpdfstring{Computing $\Omega^{\EPin[k]}_4$: Spectral Sequences and the Smith Homomorphism}{}}\label{subsubsec:smithxtn}

First of all, we collect some results of group cohomology of $\Z/k$, which will be needed when writing down entries of the spectral sequence. 

We recall that if $4\mid k$, 
\begin{equation}\label{eq:Z/kcoho}
    H^*(B\Z/k;\Z/2)\cong\Z/2[x, y]/(x^2)\,,\quad \abs x = 1,~ \abs y = 2\,,
\end{equation}
\begin{equation}
    H^*(B\Z/k;\Z)\cong \Z[\bar{y}]/(k\bar{y})\,,\quad \abs {\overline y} = 2\,,
\end{equation}
and $\overline y \bmod 2 = y$. Moreover, $s=w_1(\sigma)=x$. Now we deal with the twisted integral cohomology groups.
\begin{lem}
\label{tw_BZk_lem}
As a module over $A_k\coloneqq H^*(B\Z/k;\Z)\cong \Z[\bar{y}]/(k\bar{y})$ with $\abs {\overline y} = 2$,
\begin{equation}
\label{twisted_BZk}
    H^*(B\Z/k; \Z_{s}) \cong (\Sigma A_k\cdot \overline x)/(2\overline x).
\end{equation}
The class $\overline x\in H^1(B\Z/k; \Z_{s})$ is the twisted Euler class of $\sigma\to B\Z/k$, and $\overline{x} \bmod 2 =x$.
\end{lem}
To unpack the notation in Eq.~\eqref{twisted_BZk} a bit: 
$\Sigma^1A_k$ means to take a copy of $A_k$ and raise the grading by $1$; $\Sigma^1 A_k\cdot \overline x$ is generated as an $A_k$-module by $\overline x$ in degree $1$. Then we quotient by $2\overline x$, so
the twisted cohomology groups of $B\Z/k$ begin $0$, $\Z/2$, $0$, $\Z/2$, $0$, $\dots$, with the copies of $\Z/2$ generated by $\overline x$, $\overline y\overline x$, ${\overline{y}}^2\overline x$, and so on.
\begin{proof}[Proof sketch of \cref{tw_BZk_lem}]
Following Čadek~\cite[Lemma 1]{Cad99}, consider the Gysin sequence for $\sigma\to B\Z/k$. The sphere bundle of $\sigma$ is homotopy equivalent to $B\Z/(k/2)$, so the Gysin sequence is a long exact sequence of the form
\begin{equation}
\label{baby_Smith}
\begin{tikzcd}[column sep = 0.3cm]
	\dotsb & {H^{n-1}(B\Z/k; \Z_{s})} & {H^n(B\Z/k;\Z)} & {H^n(B\Z/(k/2);\Z)} & {H^n(B\Z/k; \Z_{s})} & \dotsb
	\arrow["{\cdot\overline x}", outer sep=3pt, from=1-2, to=1-3]
	\arrow[from=1-3, to=1-4]
	\arrow[from=1-4, to=1-5]
	\arrow[from=1-1, to=1-2]
	\arrow[from=1-5, to=1-6]
\end{tikzcd}\end{equation}
Studying the effect of the map $B\Z/(k/2)\to B\Z/k$ on cohomology, we see that Eq.~\eqref{baby_Smith} breaks into a bunch of short exact sequences, from which the lemma follows.
\end{proof}

\begin{rem}
This Gysin sequence is closely analogous to the Smith long exact sequence in \cref{Smith_LES}, just with bordism replaced with homology. Hence this problem can be thought of as our first example solved by the Smith homomorphism.
\end{rem}
\begin{rem}
Though the twisted cohomology groups in \cref{tw_BZk_lem} do not form a ring, the product of two classes in
$\Z_{s}$-cohomology lands in untwisted $\Z$-cohomology, inducing a $(\Z\times\Z/2)$-graded ring
structure on $H^*(B\Z/k;\Z\oplus\Z_{s})$, as observed by Čadek~\cite[\S 1]{Cad99} (see also
Costenoble-Waner~\cite{CW92, CW16}). It is possible to extend \cref{tw_BZk_lem} to show
\begin{equation}
	H^*(B\Z/k;\Z\oplus\Z_{s}) \cong \Z[\overline x, \overline y]/(2\overline x, k\overline y, (k/2)\overline y - \overline x{}^2),
\end{equation}
with $\abs{\overline x} = (1, 1)$ and $\abs{\overline y} = (2, 0)$, for example by using the local coefficients Serre
spectral sequence~\cite[Theorem 2.19]{Sie67} in a manner similar to~\cite[Theorem 5.49]{Deb21} and~\cite[Appendices A.4, E.5.a.b, E.5.b.b]{MCB22}; we do not need
this extra structure, so do not prove it.
\end{rem}
\begin{cor}
\label{sea_star_coh}
Let $\beta_{\U}\colon H^k(\bl;\U)\to H^{k+1}(\bl;\Z)$ denote the Bockstein associated to the short exact sequence
\begin{equation}
\label{exp_seq}
   \shortexact*[][e^{2\pi i(\bl)}]{\Z}{\R}{\U}.
\end{equation}
Then, as an $A_k$-module,
\begin{equation}
    H^*(B\Z/k; (\U)_{s})\cong( 
    A_k\cdot\overline 1)/(2\overline 1),
\end{equation}
meaning the twisted $\U$-valued cohomology groups of $B\Z/k$ begin $\Z/2$, $0$, $\Z/2$, $0$, \dots, with the copies of $\Z/2$ generated by classes $\overline 1$, $\overline y \overline 1$, $\overline y{}^2 \overline 1$, etc, and the $\beta_{\U}$-image of $\overline 1$ is $\overline x$.
\end{cor}
\begin{proof}
Plug \cref{tw_BZk_lem} into the long exact sequence in cohomology associated to the exponential exact sequence~\eqref{exp_seq}. The result follows as soon as we know $H^*(B\Z/k;\R_{s})$ vanishes in all degrees, which follows from \cref{tw_BZk_lem} and the universal coefficient theorem.
\end{proof}

\begin{prop}[{Botvinnik-Gilkey~\cite[\S 5]{BG97}}]
\label{order_4}
For all $k$, $\Omega_4^{\EPin[k]}$ has order $4$.
\end{prop}

We obtain this result through the Atiyah-Hirzeburch spectral sequence, which is a different technique from Botvinnik-Gilkey. We include this proof because we will use the details of our argument later, both in \cref{epink_thm_middle} to finish the computation of $\Omega_4^{\EPin[k]}$ and in Appendix~\ref{app:cascade} to provide an interpretation of the layers of the Atiyah-Hirzebruch spectral sequence in the context of anomalies of fermionic topological order.

We actually show that $\mho_{\EPin[k]}^4$ has order $4$, which by \cref{upside_down} is equivalent to \cref{order_4}. We have two reasons for our change to Pontrjagin-dualized bordism: to simplify the differentials and to make contact with a physically motivated interpretation of this spectral sequence due to \cite{Bulmash:2021ryq,Wang2018}. This technique is spiritually similar to a strategy of Campbell~\cite[\S 7.4]{Cam17}, also used in~\cite[\S 5.1]{FH20} and~\cite[\S 5.3.1]{Deb21}; heuristically, the difference is whether $\U$ carries the discrete topology, as it does for us, or the usual topology, as it does for Campbell.

By \cref{shearing_lem,epin_defn}, there is an isomorphism $\mho_{\EPin[k]}^4\cong \mho_\Spin^4((B\Z/k)^{\sigma-1})$ natural in $k$. Thus we will study the Atiyah-Hirzebruch spectral sequence
\begin{equation}
\label{PD_trunc_AHSS}
E_2^{p,q} = H^p(X; \mho_\Spin^q) \Longrightarrow \mho_\Spin^{p+q}(X),
\end{equation}
where $X=B\Z/k$. Using~\eqref{IU_coh_2}, the coefficient groups $\mho_\Spin^*$ begin $\U$, $\Z/2$, $\Z/2$, $0$, $\U$, $0$, $0$, $0$ in degrees $0$ through $7$. \begin{lem}
\label{diffs_forms}
Let $i\colon \Z/2\to\U$ be the unique injective homomorphism, $r_2\colon\Z\to\Z/2$ be reduction mod $2$, and $\beta_{\U}$ be the Bockstein from \cref{sea_star_coh}. Then, in~\eqref{PD_trunc_AHSS},
\begin{enumerate}
    \item $d_2\colon E_2^{p,1}\to E_2^{p+2,0}$ is identified with the map $i\circ\Sq^2\colon H^p(X;\Z/2)\to H^{p+2}(X;\U)$,
    \item\label{df_2} $d_2\colon E_2^{p,2}\to E_2^{p+2,1}$ is identified with $\Sq^2\colon H^p(X;\Z/2)\to H^{p+2}(X;\/2)$, and
    \item\label{df_3} $d_3\colon E_3^{p,4}\to E_3^{p+3, 2}$ is identified with the map $H^p(X;\U)\to \ker(\Sq^2)\subset H^{p+3}(X;\Z/2)$ given by $\Sq^2\circ r_2\circ \beta_{\U}$.
\end{enumerate}
\end{lem}
\begin{rem}
The statement of part~\eqref{df_3} of \cref{diffs_forms} relies on the proof of part~\eqref{df_2}: because $E_2^{*,3} = 0$ and $d_2$ out of $E_2^{*, 2}\cong H^*(X;\Z/2)$ is identified with $\Sq^2$, $E_3^{p,2}\cong \ker(\Sq^2)\subset H^p(X;\Z/2)$ as promised in~\eqref{df_3}. We will prove part~\eqref{df_2} without reference to part~\eqref{df_3}, so there is no circular logic.
\end{rem}
\begin{proof}[Proof of \cref{diffs_forms}]
Let $I_\Z\colon\cat{Sp}^{\mathrm{op}}\to\cat{Sp}$ denote the \term{Anderson duality} functor~\cite{And69, Yos75}; then, there is a map $\alpha\colon I_{\U}X\to \Sigma I_\Z X$ whose fiber is $\mathrm{Map}(X, H\C)$~\cite[Definition 4.11]{And69}.\footnote{The map $\alpha$ and its fiber are sometimes stated differently, e.g.\ using $\Q/\Z$ and $H\Q$ instead of $\U$ and $H\C$, respectively, or $\R/\Z$, resp. $H\R$, or $\C^\times$, resp.\ $H\C$. In all cases the arguments are essentially unchanged.}\textsuperscript{,}\footnote{There is a sense in which $I_\Z$ is the analogue of $I_{\U}$, but with $\U$ given the continuous topology, and the fiber sequence $\mathrm{Map}(X, H\C)\to I_{\U} X\to \Sigma I_\Z X$ is a rotated counterpart of the exponential sequence. See, e.g.,~\cite[\S 5.3]{FH21InvertibleFT}.} Consider the following zigzag of maps of spectra:
\begin{equation}
\label{zigzag}
    I_{\U}\MTSpin \overset{\phi}{\longleftarrow}
    I_{\U}\ko \overset{\alpha}{\longrightarrow}
    \Sigma I_\Z\ko \overset{\zeta}{\longleftarrow}
    \Sigma I_\Z\KO \underset{\simeq}{\overset{\chi}{\longrightarrow}}
    \Sigma^5\KO,
\end{equation}
where $\alpha$ is as above and
\begin{itemize}
    \item $\phi$ is the Pontrjagin dual of the Atiyah-Bott-Shapiro map $\MTSpin\to\ko$~\cite{ABS64},
    \item $\zeta$ is the Anderson dual of the connective cover map $\ko\to\KO$, and
    \item $\chi$ is the map implementing the shifted Anderson self-duality of $\KO$~\cite[Theorem 4.16]{And69}.\footnote{Anderson's lecture notes are unpublished; see Yosimura~\cite[Theorem 4]{Yos75} for a published account of Anderson's proof. There are also at least four other proofs of the shifted Anderson self-duality of $\KO$, each by very different methods, due to Freed-Moore-Segal~\cite[Proposition B.11]{FMS07}, Heard-Stojanoska~\cite[Theorem 8.1]{HS14}, Ricka~\cite[Corollary 5.8]{Ric16}, and Hebestreit-Land-Nikolaus~\cite[Example 2.8]{HLN20}.}
\end{itemize}
The differentials corresponding to ours in the $\KO$-cohomology Atiyah-Hirzebruch spectral sequence were computed by Bott~\cite{Bot69} (see Anderson-Brown-Peterson~\cite[Proof of Lemma 5.6]{ABP67} for an explicit description), and we will chase them through~\eqref{zigzag} to determine the differentials in the theorem statement. Specifically,
in the $\KO$-cohomology Atiyah-Hirzebruch spectral sequence,
\begin{enumerate}
    \setcounter{enumi}{-1}
    \item $d_2\colon E_2^{p,0}\to E_2^{p+2, -1}$ is $\Sq^2\circ r\colon H^p(X;\Z)\to H^{p+2}(X;\Z/2)$, where $r$ is the reduction mod $2$ map;
    \setcounter{enumi}{-2}
    \item $d_2\colon E_2^{p,-1}\to E_2^{p+2, -2}$ is $\Sq^2\colon H^p(X;\Z/2)\to H^{p+2}(X;\Z/2)$;
    \setcounter{enumi}{-3}
    \item $d_3\colon E_3^{p,-2}\to E_3^{p+3, -4}$ is $\beta\circ\Sq^2\colon H^p(X;\Z/2)\to H^{p+3}(X;\Z)$, where $\beta$ is the Bockstein $H^*(\bl;\Z/2)\to H^{*+1}(\bl;\Z)$.
\end{enumerate}
We need $\Sigma^5\KO$, which amounts to adding $5$ to $q$ in the above formulas.

Anderson-Brown-Peterson~\cite{ABP67} showed that the Atiyah-Bott-Shapiro map induces an isomorphism on homology groups in degrees $7$ and below, so by the universal property of $I_{\U}$, $\phi$ is an isomorphism on homotopy groups in degrees $-7$ and above, i.e.\ is an isomorphism on generalized cohomology of a point in cohomological degrees $7$ and below. This implies that in the map of cohomological Atiyah-Hirzebruch spectral sequences induced by $\phi$, the $E_2$-pages coincide for $q\le 7$ and this identification commutes with all differentials out of $E_r^{p,q}$ for $q\le 7$. Therefore in the degrees we care about, the effect of $\alpha$ on differentials may as well be the identity map.

Likewise, by definition $\ko\to\KO$ is an isomorphism on homotopy groups in degrees $0$ and above. By invoking the universal property of Anderson duality~\cite[Lemma 4.13]{And69} or by explicitly tracing through the long exact sequence induced by $H\C\wedge X\to I_{\U}X\to \Sigma I_\Z X$, we learn $\zeta$ is an isomorphism on homotopy groups in degrees $1$ and below, i.e.\ is an isomorphism on generalized cohomology of a point in degrees $-1$ and above. Arguing as we did for $\phi$, we learn $\zeta$ is an isomorphism on $E_2$-pages of Atiyah-Hirzebruch spectral sequences in degrees $-1$ and above, and commutes with all differentials emerging from $E_r^{p,q}$ with $q\ge -1$. Once again in the range we care about we can treat this as the identity.

That leaves $\chi$ and $\alpha$, and $\chi$ is a homotopy equivalence, so up to isomorphism will not change Atiyah-Hirzebruch differentials; we focus on $\alpha$. We saw above that the fiber of $\alpha$ is $F\coloneqq \mathrm{Map}(\ko, H\C)$; by definition $F^*(\mathrm{pt})\cong H^*(\ko;\C)$, which is isomorphic to $\C$ in nonnegative degrees $0\bmod 4$ and otherwise vanishes. Using this and the long exact sequence in cohomology of a point associated to the cofiber sequence
\begin{equation}
\label{ko_HC_cofib}
    \mathrm{Map}(\ko, H\C)\longrightarrow I_{\U}\ko\longrightarrow \Sigma(I_\Z\ko),
\end{equation}
we see that $\alpha$ is an isomorphism in cohomological degrees $1$ and $2$; comparing with the $\KO$-cohomology differential that we described above, we have proven part~\eqref{df_2} of the theorem.

For the other two differentials in the theorem statement, we use the long exact sequence in cohomology associated to~\eqref{ko_HC_cofib} to see that, when passing from $\Sigma I_\Z\ko$ to $I_{\U}\ko$ in degrees $-1$ and $0$ and in degrees $3$ and $4$, we must precompose with $\beta_{\U}$, finishing the proof.
\end{proof}

\begin{proof}[Proof of \cref{order_4}]
Now we can directly write the $E_2$-page of our Atiyah-Hirzebruch spectral sequence Eq.~\eqref{PD_trunc_AHSS}. We do so in \cref{eq:AHSScomp}, where $x$, $y$, $\overline 1$, and $\overline y$ are as in Eq.~\eqref{eq:Z/kcoho} and \cref{sea_star_coh}.

Since we are running the Atiyah-Hirzebruch spectral sequence for $(B\Z/k)^{\sigma-1}$ as a Thom spectrum, we need to pass through the Thom isomorphism. Let $U\in H^0((B\Z/k)^{\sigma-1};\Z/2)$ denote the mod $2$ Thom class; because $\sigma$ is not orientable, we do not have an integral Thom class, so given $\gamma\in H^*(B\Z/k;(\U)_{s})$, we let $\overline{U}\gamma\in H^*((B\Z/k)^{\sigma-1};\U)$ denote the image of $\gamma$ under the Thom isomorphism. We can now write down the entries in Eq.~\eqref{PD_trunc_AHSS}. 
\begin{figure}[!htbp]
\begin{sseqdata}[name=epin, cohomological Serre grading, classes = {draw = none}, xrange={0}{5}, xscale=1.8, yrange={0}{4}, axes
type=center, y axis origin=0, x tick gap=1em, x axis extend end=0.75cm, >=stealth]
    \class["\overline{U} {\overline 1}"](0, 0)
    \class["\overline{U} \overline{y}{\overline 1}"](2, 0)
\class["\overline{U} \overline{y}^2{\overline 1}"](4, 0)

\class["\overline{U}{\overline 1}"](0, 4)
\class["\overline{U}\overline{y}{\overline 1}"](2, 4)
\class["\overline{U}\overline{y}^2{\overline 1}"](4, 4)

\class["U"](0, 1)
\class["Ux"](1, 1)
\class["Uy"](2, 1)
\class["Uxy"](3, 1)
\class["Uy^2"](4, 1)
\class["Uxy^2"](5, 1)

\class["U"](0, 2)
\class["Ux"](1, 2)
\class["Uy"](2, 2)
\class["Uxy"](3, 2)
\class["Uy^2"](4, 2)
\class["Uxy^2"](5, 2)
\end{sseqdata}
\centering
\printpage[name=epin, page=2]
\caption{The Atiyah-Hirzebruch spectral sequence computing $\mho_\Spin^*((B\Z/k)^{\sigma-1})$.}
\label{eq:AHSScomp}
\end{figure}

We have formulas for the differentials thanks to \cref{diffs_forms}, and now we evaluate them. As described above, $\beta_{\U}$ is an isomorphism in positive degrees in the twisted cohomology of a finite group, so we will suppress it in the arguments below.

Recall that because $\sigma-1$ has rank $0$, the mod $2$ Thom isomorphism $a\mapsto Ua$ is a degree-preserving isomorphism of graded abelian groups $H^*(B\Z/k;\Z/2)\to H^*((B\Z/k)^{\sigma-1};\Z/2)$, but it
does not commute with the Steenrod squares. Instead, $\Sq^1(U) = Uw_1(\sigma-1)$ and
$\Sq^2(U) = Uw_2(\sigma-1)$~\cite[Théorème II.2]{Tho52}; the values of Steenrod squares on
classes of the form $\overline U a$ for $a\in H^*(B\Z/k;\Z/2)$ then follow from the Cartan formula. We have
$w_1(\sigma-1) = x$ and $w_2(\sigma-1) = 0$. Using this formula, we obtain the following differentials.
\begin{enumerate}
    \item For the $d_3$ beginning in degree $q = 4$: the Thom isomorphism commutes with $r$ and $\beta$, 
    so $d_3(\overline U \overline 1) = \Sq^2(Ux) = 0$. 
    \item For the $d_2$ beginning in degree $q = 3$, $d_2(Ux) = 0$ and $d_2(Uy) = Uy^2$.
    \item For the $d_2$ beginning in degree $q = 2$, $d_2(Uy)=\overline U \overline y{}^2 \overline 1$, and 
    $d_2(Uxy)=0$ for degree reason. 
\end{enumerate}
Thus the piece of the $E_4$-page in total degree $4$ consists of $E_4^{0,4} = \Z/2\cdot \overline U\overline 1$ and $E_4^{2,2} = \Z/2\cdot Uy$. All higher differentials vanish on these two summands for degree reasons, so $E_\infty^{p,4-p}$ consists of two $\Z/2$ summands, implying $\mho_\Spin^4((B\Z/k)^{\sigma-1})$ is an abelian group of order $4$.
\end{proof}
\begin{rem}
The extension question raised by \cref{order_4} is a common feature of Atiyah-Hirzebruch spectral sequence computations of twisted spin bordism groups in low degrees; for example, in the analogous computation of $\mho^2_{\Pin^-}\cong\Z/8$, the Atiyah-Hirzebruch spectral sequence reports this $\Z/8$ as three $\Z/2$ summands on the $E_\infty$-page, and likewise $\mho^4_{\Pin^+}\cong\Z/16$ is broken into four $\Z/2$ summands on the $E_\infty$-page of its Atiyah-Hirzebruch spectral sequence. The same is true for the homologically graded Atiyah-Hirzebruch spectral sequences computing bordism groups.

Typically one can resolve these extension questions using the Adams spectral sequence, whose extra structure can often be used to disambiguate these extension problems; see Beaudry-Campbell \cite[\S 4.8, Figure 30]{BC18}. For example, Campbell~\cite[Theorems 6.4 and 6.7]{Cam17} runs the Adams spectral sequences computing pin\textsuperscript{$+$} and pin\textsuperscript{$-$} bordism in low degree: the $\Z/8$ and $\Z/16$ of interest are visible on the $E_\infty$-page.\footnote{For more examples contrasting the Atiyah-Hirzebruch and Adams spectral sequences when computing twisted spin bordism, see~\cite[Appendices E and F]{KPMT20} and~\cite{Ped17}.}

For $\EPin[k]$, Botvinnik-Gilkey~\cite[\S 5]{BG97} and (for $k = 4$) Wan-Wang-Zheng~\cite[\S B.1]{WWZ20} run the Adams spectral sequence to compute $\Omega_4^{\EPin[k]}$. However, the Adams spectral sequence does not resolve the ambiguity on its $E_\infty$-page: there is the potential for a ``hidden extension,'' and without addressing it, one cannot distinguish between $\Z/2\oplus\Z/2$ and $\Z/4$. So we have to try something different.
\end{rem}
It turns out that the answer depends on $k$.
\begin{thm}\hfill
\label{epink_thm_middle}
\begin{enumerate}
	\item\label{epink_4} If $k \equiv 4\bmod 8$, $\Omega_4^{\EPin[k]}\cong\Z/4$.
	\item\label{epink_8} If $k\equiv 0\bmod 8$, $\Omega_4^{\EPin[k]}\cong\Z/2\oplus\Z/2$.
\end{enumerate}
\end{thm}
For $k = 4$, \cref{epink_thm_middle}
corrects a mistake in the literature (e.g.\ \cite[\S B.1]{WWZ20}), where this group was claimed to be isomorphic to
$\Z/2\oplus\Z/2$. We will discuss manifold representatives for generators of $\Omega_4^{\EPin[k]}$ in~\S\ref{app:manifold_generator}.
\begin{proof}[Proof of \cref{epink_thm_middle}]
We will prove part~\eqref{epink_4} first, then use it to prove part~\eqref{epink_8}.

We will fit $\Omega_4^{\EPin[k]}$ into a long exact sequence from \cref{Smith_LES}, whose other terms are already known. In \cref{Smith_LES}, choose $\xi$ to be spin-structures, $X = B\Z/k$, $V = 0$, and $W = \sigma$ that was introduced as the pullback of the tautological line bundle from $B\Z/2$. By directly plugging this data into Eq.~\eqref{eq:LESsmith} we get a long exact sequence
\begin{equation}
\label{epin_specific_LES}
	\dotsb\longrightarrow \Omega_m^\Spin(B\Z/(k/2))\longrightarrow \Omega_m^\Spin(B\Z/k)\longrightarrow
	\Omega_{m-1}^{\EPin[k]}\longrightarrow \Omega_{m-1}^\Spin(B\Z/(k/2))\longrightarrow\dotsb\,.
\end{equation}

For $k=4$, we can directly write down many entries in the long exact sequence: work of Mahowald-Milgram~\cite{MM76} implies $\Omega_4^{\Spin}(B\Z/2)\cong\Z$ and $\Omega_5^{\Spin}(B\Z/2) =
0$,\footnote{See also Mahowald~\cite[Lemma 7.3]{Mah82}, Bruner-Greenlees~\cite[Example 7.3.1]{BG10},
Siegemeyer~\cite[Theorem 2.1.5]{Sie13}, and García-Etxebarria-Montero~\cite[(C.18)]{Garcia-Etxebarria:2018ajm} for
other works computing spin bordism or $\ko$-homology of $B\Z/2$ in these degrees.} and
Bruner-Greenlees~\cite[Example 7.3.3]{BG10} show $\Omega_5^\Spin(B\Z/4)\cong\Z/4$.\footnote{Siegemeyer~\cite[\S
2.2]{Sie13} and Davighi-Lohitsiri~\cite[\S A.3]{Davighi:2020uab} provide additional computations of this group via
other methods.}
Therefore, the long exact sequence~\eqref{epin_specific_LES} has the form
\begin{equation}
\label{epin_smith_final_LES}
	0\longrightarrow \Z/4\overset {S_\sigma} \longrightarrow \Omega_4^{\EPin[k]}\overset\varphi
	\longrightarrow \Z.
\end{equation}
We already know $\Omega_4^{\EPin}$ is finite, so $\varphi = 0$ and $S_\sigma$ is an isomorphism. Therefore, we immediately establish that $\Omega_4^{\EPin}=\Omega_4^{\EPin[4]}\cong \Z/4$.

For other values of $k\equiv 4\bmod 8$, the Smith homomorphism $S_\sigma\colon \Omega_m^\Spin(B\Z/k)\rightarrow
	\Omega_{m-1}^{\EPin[k]}$ is not an isomorphism, but we can still follow the calculation of $k=4$ after applying the trick of localizing at $2$.\footnote{This is a mathematical procedure whose effect on a finitely generated abelian group sends free summands to free summands, preserves all $2$-power torsion, and sends all odd-power torsion to $0$. The reader is welcome to take this as a definition of localization in this paper; for a more general introduction to localization, see~\cite{atiyah1994introduction}.} If $k\equiv 4\bmod 8$, then
$k/2 $ is two times an odd integer, so the inclusion $B\Z/2\hookrightarrow B\Z/(k/2)$ induces an isomorphism on
$\Z_{(2)}$-cohomology, hence also on $2$-localized generalized homology (e.g.\ using the Atiyah-Hirzebruch spectral
sequence). Therefore for $k\equiv 4\bmod 8$, we can replace $B\Z/(k/2)$ with $B\Z/2$ in Eq.~\eqref{epin_specific_LES}. And the long exact sequence~\eqref{epin_specific_LES} has the form
\begin{equation}
\label{epink_smith_final_LES}
	0\longrightarrow \Z/4\overset{S_\sigma} \longrightarrow \Omega_4^{\EPin[k]}\otimes\Z_{(2)} \overset\varphi
	\longrightarrow \Z_{(2)}.
\end{equation}
Thus, the localization of $\Omega_4^{\EPin[k]}$ at $2$ is isomorphic to $\Z/4$; since we know from \cref{order_4} that $\Omega_4^{\EPin[k]}$ is either $\Z/2\oplus\Z/2$ or $\Z/4$, we conclude that for all $k\equiv 4 \bmod 8$, $\Omega_4^{\EPin[k]}\cong\Z/4$.

When $8\mid k$, 
the algebraic
information of the bordism groups in the Smith long exact sequence is not enough to clarify whether
$\Omega_4^{\EPin[k]}$ is isomorphic to $\Z/4$ or to $\Z/2\oplus\Z/2$. Thus we have to do something different: study the map of Atiyah-Hirzebruch spectral sequences induced by the map
$p\colon B\Z/k\to B\Z/4$. This map is the pullback map on cohomology on the $E_2$-page and commutes with all differentials,
giving us an induced map on the $E_\infty$-page which is compatible with the filtration on
$\Omega_{p+q}^{\EPin[k]}$. Like in the proof of \cref{order_4}, we will study the $\mho_\Spin^*$ Atiyah-Hirzebruch spectral
sequence; see that proof for the strategy and notation. We will let $p^*$ denote the pullback map associated to $p$ in both cohomology and Pontrjagin-dualized bordism; the specific map will always be clear from context.

For any $\ell$ divisible by $4$, let
$E_*^{*,*}(\ell)$ denote the spectral sequence for $\mho_\Spin^*((B\Z/\ell)^{\sigma-1})$. In the proof of \cref{order_4},
we saw that for both $\ell = 4$ and $\ell = k$, $E_\infty^{\bullet,4-\bullet}$ has two nonzero summands: $E_\infty^{0,4}
\cong \Z/2$ and $E_\infty^{2,2}\cong \Z/2$. Thus in both spectral
sequences, the group we want to compute is an extension of $E_\infty^{0,4}$ by $E_\infty^{2,2}$, and compatibility
of these filtrations with the map $(B\Z/k)^{\sigma-1}\to (B\Z/4)^{\sigma-1}$ implies that there is a commutative diagram of short exact sequences
\begin{equation}
\label{AHSS_Einf_map}
\begin{tikzcd}
	0 & {E_\infty^{2,2}(4)} & \mho_\Spin^4((B\Z/4)^{\sigma-1}) & {E_\infty^{0,4}(4)} & 0 \\
	0 & {E_\infty^{2,2}(k)} & \mho_\Spin^4((B\Z/k)^{\sigma-1}) & {E_\infty^{0,4}(k)} & 0
	\arrow[from=1-1, to=1-2]
	\arrow["f_0", from=1-2, to=1-3]
	\arrow["g_0", from=1-3, to=1-4]
	\arrow[from=1-4, to=1-5]
	\arrow[from=2-1, to=2-2]
	\arrow["f", from=2-2, to=2-3]
	\arrow["g", from=2-3, to=2-4]
	\arrow[from=2-4, to=2-5]
	\arrow["p^*_2", from=1-4, to=2-4]
	\arrow["p^*", from=1-3, to=2-3]
	\arrow["p^*_1", from=1-2, to=2-2]
\end{tikzcd}
\end{equation}
We know $\mho_\Spin^4((B\Z/4)^{\sigma-1})\cong\Z/4$, so $f_0\colon\Z/2\to\Z/4$ sends $1\mapsto 2$ and
$g_0\colon\Z/4\to\Z/2$ is reduction mod $2$. Next we want to understand the vertical arrows.
\begin{itemize}
	\item Since $E_\infty^{2,2}(k)$ is generated by the class $Uy$, $p_1^*$ is determined by the image of
	$Uy$ under the pullback $H^2((B\Z/4)^{\sigma-1};\Z/2)\to H^2((B\Z/k)^{\sigma-1};\Z/2)$. Naturality
	of the Thom isomorphism implies that it suffices to compute the pullback of $y$ under $p^*\colon H^2(B\Z/4;\Z/2)\to
	H^2(B\Z/k;\Z/2)$. It turns out $p^*(y) = 0$,\footnote{One way to see why $p^*(y) = 0$ is as follows: if $V_\rho\to B\Z/4$ is the complex line bundle associated to the rotation representation of
	$\Z/4$ on $\C$, then $y = w_2(V_\rho)$, so to show $p^*(y) = 0$, it suffices to show that $p^*(V_\rho)$ is spin, i.e.\
	that the corresponding map $\Z/k\to\Z/4\to\SO(2)$ lifts across the double cover $\Spin(2)\to\SO(2)$. This can be
	done, e.g. by\ sending $1\in\Z/k$ to an eighth root of unity in $\U \cong \Spin(2)$.} so $Uy\mapsto 0$ as well and $p_1^* = 0$.
	\item By contrast, $E_\infty^{0,4}(k)$ is generated by the class $\overline{U}\overline 1\in
	H^0((B\Z/k)^{\sigma-1};\U_s)$; 
 naturality of $\beta_{\U}$ implies we may as well check on the corresponding class in $H^1((B\Z/k)^{\sigma-1};\Z_s)$, which is the Euler class of $\sigma$. The
	map $\Z/k\to\Z/2$ factors through $\Z/4$, so $p_2^*$ is an isomorphism. 
\end{itemize}

Suppose $\mho_\Spin^4((B\Z/k)^{\sigma-1})\cong\Z/4$. Then we still have $f\colon\Z/2\to\Z/4$ sends $1\mapsto 2$ and
$g\colon\Z/4\to\Z/2$ is reduction mod $2$. The commutativity of the right-hand square in~\eqref{AHSS_Einf_map} implies $p_2^*\circ g_0 = g\circ p^*$, so $p^*\colon\Z/4\to\Z/4$ must map $1$ to either $1$ or $-1$. However, using the commutativity of the left square, $p^*\circ f_0 = f\circ p_1^*$, so $p(1) = \pm 1$ cannot be satisfied. Therefore $\mho_\Spin^4((B\Z/k)^{\sigma-1})\cong\Z/2\oplus \Z/2$.
\end{proof}

\begin{rem}
There are a few other ways to show that the extension in the Atiyah-Hirzebruch spectral sequence splits when $8\mid
k$. One can write down a very similar proof by studying the analogous comparison map on Adams spectral sequences,
for example, using Botvinnik-Gilkey's description~\cite[\S 5]{BG97} of the $E_\infty$-page of the Adams spectral
sequences for $(B\Z/k)^{\sigma-1}$ for all $k\equiv 0\bmod 4$. Alternatively, studying the Smith long exact
sequence~\eqref{epin_specific_LES} produces an isomorphism
\begin{equation}
\label{coker_map}
 \Omega_4^{\EPin[k]}\cong\mathrm{coker}(\Omega_5^\Spin(B\Z/(k/2))\to\Omega_5^\Spin(B\Z/k)).
\end{equation}
$\Omega_5^\Spin(B\Z/(k/2))$ and $\Omega_5^\Spin(B\Z/k)$ are known to be generated by lens space bundles
$Q_\ell^5(1, j)$ over $S^2$, with a complete invariant given by a collection of $\eta$-invariants~\cite[\S 5]{BGS97}, and Barrera-Yanez and Gilkey~\cite[Theorem 1.3(2)]{BYG99}, using work of Donnelly~\cite{Don78}, found a formula for the values of these $\eta$-invariants on $Q_\ell^5(1, j)$.\footnote{See~\cite[\S C.2]{DDHM23} for a slight simplification of this formula.} Using this formula,
one can completely understand the map $\Omega_5^\Spin(B\Z/(k/2))\to\Omega_5^\Spin(B\Z/k)$
and therefore compute $\Omega_4^{\EPin[k]}$ using Eq.~\eqref{coker_map}. The Smith long exact sequence is crucial here: Barrera-Yanez~\cite[\S 3]{BY99} discovered that a similar $\eta$-invariant argument without the extra information of the Smith homomorphism is insufficient to resolve the extension question.
\end{rem}

Having determined the bordism groups corresponding to the $\Z/k^T\times \Z/2^f$ symmetry, it is natural to ask what are the maps between different bordism groups induced by maps between these fermionic symmetry groups. An important motivation of this in physics comes from anomaly matching. Consider the renormalization group flow from some UV theory to an IR theory. The UV symmetry $G_{\mathrm{UV}}$ and IR symmetry $G_{\mathrm{IR}}$ are in general different but related by a homomorphism $\phi\colon G_{\mathrm{UV}}\rightarrow G_{\mathrm{IR}}$. Anomaly matching in this context means that the UV anomaly is in fact the pullback of the IR anomaly induced by the homomorphism $\phi$. Therefore, such maps between different bordism groups induced by maps between different symmetry groups are very important in this context. Concrete examples involving such interplay include the so-called ``emergibility'' problem \cite{Ye2021LSM,Zou2021} and the ``intrinsically gapless SPT'' phase \cite{Thorngren2021,Rui2023}. 

We can also read off this information by following the maps between Atiyah-Hirzeburch spectral sequences, similar to the proof of $\Omega_4^{\EPin[k]}\cong \Z/2\oplus\Z/2$ when $8\mid k$. We collect some of these results in the following proposition.

\begin{prop}\label{prop:map}\hfill
\begin{enumerate}
    \item When $k\equiv 4 \bmod 8$, the projection $\Z/k^T\rightarrow\Z/2^T$ induces the map $\Omega^{\EPin[k]}_4\rightarrow \Omega^{\Pin^+}_4$ such that $1\in\Omega_4^{\EPin[k]}\cong\Z/4$ is mapped to $4\in\Omega^{\Pin^+}_4\cong\Z/16$.
    \item When $k\equiv 4\bmod 8$, the projection $\Z/(2k)^T\rightarrow\Z/k^T$ induces the map $\Omega_4^{\EPin[2k]}\rightarrow \Omega_4^{\EPin[k]}$ such that $(1,0), (0,1)\in\Omega_4^{\EPin[2k]}\cong\Z/2\oplus\Z/2$ are mapped to $2,0\in\Omega_4^{\EPin[k]}\cong\Z/4$, respectively.
    \item When $8\mid k$, the projection  $\Z/(2k)^T\rightarrow\Z/k^T$ induces the map $\Omega_4^{\EPin[2k]}\rightarrow \Omega_4^{\EPin[k]}$ such that $(1,0), (0,1)\in\Omega_4^{\EPin[2k]}\cong\Z/2\oplus\Z/2$ are mapped to $(1,0), (0,0)\in\Omega_4^{\EPin[k]}\cong\Z/2\oplus\Z/2$, respectively.
\end{enumerate}
\end{prop}

\subsection{\texorpdfstring{Manifold Generators for $\Omega^{\EPin[k]}_4$}{}}\label{app:manifold_generator}

In this subsection, we define an epin$[k]$ manifold $\mc{M}$  and show in \cref{thm:generator} that for $k \equiv 4\bmod 8$, the bordism class of $\mc{M}$ generates $\Omega_4^{\EPin[k]}$, while for $k\equiv 0\bmod 8$, $\mc{M}$ and the K3 surface generate $\Omega_4^{\EPin[k]}$. We again use the Smith long exact sequence in our proof. A lot of topological information about $\mc{M}$ is given in our proof; notably, we present the Kirby diagram of $\mc{M}$ in \cref{fig:Kirby}, which will be explicitly needed in the derivation of the anomaly indicator for the $\Z/4^T\times \Z/2^f$ symmetry. 

In order to describe the generator $\mc{M}$ of $\Omega_4^{\EPin[k]}$, we need to collect some necessary information about the Klein bottle $K$, which will be an important piece in the construction. According to \cite{Earle}, the Klein bottle can be realized as $\C/\Gamma$, where $\C$ is the complex plane and $\Gamma$ is the group generated by the automorphisms $\mathcal A$, $\mathcal B$ of the plane,
\begin{equation}
\mc{A}z = \bar{z}+1/2, \quad \mc{B}z = z+i\,.
\end{equation}
The fundamental group of the Klein bottle is generated by an orientable loop $a$, shown on the complex plane as a straight line from $z=0$ to $z=i$, and an unorientable loop $b$, shown as a straight line from $z=0$ to $z=1/2$. They satisfy the relation $bab^{-1}=a^{-1}$, i.e.,
\begin{equation}\label{eq:pi1K}
\pi_1(K) \cong \langle a, b\mid bab^{-1}=a^{-1}\rangle
\end{equation}
An element $a^m b^n, m,n\in\Z$ of the group will be denoted by $(m,n)$ in the paper. $\pi_1(K)$ can also be written as $\Z\rtimes\Z$, where the normal $\Z$ summand is generated by $a$, the other $\Z$ summand is generated by $b$ and the semidirect product displays the nontrivial $b$ action on $a$.
Moreover, because the universal cover of $K$ is contractible, $K\simeq B(\Z\rtimes\Z)$. The $\Z/2$-cohomology ring of $K$ is given by 
\begin{equation}
    H^*(K; \Z/2) = \Z/2[A_a, A_b]/(A_a^2+A_bA_a, A_b^2)\,, \quad |A_a| = |A_b| = 1.
\end{equation}
Here $A_a$, $A_b$ can be identified as the cohomology classes of the following cochains:
\begin{equation} 
A_a(a^m b^n) = m\bmod 2,\quad A_b(a^m b^n) = n\bmod 2.
\end{equation}

Consider the universal Klein bottle bundle $K\to E\to B\Diff(K)$, where $\Diff(K)$ is the diffeomorphism group of $K=\mathbb{C}/\Gamma$.
\begin{defn}\label{defn:alpha}
Let $\alpha\colon \R/\Z\to\Diff(K)$ be the map sending $t\in\R$ to the automorphism of $\C$ defined by $z\mapsto z+t$, which descends to a diffeomorphism of $K = \C/\Gamma$ that only depends on the value in $\R/\Z$. We will also think of this as a map out of $\U$ via the exponential isomorphism $\R/\Z\overset\cong\to\U$.
\end{defn}
Since $\R/\Z$ is connected, the image of $\alpha$ is contained in the subgroup $\Diff_0(K)\subset\Diff(K)$ of diffeomorphisms in the connected component of the identity.
\begin{lem}[{Earle-Eels~\cite[\S 11]{Earle}}]
The map $\alpha\colon\R/\Z\to\Diff_0(K)$ is a homotopy equivalence; in particular, $B\alpha\colon B(\R/\Z)\to B\Diff(K)$ is an isomorphism on $\pi_k$ for $k> 1$.\footnote{In fact, $\pi_0(\Diff(K))\ne 0$, so the map of classifying spaces is not an isomorphism on $\pi_1$. This follows from a theorem of Lickorish~\cite{Lic65}.}
\end{lem}
Therefore, we have $\pi_2(B\Diff(K))\cong\Z$, and the generator can be chosen to be the image of the generator of $\pi_2(B(\mathbb{R}/\mathbb{Z}))\cong \Z$ under the induced action of $B\alpha$ on $\pi_2$.

\begin{lem}
 For the universal Klein bottle bundle
 $K\to E\to B\Diff(K)$, the boundary map $\partial_0 \colon \pi_2(B\Diff (K))\rightarrow \pi_1(K)$ in the long exact sequence of homotopy groups maps the generator of $\pi_2(B\Diff(K))\cong \Z$ to $b^2\in\pi_1(K)$. 
\end{lem}

\begin{proof}
Consider $\U$ which wraps around the loop $b$ twice. This defines a map $j\colon \U \rightarrow K$ whose image in $K= \mathbb{C}/\Gamma$ is the straight line from $z=0$ to $z=1$. 

Similar to \cref{defn:alpha}, we define $\beta\colon \mathbb{R}/\mathbb{Z}\rightarrow\Diff(\U)$ to be the map sending $t\in\R$ to the automorphism of $\C$ defined by $z\mapsto z+t$, which also descends to a diffeomorphism of $\U$ that only depends on the value in $\R/\Z$. Immediately, we see that $\beta$ gives a homotopy equivalence of $\mathbb{R}/\mathbb{Z}$ with $\Diff_0(\U)\subset \Diff(\U)$ of diffeomorphisms in the connected component of identity. Thus, $j$ induces a map $\tilde{j}\colon B\Diff(\U)\rightarrow B\Diff(K)$ which is an isomorphism on $\pi_k$ for $k>1$.

Now we have the following map of two fibrations,
\begin{equation}
\begin{tikzcd}
	\U & E\U & B\Diff\U \\
	K & EK & B\Diff(K)
	\arrow[from=1-1, to=1-2]
	\arrow[from=1-2, to=1-3]
	\arrow["j", from=1-1, to=2-1]
	\arrow[from=1-2, to=2-2]
	\arrow["\tilde{j}", from=1-3, to=2-3]
	\arrow[from=2-1, to=2-2]
	\arrow[from=2-2, to=2-3]
\end{tikzcd}
\end{equation}
This gives a map of two long exact sequences of homotopy groups,
\begin{equation}
\begin{tikzcd}
	\pi_2(B\Diff(\U)) & \pi_1(\U) \\
	\pi_2(B\Diff(K)) & \pi_1(K)
	\arrow["\cong", from=1-1, to=1-2]
	\arrow["j_*", from=1-2, to=2-2]
	\arrow["\cong", from=1-1, to=2-1]
	\arrow["\partial_0", from=2-1, to=2-2]
\end{tikzcd}
\end{equation}
Since $j_*$ maps the generator of $\pi_1(\U)\cong\Z$ to $b^2\in \pi_1(K)$, $\partial_0 \colon \pi_2(B\Diff (K))\rightarrow \pi_1(K)$ also maps the generator to $b^2\in \pi_1(K)$.
\end{proof}

Now we can write down a generator $\mc{M}$ of $\Omega_4^{\EPin[k]}$. The generator we choose is a Klein bottle bundle over $S^2$. 
\begin{thm}\label{thm:generator}
    The following data defines a closed $4$-manifold with epin$[k]$ structure.
    \begin{enumerate}
        \item The manifold $\mc{M}$ itself is a nontrivial Klein bottle bundle over $S^2$ such that the classifying map of the Klein bottle bundle $f_{\mc{M}}\colon S^2\rightarrow B\Diff(K)$ is $k/2$ times the generator of $\pi_2(B\Diff(K))\cong \Z$ described above. 

        \item The principal $\Z/k$-bundle on $\mc{M}$ is defined by a map $f\colon \mc{M}\rightarrow B\Z/k$, which is determined by the induced map $\pi_1(\mc{M})\rightarrow \pi_1(B\Z/k)\cong\Z/k$, such that $a\rightarrow 0$ and $b\rightarrow 1$.

        \item The spin structure on $T\mc{M}\oplus f^*(\sigma)$ is chosen such that the orientable cycle $a$ has a non-bounding spin-structure.\footnote{Of the $|H_1(\mc{M}; \Z/2)|=4$ spin-structures, two satisfy this condition, and they correspond to a generator and the inverse of the generator of $\Omega_4^{\EPin[k]}$, respectively.}       
    \end{enumerate}
    If $k\equiv 4\bmod 8$, the bordism class of $\mathcal M$ generates $\Omega_4^{\EPin[k]}\cong\Z/4$. If $k\equiv 0\bmod 8$, the bordism classes of $\mathcal M$ and the K3 surface generate $\Omega_4^{\EPin[k]}\cong\Z/2\oplus\Z/2$.
\end{thm}

For $k\equiv 4\bmod 8$, $\Omega_4^{\EPin[k]}\cong\Z/4$, so this lemma detects the unique generator outright. For
$k\equiv 0\bmod 8$, $\Omega_4^{\EPin[k]}\cong\Z/2\oplus\Z/2$, so this lemma provides one of two generators, while the
other generator may be taken to be the K3 surface with trivial $\Z/k$-bundle.

\begin{rem}
\label{smith_and_generators}
To find manifold representatives of the generators of $\Omega_4^{\EPin[k]}$, in principle one could coax the generators out of the Smith homomorphism $\Omega_5^\Spin(B\Z/k)\overset\cong\to
\Omega_4^{\EPin[k]}$ we studied in \S\ref{subsec:smith_cal}. Assume $k = 4$ for now, so that $\Omega_4^{\EPin[4]}\cong\Z/4$ by \cref{epink_thm_middle}. 
Let $M$ be a closed spin $5$-manifold with principal $\Z/4$-bundle $P\to M$ such that $(M, P)$ is a generator for $\Omega_5^\Spin(B\Z/4)\cong \Z/4$. Let $f\colon M\to B\Z/4$ be the classifying map for $P$. Because the Smith homomorphism $\Omega_5^{\Spin}(B\Z/4)\to\Omega_4^{\EPin[4]}$ is an isomorphism, if $N\subset M$ is a smooth representative of the Poincaré dual of $f^*(x)$,
then $(N, P|_N)$ generates $\Omega_4^{\EPin[4]}$.\footnote{According to \cref{smith_caveat}, we should in principle be careful about using $f^*(x)$ versus the twisted spin cobordism Euler class of $f^*(\sigma)$, but by using $f^*(x)$ we simplify the calculations and do not lose any information in this example.}
$\Omega_5^\Spin(B\Z/4)$ is well-understood: one choice of the
generating manifold, called $Q_4^5$, is a fiber bundle over $S^2$ with fiber the lens space $L_4^3 =
S^3/(\Z/4)$~\cite[\S 5]{BGS97}.  If $P\to L_4^3$ denotes the quotient $\Z/4$-bundle $S^3\to S^3/(\Z/4) = L_4^3$, with classifying map $g\colon L_4^3\to B\Z/4$,
the Poincaré dual of $g^*(x)\in H^1(L_4^3;\Z/2)$ can be represented by a Klein bottle
\cite{Bredon1969NonorientableSI}, suggesting that the Poincaré dual of $f^*(x)\in H^1(Q_4^5;\Z/2)$ can be represented
by a Klein bottle bundle over $S^2$. However, rather than construct the generator of $\Omega_4^{\EPin[4]}$ as a submanifold of $Q_4^5$, in our presentation we only used the Smith homomorphism as inspiration that we should look for a Klein bottle bundle over $S^2$. Then the task is to write one down with the right properties such that it indeed generates $\Omega_4^{\EPin[4]}$.

For other values of $k$ not much changes; when $k\equiv 4\bmod 8$ the story is almost exactly the same, except that the Smith homomorphism is merely surjective, not an isomorphism, before localizing at $2$. For $k\equiv 0\bmod 8$, $\Omega_4^{\EPin[k]}\cong\Z/2\oplus\Z/2$; the above argument works for a generator of one of the two $\Z/2$ summands, and the second summand is generated by the K3 surface, as follows from the generator of $\Omega_4^\Spin((B\Z)^{\sigma-1})$ identified below, together with Eq.~\eqref{plugged_in_LES}.
\end{rem}

We start the proof of \cref{thm:generator} by proving the following detection lemma for the generator of $\Omega_4^{\EPin[k]}$, which is built from the Smith homomorphism as well; then we check that $\mc{M}$ indeed satisfies the requirement in the detection lemma.
\begin{lem}
\label{gen_lem}
Suppose $\mc{M}$ is a closed $4$-dimensional epin$[k]$ manifold, with the associated principal $\Z/k$-bundle $\mc{P}\to\mc{M}$ classified by a map $f\colon \mc{M}\rightarrow B\Z/k$. Let $i\colon\mc{N}\hookrightarrow \mc{M}$ be a smooth representative of the
Poincaré dual of $f^*(y)\in H^2(\mc{M};\Z/2)$, and $\mc{S}\subset\mc{N}$ be the Poincaré dual of $(i\circ f)^*(x)\in
H^1(\mc{N};\Z/2)$. Then $\mc{S}$ has a spin-$\Z/k$ structure induced from the epin$[k]$ structure on $\mc{M}$, and if $\mc{S}\cong
\Snb^1$ as a spin-$\Z/k$ manifold, then $[\mc{M}]\in\Omega_4^{\EPin[k]}$ is neither zero nor the class of the K3 surface.
\end{lem}
Here, a spin-$\Z/k$ structure is a $\Spin\times_{\set{\pm 1}}\Z/(2k)$ structure where the diagonal $\set{\pm 1}$ subgroup is quotiented out.

\begin{proof}
Let $\rho\colon\Z/k\to\U$ be the standard one-(complex-)-dimensional rotation representation and $V_\rho\to
B\Z/k$ be the associated complex line bundle arising as the pullback of the tautological line bundle on $B\U\cong\mathbb{CP}^\infty$.

Recall the long exact sequence from \cref{Smith_LES}, built around the Smith homomorphism. We need two instances of
this long exact sequence, both with $\xi = \Spin$ and $X = B\Z/k$:
\begin{enumerate}
	\item $V = \sigma$ and $W = V_\rho$, and
	\item $V = V_\rho\oplus\sigma$ and $W = \sigma$.
\end{enumerate}
\cref{which_sphere_bundle} identifies $S(\sigma)\simeq B\Z/(k/2)$. Moreover, by using the same procedure as in \cref{which_sphere_bundle}, since the homotopy fiber of
the map $B\Z/k \rightarrow B\U$ is $\U/\Z/k = \U$ as a topological space,
we find that $S(V_\rho) = B\Z$. Then, we have long exact sequences
\begin{subequations}
\begin{gather}
	\label{first_generator_LES}
	\dotsb\to \Omega_k^\Spin((B\Z)^{\sigma-1}) \to \Omega_k^{\EPin[k]}
	\overset{S_{V_\rho}}{\to} \Omega_{k-2}^\Spin((B\Z/k)^{V_\rho\oplus\sigma-3}) \to
	\Omega_{k-1}^\Spin((B\Z)^{\sigma-1}) \to \dotsb\\
	\dotsb\to \Omega_k^\Spin((B\Z/(k/2))^{V_\rho-2}) \to \Omega_k^\Spin((B\Z/k)^{V_\rho\oplus\sigma-3})
	\overset{S_\sigma}{\longrightarrow} \Omega_{k-1}^\Spin((B\Z/k)^{V_\rho-2}) \to \dotsb
	\label{second_generator_LES}
\end{gather}
\end{subequations}

Many of these bordism groups are known. 
\begin{itemize}
	
 \item Campbell~\cite[\S 7.8, \S 7.9]{Cam17} identifies $\Omega_*^\Spin((B\Z/2^\ell)^{V_\rho-2})$ with
	{spin\text{-}$\Z/2^{\ell+1}$-bordism}. The case of $\ell=1$ was computed in \cite{Gia76}, and \cite{Cam17} calculates the other bordism groups in degrees $4$ and below. Specifically, we need $\Omega_2^{\Spin\text{-}\Z/k} = 0$ and $\Omega_1^{\Spin\text{-}\Z/(2k)}\cong\Z/(2k)$~\cite[Theorems 7.9 and 7.10]{Cam17}.
	
 \item Botvinnik-Gilkey~\cite{BG97} study the Adams spectral sequence for
	$\ko_*((B\Z/k)^{V_\rho\oplus\sigma-3})$, and Barrera-Yanez~\cite[Theorem 3.1]{BY99} resolves some extension questions. In dimensions $7$ and below this spectral sequence is isomorphic to the Adams spectral
	sequence for $\Omega_*^\Spin((B\Z/k)^{V_\rho\oplus\sigma-3})$. Though Botvinnik-Gilkey do not explicitly identify their
	twisted $\ko$-homology groups, from their computations~\cite[\S 5]{BG97} it follows that
	$\Omega_m^\Spin((B\Z/k)^{V_\rho\oplus\sigma-3})$ vanishes for $m = 1,3$ and is isomorphic to $\Z/2$ for $m =
	2$.
    \item In~\cite[Footnote 29]{DDKLPTT2} it is shown that $\Omega_4^\Spin((B\Z)^{\sigma-1})\cong\Z/2$, generated by the K3 surface with trivial map to $\Z$. Thus, the map $\Omega_4^\Spin\to\Omega_4^\Spin((B\Z)^{\sigma-1})$ is surjective.
\end{itemize}

Put these computations into~\eqref{first_generator_LES}, we obtain a long exact sequence
\begin{equation}
\label{plugged_in_LES}
	\underbracket{\Omega_3^\Spin((B\Z/k)^{V_\rho\oplus\sigma-3})}_{=0} \overset{\iota}{\longrightarrow}
	\underbracket{\Omega_4^\Spin((B\Z)^{\sigma-1})}_{\cong\Z/2} \longrightarrow
	\Omega_4^{\EPin[k]} \overset{S_{V_\rho}}{\longrightarrow}
	\underbracket{\Omega_2^\Spin((B\Z/k)^{V_\rho\oplus\sigma-3})}_{\cong\Z/2}.
\end{equation}
First we see that $\iota$ is injective; since we know the K3 surface with trivial map to $B\Z$ generates $\Omega_4^\Spin((B\Z)^{\sigma-1})\cong\Z/2$,
a $\Z/2$ subgroup of $\Omega_4^{\EPin[k]}$ is
generated by the K3 surface with trivial $\Z/k$-bundle. The complement of this subgroup maps under $S_{V_\rho}$ to
a nonzero element of $\Omega_2^\Spin((B\Z/k)^{V_\rho\oplus\sigma-3})$, and since $\Omega_4^{\EPin[k]}$ has four
elements from \cref{order_4}, $S_{V_\rho}$ is surjective. Thus, if $\mc{M}$ is an epin$[k]$ $4$-manifold whose
bordism class is distinct from $0$ and $[\mathrm{K3}]$,
%
and $i\colon\mc{N}\hookrightarrow \mc{M}$ is a smooth representative of the Poincaré dual of $f^*(y)\in H^2(\mc{M};\Z/2)$, then $[\mc{N}] = S_{V_\rho}([\mc{M}])$ in $\Omega_2^\Spin((B\Z/k)^{V_\rho\oplus\sigma-3})$.\footnote{Again, because of \cref{smith_caveat} one should be careful of which cohomology theory one takes duals in: in principle by using mod $2$ cohomology instead of spin cobordism, one could lose information. But in the course of our calculation, we see that we learn enough to detect $\mathcal M$, so this simplification is OK. The same is true for $\mathcal S$ below.}

Now~\eqref{second_generator_LES} gives:
\begin{equation}\label{eq:smithsigma}
	\underbracket{\Omega_2^{\Spin\text{-}\Z/k}}_{=0} \longrightarrow
	\underbracket{\Omega_2^\Spin((B\Z/k)^{V_\rho\oplus\sigma-3})}_{\cong\Z/2} \overset{S_\sigma}{\longrightarrow}
	\underbracket{\Omega_1^{\Spin\text{-}\Z/(2k)}}_{\cong\Z/(2k)},
\end{equation}
which identifies the map $S_\sigma\colon
\Omega_2^\Spin((B\Z/k)^{V_\rho\oplus\sigma-3})\to\Omega_1^{\Spin\text{-}\Z/(2k)}$ with the map
$\Z/2\to\Z/(2k)$
sending $1\mapsto k$. Thus with $\mc{M}$ and $\mc{N}$ above, if $\mc{S}\subset \mc{N}$ is a smooth representative of the
Poincaré dual of $(i\circ f)^*(x)\in H^1(\mc{N};\Z/2)$, then $\mc{S}$ has a spin-$\Z/(2k)$ structure and $[\mc{S}]$ corresponds to $k\in\Omega_1^{\Spin\text{-}\Z/(2k)}\cong\Z/(2k)$. To finish off the theorem, all we need to know is that
$\Snb^1$ with trivial $\Z/(2k)$-bundle also represents $k\in\Omega_1^{\text{Spin-}\Z/(2k)}\cong \Z/(2k)$. For $k =
4$,~\cite[Footnote 52]{DDHM23} explains how to coax this out of the Atiyah-Hirzebruch spectral sequence computing
$\Omega_*^{\text{Spin-}\Z/(2k)}$; the proof for arbitrary $k$ is the same.
\end{proof}

\begin{proof}[Proof of \cref{thm:generator}]
We start the proof by deriving a few topological properties of $\mc{M}$, in particular its $\Z/2$ cohomology ring, and show that $\mc{M}$ is indeed an epin manifold. We then check that $\mc{M}$ satisfies the desired properties of \cref{gen_lem}, hence indeed a generator of $\Omega_4^{\EPin[k]}$ that is not the class of the K3 surface.

Consider the long exact sequence of homotopy groups
\begin{equation}\label{eq:smithsigma_more}
	\underbracket{\pi_2(S^2)}_{\cong\Z} \overset{\partial_{\mc{M}}}{\longrightarrow}
	\underbracket{\pi_1(K)}_{\cong\Z\rtimes\Z} \longrightarrow
	\underbracket{\pi_1(\mc{M})}\longrightarrow \underbracket{\pi_1(S^2)}_{\cong 0}\,.
\end{equation}
The map $\partial_{\mc{M}}$ is $\partial_0\circ \left(f_{\mc{M}}\right)_*$, where $f_{\mc{M}}\colon \mc{M}\rightarrow B\Diff(K)$ is the classifying map of the Klein bottle bundle and $\partial_0 \colon\pi_2(B\Diff(K))\rightarrow \pi_1(K)$ is the boundary map for the universal bundle of $K$. Hence, $\partial_{\mc{M}}$ maps the generator of $\pi_2(S^2)\cong \Z$ to $k/2$ times the generator of $\pi_2(B\Diff(K))\cong \Z$ and further to $(0,k) \in \pi_1(K)\cong \Z\rtimes\Z$. From the exactness of Eq.~\eqref{eq:smithsigma_more} we see that $\pi_1(\mc{M}) \cong \Z\rtimes\Z/k$, i.e.,
\begin{equation}\label{eq:pi1M}
\pi_1(\mc{M}) \cong \langle a, b\mid bab^{-1}=a^{-1}, b^k = 1\rangle
\end{equation}

By taking the classifying space of Eq.~\eqref{eq:smithsigma_more}, we have the following fiber sequence:
\begin{equation}
    \underbracket{B(\Z\rtimes\Z)}_{\cong K} \longrightarrow B(\Z\rtimes\Z/k) \longrightarrow \underbracket{B^2\Z}_{\cong \mathbb{CP}^\infty}\,.
\end{equation}
Then we observe that we have the following commutative diagram
\begin{equation}
\label{EPin_pullback}
\begin{gathered}
\begin{tikzcd}
	{K} & {\mc{M}} & {} &{S^2} \\
	{}&{B(\Z\rtimes\Z/k)} & {B\Z/k} & {B^2\Z\cong \mathbb{CP}^\infty}
	\arrow[from=1-1, to=1-2]
	\arrow[from=1-1, to=2-2]
	\arrow["f", from=1-2, to=2-3]
	\arrow["p", from=1-2, to=1-4]
	\arrow["\rho", from=2-3, to=2-4]
	\arrow["i", from=1-4, to=2-4]
	\arrow[from=2-2, to=2-3]
	\arrow[from=1-2, to=2-2]
\end{tikzcd}
\end{gathered}
\end{equation}
where $p$ is the projection to the base $S^2$, $\rho$ is induced from the standard rotation representation of $\Z/k$, and $i$ is the natural inclusion of $S^2$ into $\mathbb{CP}^\infty$. The cohomology rings of all of these spaces except $\mc{M}$ are known, and we can use them to build the cohomology of $\mc{M}$. Specifically, we have
\begin{itemize}
    \item $H^*(\mathbb{CP}^\infty, \Z/2) \cong \Z/2[c_1]\,,\quad|c_1| = 2$
    \item $H^*(S^2, \Z/2) \cong \Z/2[B]/(B^2)\,,\quad|B| = 2$
    \item $H^*(B\Z/k;\Z/2)\cong\Z/2[x, y]/(x^2)\,,\quad \abs x = 1,~ \abs y = 2\,$
    \item $H^*(K; \Z/2) \cong \Z/2[A_a, A_b]/\{A_a^2+A_bA_a,~ A_b^2\}\,, \quad |A_a| = |A_b| = 1$
    \item $H^*(B(\Z \rtimes \Z/k);\Z/2) \cong \Z/2[\hat{\mathscr{A}}_a,\hat{\mathscr{A}}_b,\hat{\mathscr{B}}]/\{\hat{\mathscr{A}}_a^2+\hat{\mathscr{A}}_a\hat{\mathscr{A}}_b,~ \hat{\mathscr{A}}_b^2\}\,,\quad |\hat{\mathscr{A}}_a| = |\hat{\mathscr{A}}_b| = 1,~|\hat{\mathscr{B}}|=2$
\end{itemize}
In particular, $\hat{\mathscr{A}}_a$, $\hat{\mathscr{A}}_b$ and $\hat{\mathscr{B}}$ in $H^*(B(\Z \rtimes \Z/k);\Z/2)$ can be defined by the following cochains:
\begin{subequations}
\begin{align}
    \hat{\mathscr{A}}_a(a^m b^n) &= m\bmod 2\\
    \hat{\mathscr{A}}_b(a^m b^n) &= n\bmod 2\\
    \hat{\mathscr{B}}(a^{m_1} b^{n_1}, a^{m_2} b^{n_2}) &= \frac{(n_1\bmod k) + (n_2\bmod k) - (n_1 + n_2 \bmod k)}{k}.
\end{align}
\end{subequations}
Then from the Serre spectral sequence with respect to $K \to \mathcal{M} \to S^2$, the $\Z/2$ cohomology ring of $\mathcal{M}$ is given by
\begin{equation}\label{eq:cohM}
    H^*(\mc{M};\Z/2) = \Z/2[\mathscr{A}_a,\mathscr{A}_b,\mathscr{B}]/\{\mathscr{A}_a^2+\mathscr{A}_a\mathscr{A}_b,~ \mathscr{A}_b^2,~ \mathscr{B}^2\}\,.
\end{equation}
In particular, the conditions $\mathscr{A}_a^2=\mathscr{A}_a\mathscr{A}_b$ and $\mathscr{A}_b^2=0$ come from the pullback of conditions $\hat{\mathscr{A}}_a^2=\hat{\mathscr{A}}_a\hat{\mathscr{A}}_b$, $\hat{\mathscr{A}}_b^2=0$ in $H^*(B(\Z \rtimes \Z/k);\Z/2)$, and the condition $\mathscr{B}^2=0$ comes from the pullback of the condition $B^2=0$ in $H^*(S^2; \Z/2)$. Maps between these cohomology groups induced by maps between different spaces can be obtained by inspecting the explicit cochain representatives or the Serre spectral sequence. We have
\begin{equation}\label{eq:pullback}
\begin{gathered}
\begin{tikzcd}
	{\underbrace{\mathscr{B}}_{\in H^2(\mc{M}; \Z/2)}} & {} &{\underbrace{B}_{\in H^2(S^2; \Z/2)}} \\
	{\underbrace{\hat{\mathscr{B}}}_{\in H^2(B(\Z\rtimes\Z/k); \Z/2)}} & {\underbrace{y}_{\in H^2(B\Z/k; \Z/2)}} & {\underbrace{c_1}_{\in H^2(\mathbb{CP}^\infty; \Z/2)}}
	\arrow["f^*", from=2-2, to=1-1]
	\arrow[from=2-1, to=1-1]
	\arrow["\rho^*", from=2-3, to=2-2]
	\arrow["i^*", from=2-3, to=1-3]
	\arrow["p^*", from=1-3, to=1-1]
	\arrow[from=2-2, to=2-1]
\end{tikzcd}
\end{gathered}
\end{equation}

Now we are ready to show that $\mc{M}$ is an epin$[k]$ manifold. According to the requirement in Eq.~\eqref{eq:requirement}, we just need to show that 
\begin{itemize}
\item $w_1(T\mc{M}) = f^*(x)$: Since the induced map $\pi_1(\mc{M})\rightarrow\Z/4$ maps $a\mapsto 0$ and $b\mapsto 1$ where $a$ is orientable and $b$ is unorientable, from the explicit cochain representatives we immediately have $w_1(T\mc{M}) = f^*(x) = \mathscr{A}_b$.
\item $w_2(T\mc{M})=0$: For any $u\in H^2(\mc{M}; \Z/2)$, consider $u^2 = \Sq^2(u) = \nu_2(T\mc{M}) u$ where $\nu_2$ is the second Wu class given by $\nu_2(T\mc{M}) = w_1(T\mc{M})^2+w_2(T\mc{M})$. Because $u$ can be written as a combination of $\mathscr{A}_a^2$ and $\mathscr{B}$, $u^2=0$ and hence $\nu_2(T\mc{M}) = 0$. Therefore, we must have $w_2(T\mc{M})= w_1^2(T\mc{M}) = \mathscr{A}_b^2 = 0$. 
\end{itemize}
We conclude that $\mc{M}$ is indeed an epin$[k]$ manifold. 

Now we can directly check that the constructed $\mc{M}$  satisfies the condition stated in \cref{gen_lem}. In particular, from Eq.~\eqref{eq:pullback} $f^*(y) = \mathscr{B}$, and $\mc{N}$ as a smooth representative of the Poincaré dual of $\mathscr{B}=p^*(B)$ can be chosen to be the Klein bottle $K$ with induced $\Z/k$-bundle structure and spin-structure from $\mc{M}$. Finally, $(i\circ f)^*(x) = A_b\in H^1(K;\Z/2)$ and $\mc{S}$ as a smooth representative of $A_b$ can be chosen to be $a$, which by construction is indeed $\Snb^1$. Therefore, we establish that $\mc{M}$ with the stated $\Z/k$-bundle structure and spin-structure is indeed a nonzero element of $\Omega_4^{\EPin[k]}$ that is not the class of K3 surface.
\end{proof}

\begin{figure}[!htbp]
\includegraphics[width=\textwidth]{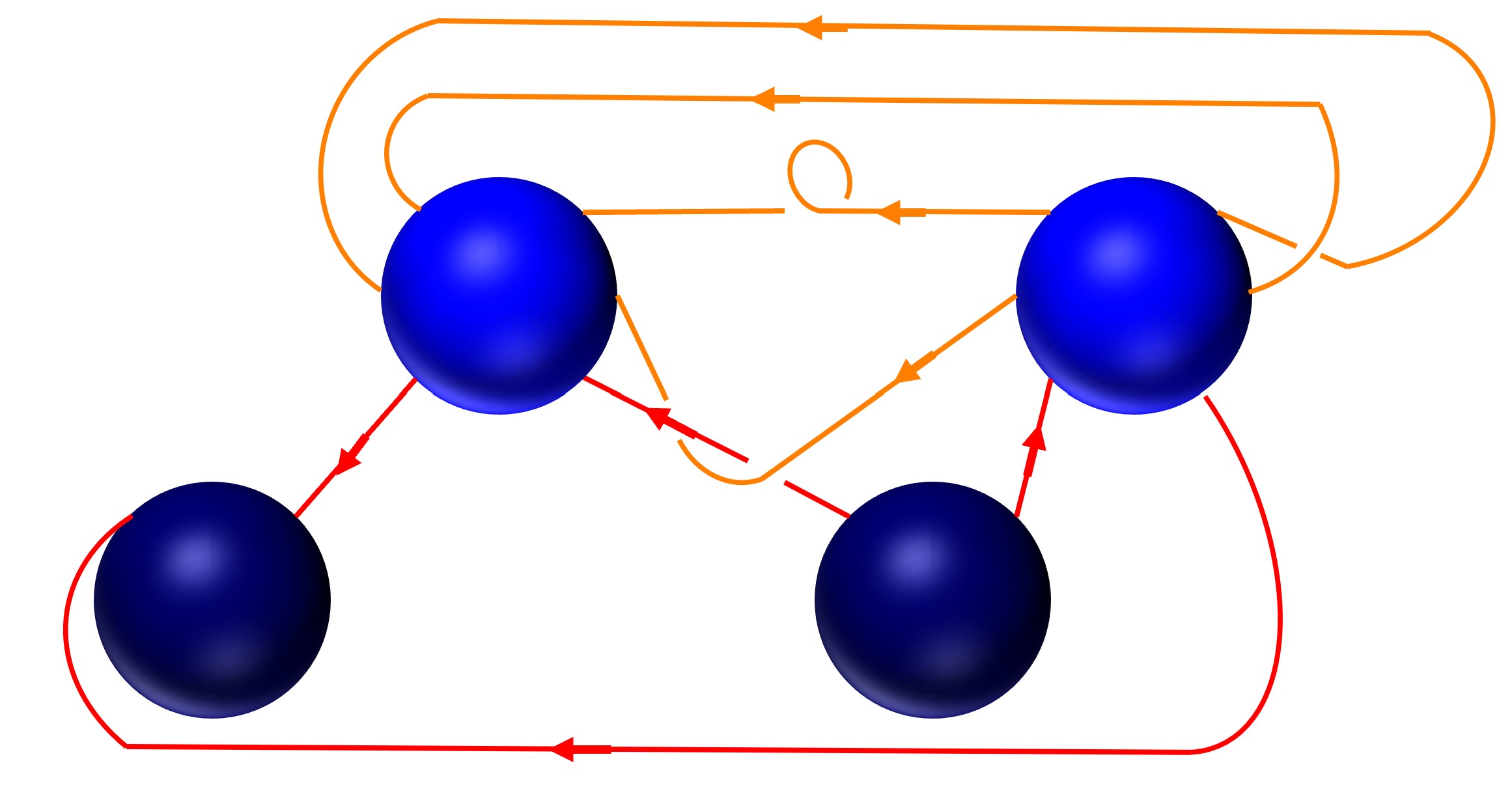}
\caption{The Kirby diagram of the generator $\mc{M}$. The blue balls and dark blue balls illustrate two 1-handles, and the red lines and orange lines illustrate two 2-handles. The 1-handle denoted by blue balls is unorientable while the 1-handle denoted by dark blue balls is orientable.}
\label{fig:Kirby}
\end{figure}

Finally, after obtaining the generating manifold $\mc{M}$, the final topological ingredient that we need is the Kirby diagram of $\mc{M}$ for $k=4$. Recall that the fundamental groups of both the Klein bottle and $\mc{M}$ are generated by an orientable cycle $a$ and an unorientable cycle $b$, such that they satisfy the condition $ba b^{-1} a = 1$. The dark blue and blue 1-handle, shown in \cref{fig:Kirby}, correspond to cycle $a$ and $b$, respectively. The red line which spans the blue and dark blue balls represents the 2-handle coming from the Klein bottle, and can be drawn by following the path $bab^{-1}a$. The orange line which spans only the blue balls represent the 2-handle coming from the base $S^2$. Since the fibration is given by $\pi_2(S^2) \to \pi_1(K)$ that sends the generator 1 to $(0,4)$, the equator of $S^2$ is wrapped along the unorientable loop $b$ of the Klein bottle four times, and this is reflected by the orange line traveling along the blue 1-handle four times. The red line crosses the orange line under and over exactly once, reflecting the $+1$ intersection number of the two 2-handles. We also need the orange 2-handle to have self-intersection zero, hence we add an extra loop on the upper part of the diagram to cancel the self-crossing on the right of the rightmost blue ball. In summary, the minimal handle-decomposition of $\mc{M}$ contains 1 0-handle, 2 1-handles, 2 2-handles, 2 3-handles, and 1 4-handle, and its Kirby diagram is given in \cref{fig:Kirby}. 

Now we are well-equipped with all the necessary topological information for writing down the \textit{specific} anomaly indicator for the $\Z/4^T\times \Z/2^f$ symmetry in \S\ref{sec:anomalycomp}.

\section{Data of Fermionic Topological Orders}\label{app:data}

For the reader's convenience, in this section we explicitly write down the data of the fermionic topological orders considered in this paper, i.e., $\U_5$, $\U_2\times \U_{-1}$ (semion-fermion theory), and $\SO(3)_3$, together with the data of the $\Z/4^{Tf}$ or the $\Z/4^T\times \Z/2^f$ symmetry action. After directly plugging in the data into Eq.~\eqref{eq:main_computation}, we see that the $\Z/4^T\times \Z/2^f$ anomalies of the three fermionic topological orders correspond to $\nu=0,2,3$ in $\mho^4_\EPin\cong\Z/4$, respectively.

\subsection{\texorpdfstring{$\U_5$}{}}\label{subapp:U(1)5}

Anyons in $\U_5$ can be labeled by integers $a=0,\dots,9$. $F$-symbols can be chosen to be all 1 while $R$-symbols can be chosen to be
\begin{equation}
    R^{a,b} = \exp\left(\frac{\pi i}{5}ab\right)\,.
\end{equation}
Here we omit the subscript of the $R$-symbol since the outcome of the fusion rules is unique.

The $\Z/4^T\times \Z/2^f$ symmetry permutes anyons in $\U_5$ by $a\mapsto 3a \pmod{10}$. This symmetry is a ``genuine'' $\Z/4^T$ action in the sense that all nontrivial elements in $\Z/4^T$ permute anyons. All the $U$-symbols and $\eta$-symbols can be chosen to be 1. 

\subsection{\texorpdfstring{$\U_2\times \U_{-1}$}{}} \label{subapp:semionfermion}

        Anyons of $\U_2\times \U_{-1}$, or the semion-fermion theory, can be labeled by $1,s,\tilde{s},\psi$. This theory is a direct product of the free-fermion theory $\{1,\psi\}$ and the semion theory $\U_2$. The fusion rules can be given as follows:
        \begin{align}
            s \times s = \tilde{s} \times \tilde{s} = \psi \times \psi &= 1\,,\\
            s \times \psi &= \tilde{s}\,, 
        \end{align}
        The nontrivial $F$-symbols are $F^{abc}=-1$ when $(a,b,c)$ is any combination of only $s$ and $\tilde{s}$. The $R$-symbols, with the anyons ordered by $(1,s,\psi, \tilde{s})$, are
        \begin{equation}
            R^{ab} = \begin{pmatrix}
            1 & 1 & 1 & 1\\
            1 & i & 1 & i\\
            1 & 1 & -1 & -1\\
            1 & i & -1 & -i
            \end{pmatrix}\,.
        \end{equation}
        
        The theory has the $\Z/4^{Tf}$ time-reversal symmetry, which exchanges $s$ and $\tilde{s}$. Denote the generator of bosonic $\Z/2^T$ group as $\mc{T}$. $U$-symbols can be written as a matrix, with the row and the column denoting $(a,b)\in \{1,s,\psi,\tilde{s}\}$ of $U_{\mc T}(a,b;a\times b)$
        \begin{equation}\label{eq:Usymbol_semionfermion}
            U_{\mc T}(a,b;a\times b) = \begin{pmatrix}
            1 & 1 & 1 & 1\\ 
            1 & 1   & 1 & 1\\
            1 & -1 & 1 & -1\\
            1 &-1 & 1 & -1
            \end{pmatrix}
        \end{equation}
        We can choose $\eta$-symbols such that $\eta_{\psi}({\mc T},{\mc T}) = -1$, $\eta_s({\mc T},{\mc T}) = -i$ and $\eta_{\tilde s}({\mc T},{\mc T}) = i$.
        
        The $\Z/4^T\times \Z/2^f$ symmetry we consider acts on the theory through the natural projection $p:~\Z/4^T\rightarrow\Z/2^T$. The $U$-symbols and $\eta$-symbols of $\Z/4^T$ can be simply obtained from the pullback of $\Z/2^T$.

\subsection{\texorpdfstring{$\SO(3)_3$}{}} \label{subapp:SO(3)3}
        
        $\SO(3)_3$ can be thought of as a subcategory of $\SU(2)_6$. Anyons in $\SU(2)_6$ can be labeled by integers $a=0,\dots,6$, and anyons in $\SO(3)_3$ can be labeled by $1,s,\tilde{s},\psi$, which are identified as $0,2,4,6$ in $\SU(2)$, respectively. With these identifications, the fusion rules can be given as follows:
               \begin{align}
            \psi \times \psi &= 1\,,\\
            \psi \times s &= \tilde{s}\,,\\
            \psi \times \tilde{s} &= s\,,\\
            s \times s = \tilde{s}\times \tilde{s} &= 1+s+\tilde{s}\,,\\
            s \times \tilde{s} &= \psi + s + \tilde{s}\,.
        \end{align}
        
        Let $q=e^{\pi i/4}$. The $R$-symbols of $\SU(2)_6$ are relatively easy to display:
        \begin{equation}
            R^{a, b}_c = (-1)^{(a+b-c)/2}q^{\frac{1}{8}\left(c(c+2)-a(a+2)-b(b+2)\right)}
        \end{equation}      
        
        The $F$-symbols require a set of auxiliary functions
        \begin{align}
            \lfloor n \rfloor &= \sum_{m=1}^n q^{(n+1)/2-m}\\
            \lfloor n \rfloor! &= \lfloor n \rfloor \lfloor n-1 \rfloor \cdots \lfloor 1 \rfloor\\
            \Delta(a,b,c) &= \sqrt{\frac{\lfloor (a+b-c)/2\rfloor!\lfloor (a-b+c)/2\rfloor!\lfloor (-a+b+c)/2\rfloor!}{\lfloor (a+b+c+2)/2\rfloor!}}
        \end{align}
        for $n\geq 1$, and with $\Delta$ defined only when $a,b,c$ satisfy the triangle inequality. We also define $\lfloor 0\rfloor!=1$. With these definitions, we can define the $F$-symbols by the following formula:
        \begin{align}
            &F^{abc}_{def} = (-1)^{(a+b+c+d)/2}\Delta(a,b,e)\Delta(c,d,e)\Delta(b,c,f)\Delta(a,d,f) \sqrt{\lfloor e+1 \rfloor \lfloor f+1 \rfloor} \times\nonumber \\
            &\times \sum_{n}' \frac{(-1)^{n/2}\lfloor (n+2)/2\rfloor!}{\lfloor(a+b+c+d-n)/2\rfloor!\lfloor(a+c+e+f-n)/2\rfloor!\lfloor b+d+e+f-n\rfloor!} \times \nonumber \\
            &\times \frac{1}{\lfloor (n-a-b-e)/2\rfloor ! \lfloor (n-c-d-e)/2\rfloor ! \lfloor (n-b-c-f)/2\rfloor ! \lfloor (n-a-d-f)/2\rfloor !}
        \end{align}
        where the summation runs over even integers such that $\max(a + b + e, c + d + e, b + c + f, a + d + f) \leq n \leq \min(a + b + c + d, a + c + e + f, b + d + e + f)$. The quantum dimensions are given by $d_1=d_{\psi}=1$ and $d_s=d_{\tilde{s}}=1+\sqrt{2}$, with total quantum dimension $\mathcal{D}^2 = 8+4\sqrt{2}$. The topological spins are $\theta_1 = 1, \theta_{\psi} = -1, \theta_{s}=i, \theta_{\tilde{s}}=-i$.

        The theory has the $\Z/4^{Tf}$ time-reversal symmetry, which exchanges $s$ and $\tilde{s}$. Denote the generator of bosonic $\Z/2^T$ group as $\mc{T}$. The non-trivial $U$ symbols are
        \begin{align}
            U_{\mc{T}}(s,\tilde{s};\psi)=U_{\mc{T}}(\tilde{s},\psi;s) = U_{\mc{T}}(\psi,s;\tilde{s})=U_{\mc{T}}(s,s;s)=U_{\mc{T}}(\tilde{s},\tilde{s};\tilde{s})&=i\\
            U_{\mc{T}}(s,\psi;\tilde{s})=U_{\mc{T}}(\psi,\tilde{s};s)=U_{\mc{T}}(\tilde{s},s;\psi)&=-i
        \end{align}
        and $U_{\mc{T}}(a,b;c)=-i$ when two of $(a,b,c)$ are $s$ and the third is $\tilde{s}$ or vice-versa. Finally, the $\eta$ symbols are all trivial except
        \begin{equation}
            \eta_{\psi}({\mc T},{\mc T})=-1.
        \end{equation}
        One can check exhaustively by a computer that these data satisfy all of the consistency conditions.

        The $\Z/4^T\times \Z/2^f$ symmetry we consider acts on the theory through the natural projection $p:~\Z/4^T\rightarrow\Z/2^T$. The $U$-symbols and $\eta$-symbols of $\Z/4^T$ can be simply obtained from the pullback of $\Z/2^T$.

\section{Anomaly Cascade for Fermionic Topological Orders }\label{app:cascade}

In this section, we give another argument that the anomaly vanishes for the fermionic topological order $\U_5$ with given $\Z/4^T\times \Z/2^f$ symmetry action, and also show that the two other fermionic topological orders $\U_2 \times \U_{-1}$ and $\SO(3)_3$ appearing in the paper have nontrivial anomaly. We use a conjecture of Bulmash-Barkeshli \cite{Bulmash:2021ryq}: the conjecture proceeds by unpacking the anomaly in terms of its layers in the Atiyah-Hirzebruch spectral sequence, and each layer then has an interpretation as an obstruction to extending some data when gauging the given symmetry. 
Following the conjecture, we examine whether the obstruction in each layer is trivial or not for the fermionic topological orders considered in this paper, hence obtaining the anomalies of them. Unpacking the anomaly in this way has the benefit that one can explicitly see the implementation of gauging and the obstruction of it at the level of the skeletalization data of super MTCs.
A technique coming from tensor categories called zesting also finds a nice application here and this is an opportunity to showcase it. 

Because Bulmash-Barkeshli's interpretation of the anomaly for a fermionic symmetry acting on a super-MTC is expressed in terms of the layers of the Atiyah-Hirzebruch filtration, we will quickly walk through the explicit data that appears in this filtration on $\mho_\Spin^4((BG_b)^{V-r_V})$. This filtration is induced from the Postnikov filtration on $I_{\U}\MTSpin$~\cite[Theorem 3.3]{Mau63} (see also~\cite[\S 9]{Ant24} for a modern account). Specifically, from the construction of the Atiyah-Hirzebruch spectral sequence, we learn the following.
\begin{lem}\label{lem:Einftylayers}
Given a fermionic symmetry with data $(G_b,s,\omega)$, let $V\to BG_b$ be a vector bundle on $BG_b$ of rank $r$ such that $w_1(V)=s$ and $w_2(V)=\omega$. The data of the filtration in total degree $4$ of the $E_\infty$-page of the Atiyah-Hirzebruch spectral sequence for $\mho_\Spin^*((BG_b)^{V-r_V})$ implies that there are abelian groups $F^1$ and $F^2$ and short exact sequences
\begin{subequations}
\begin{gather}
    \shortexact{E_\infty^{2,2}}{F^1}{E_\infty^{0,4}}{}\\
    \shortexact{E_\infty^{3,1}}{F^2}{F^1}{}\\
    \shortexact{E_\infty^{4,0}}{\mho_\Spin^4(X^{V - r})}{F^2}.
\end{gather}
\end{subequations}
In addition: 
\begin{itemize}
    \item $E_\infty^{0,4}$ is a subgroup of $H^0(X;\U_{s})$.
    \item $E_\infty^{2,2}$ is a subgroup of $H^2(X;\Z/2)$.
    \item $E_\infty^{3,1}$ is a subquotient of $H^3(X;\Z/2)$.
    \item $E_\infty^{4,0}$ is a quotient of $H^4(X;\U_{s})$.
\end{itemize}
\end{lem}
 The subgroups and quotient groups follow from analyzing the differentials in this spectral sequence and their formulas in \cref{diffs_forms}.

According to \cite{Bulmash:2021ryq}, the four layers in the Atiyah-Hirzebruch spectral sequence capture the anomaly for a super-MTC as four layers of obstructions to
gauging the fermionic symmetry, which they dub the anomaly cascade. For this, we need the concept of minimal modular extension of the super-MTC $\mc{C}$, which is some unitary-MTC that contains the original super-MTC as a sub-category \cite{muger2002structure,Nik:2020}. \cite{bruillard2017fermionic,Galindo2017} state that given a super-MTC $\mc{C}$ if there is one minimal extension, then there are exactly 16 up to Witt equivalence, and \cite{johnson2021minimal} prove that minimal extension always exists. This minimal modular extension is interpreted as the bosonic theory obtained from gauging fermion parity in the fermionic theory. 

\begin{conj}[The anomaly cascade, Bulmash-Barkeshli~\cite{Bulmash:2021ryq}]
\label{BBconj}
The anomaly for a super-MTC constitutes four layers, which have the following interpretation in terms of extending certain data from the super-MTC $\mc{C}$ to some unitary-MTC $\mc{B}$ as the minimal modular extension of the super-MTC: 
\begin{itemize}
\item The first layer: This is valued in $E^{0,4}_\infty$, and is the obstruction for the modular extension to be able to have time-reversal symmetry.
\item The second layer: This is valued in $E^{2,2}_\infty$, and is the obstruction of extending the data of the homomorphism $\rho: G_b \to \mathrm{Aut}(\mathcal{C})$ to a homomorphism $\check{\rho}:G_b\to \mathrm{Aut}(\mathcal{B})$.
\item The third layer: This is valued in $E^{3,1}_\infty$, and is the obstruction of extending the data of symmetry fractionalization, or the data of $\eta_a({\bf g}, {\bf h})$ as defined in \cref{def:symfrac}.
\item The fourth layer: This is valued in $E^{4,0}_\infty$, and is the anomaly of the extended $G_b$ action on the modular extension $\mc{B}$.
\end{itemize}
%
\end{conj}
\begin{rem}
Our presentation here is slightly different from the presentation in \cite{Bulmash:2021ryq} in terms of the value in each layer. We follow closely the data of the Atiyah-Hirzebruch spectral sequence on the infinity page, as in \cref{lem:Einftylayers}. In particular, we demand that the first layer to be $E^{0,4}_\infty$, which is a quotient of $H^0(X;\U_s)$ rather than $H^1(X;\Z_s)$ as in Bulmash-Barkeshli. When $s$ is nontrivial and there are anti-unitary symmetries present, $H^0(X;\U_s)$ is canonically isomorphic to $H^1(X;\Z_s)$: both are $\Z/2$-valued and are connected to each other by the Bockstein induced by $\Z\to\R\to\U$. When $s$ is trivial and there is no anti-unitary symmetry, comparing with \cite{Bulmash:2021ryq}, we simply say that the obstruction in the first layer always vanishes.
\end{rem}

In the rest of this appendix, we assume \cref{BBconj}.

Going along the lines of \cite{Bulmash:2021ryq} we can regard the information of the first three layers with the fourth as giving a mixed anomaly between fermion parity and $G_b$. Since each subsequent layer carries more refined data about the interplay between the symmetry and the MTC, the first three layers must be trivialized in sequential order. If the first three layers are completely trivialized, then the only part of the anomaly is in the bosonic sector controlled by the fourth layer.

\cref{BBconj} and the computation in Eq.~\eqref{eq:AHSScomp} together imply:
\begin{cor}\label{cor:fromBBconj}
   The anomaly of a super-MTC with $\Z/4^T\times \Z/2^f$ symmetry has nontrivial contributions from the first layer 
   in $E^{0,4}_\infty\cong \Z/2$, and the third layer in $E^{3,1}_\infty\cong \Z/2$. 
\end{cor}

As an added bonus for this appendix, we will review and apply a technical trick called zesting explained in \cite{bruillard2017fermionic}, with further applications given in \cite{delaney2021braided}. Zesting is a procedure to obtain one modular extension from another modular extension of a super-MTC. 
The zesting procedure may preserve many features of a fusion category such as its rank, Frobenius-Perron dimension, and grading, but can also alter certain data such as the central charge. Our strategy will be to use zesting to find a modular extension with the property that the central charge is 0 so that the unitary-MTC indeed has time-reversal symmetry \cite{Wang:2016qkb,Bulmash:2021ryq}; this will trivialize the first layer in \cref{cor:fromBBconj}. One particularly useful fact about zesting that we use is that zesting an abelian MTC gives another abelian MTC. 

Here we list some data of the new unitary-MTC obtained from zesting. Beginning from the fusion rules $\otimes$ of the old unitary-MTC, the fusion rules $\boxtimes$ of the new unitary-MTC are given by \cite[Section 4.1]{bruillard2017fermionic}:
\begin{equation}\label{eq:zesting_fusion}
    a_1 \boxtimes a_2 = 
    \left\{ 
    \begin{array}{lr}
    (a_1 \otimes \psi) \otimes a_2  &\text{if both $a_1$ and $a_2$ have odd grading}\,,\\
    a_2 \otimes a_2 & \text{if at least one of $a_1$ or $a_2$ has even grading}\,.\\
    \end{array}
    \right.
\end{equation}
Here, $a_1$ and $a_2$ are anyons in the original unzested theory, $\psi$ is the fermion, and an anyon $a$ in the old unitary-MTC has even (odd) grading if its braiding with $\psi$ is trivial (nontrivial). 
Zesting also gives a new set of braidings given by $R^{a_1,a_2}_\boxtimes$~\cite[Section 4.6]{bruillard2017fermionic}:
\begin{equation}\label{eq:zesting_R}
 R_{\boxtimes}^{a_1,a_2}= 
  \left\{
   \begin{array}{lr}
    b(R^{a_1,\psi}\otimes \mathrm{Id}_{a_2}) \circ R^{\psi\otimes a_1, a_2} \circ \left(F^{a_2,\psi,a_1}\right)^{-1}  &\text{if both $a_1$ and $a_2$ have odd grading}\,,\\
     R^{a_1,a_2}  & \text{otherwise}\,. 
    \end{array}
    \right.
\end{equation}
Here $R$ is the braiding and $F$ is the $F$-symbol in the old unitary-MTC. There is a constant $b$ which one has the freedom to choose so that the resulting theory has certain properties, e.g. vanishing central charge. 

\begin{prop}
\label{U15_BB}
 The anomaly of $\U_5$ for the $\Z/4^T\times \Z/2^f$ symmetry vanishes.
\end{prop}
\begin{proof}
By \cref{cor:fromBBconj}, we need to show that the obstruction in the first layer, valued in $E^{0,4}_{\infty}$, and the obstruction in the third layer, valued in $E^{3,1}_\infty$, each vanish. 

\begin{itemize}
\item $\U_5$ can be extended to an abelian bosonic unitary-MTC with $\Z/4^T$ action. 

The most natural candidate of a unitary-MTC that is a modular extension of $\U_5$ is $\U_{20}$ \cite[Section 2.6]{bruillard2017fermionic}. Anyons in $\U_{20}$ can be labeled by integers $a=0,\dots,19$, and the original anyons $a=0,\dots,9$ of $\U_5$ embed into $\U_{20}$ as $a\mapsto 2a$. However, $\U_{20}$ has central charge 1 instead of 0, hence it has no time-reversal symmetry. Fortunately, the central charge is an integer, and hence we can find another modular extension of $\U_5$ that is related to $\U_{20}$ by zesting such that this new modular extension has the desired time-reversal symmetry.

From Eq.~\eqref{eq:zesting_fusion} and Eq.~\eqref{eq:zesting_R}, we can immediately write down the desired unitary-MTC $\mc{B}$, which is a modular extension of $\U_5$ with zero central charge. The fusion rules of the objects in $\mc{B}$ have group structure $\Z/2 \times \Z/10$, with anyons labeled by $(a,b)$ with $a=0,1$ and $b=0,\dots,9$. The theory $\mc{B}$ describes as a $\Z/2\times \Z/10$ gauge theory with an extra Dijkgraaf-Witten twist.
$\mc{B}$ has trivial $F$-symbols and the $R$-symbols are given by 
\begin{equation}\label{eq:Rsymbolextended}
     R^{(a_1, b_1), (a_2, b_2)}_{(a_1+a_2,b_1+b_2)} = \exp\left(\pi i(a_1a_2+a_1b_2+\tfrac{1}{5} b_1b_2)\right)\,.
\end{equation}
The original anyon labeled by $a=1$ embeds into $\mc{B}$ as $(1,1)$. $\mathcal{B}$ has $\Z/4^T$ time-reversal symmetry generated by the following action 
\begin{equation}\label{eq:Z4T on C}
    (a, b) \longmapsto (a, \,\,3\times b \!\! \pmod{10})\,,
\end{equation}
which is compatible with the original $\Z/4^T$ time-reversal action on $\U_5$.

\item The symmetry fractionalization data of the $\Z/4^T$ symmetry in $\U_5$ can be extended to the new unitary-MTC $\mc{B}$.

All the $U$-symbols and $\eta$-symbols of $\U_5$ can be set equal to 1. It is straightforward to check that the $U$-symbols and $\eta$-symbols of $\mc{B}$ can be set equal to 1 as well.   

\end{itemize}

Hence we conclude that all obstructions vanish and the full anomaly of $\U_5$ indeed vanishes.
\end{proof}

\begin{lem}
\label{SO33_BB}
There is no modular extension of $\SO(3)_3$ with time-reversal symmetry.
\end{lem}
\begin{proof}
A modular extension of $\SO(3)_3$ is $\SU(2)_6$. 
This theory has central charge $\tfrac{9}{4}$, an odd multiple of $\tfrac{1}{4}$. From the 16 fold way classification \cite{bruillard2017fermionic}, all modular extensions of $\SO(3)_3$ must have central charge an odd multiple of $\tfrac{1}{4}$. Therefore, $\SO(3)_3$ has no modular extension that has time-reversal symmetry, and the obstruction in $E^{0,4}_{\infty}$ is thus nontrivial. 
\end{proof}

\begin{rem}
The anomaly cascade cannot determine whether $\SO(3)_3$, for a particular $\Z/4^T$ action (a particular set of choice of $U$- and $\eta$-symbols), has anomaly $1$ or $3$ in $\Omega_4^\EPin\cong\Z/4$ in a straightforward manner, because both $1$ and $3$ lead to the same obstruction at the same level of the anomaly cascade. It is easiest to determine the explicit value by the anomaly indicator for the $\Z/4^T\times \Z/2^f$ symmetry in \cref{prop:anomaly_indicator}. 
\end{rem}

The semion-fermion theory $\U_2\times \U_{-1}$ has a simple modular extension $\U_2\times \U_{-4}$, which also has $\Z/4^T$ symmetry. The anyons in $\U_2\times \U_{-4}$ are labeled by $(a,b),a=0,1,b=0,\dots,4$, and the anyons $1,s,\tilde s, \psi$ in the original semion-fermion theory correspond to $(0,0)$, $(1,0)$, $(1,2)$, $(0,2)$, respectively. The $\Z/4^T\times \Z/2^f$ symmetry is generated by
\begin{equation}
(a,b)\longmapsto(a+b \!\!\mod 2,\, 2a+b \!\mod 4)\,.
\end{equation}

\begin{prop}\label{semion_fermion_BB}
The semion-fermion theory realizes the anomaly $\nu=2 \in \mho^4_{\EPin}$. 
\end{prop}
\begin{proof}[Proof sketch]
The symmetry fractionalization data of the $\Z/4^T\times \Z/2^f$ action in $\U_2\times \U_{-4}$ cannot be made compatible with the $\eta$-symbols of the original semion-fermion theory. One can check that, in order to be compatible with the $\eta$-symbols of the original semion-fermion theory, Eq.~\eqref{eq:etaConsistency} cannot be satisfied: the quotient between the left and right hand side is not 1, but results in the double-braid between two anyons $a$ and $\mathcal{T}(\mathbf{g},\mathbf{h},\mathbf{k})\in \{1,\psi\} = \Z/2$. The phase one picks up from the double braiding is given by $\frac{\theta_{a\times \mathcal{T}}}{\theta_a \theta_{\mathcal{T}}}$, and a nontrivial $\mc{T}({\bf g}, {\bf h}, {\bf k})\in \{1, \psi\}$ defines a nontrivial element in ${H}^3(B\Z/4^T, \{1, \psi\})\cong\Z/2$. Therefore, we establish that the semion-fermion theory realizes the anomaly $\nu=2\in \mho^4_{\EPin}$.
\end{proof}


\bibliographystyle{alpha}
\bibliography{ref}

\end{document}